\documentclass{amsart}
\usepackage{amssymb,mathrsfs,MnSymbol}
\usepackage{color}

\def\bC {\mathbf{C}}

\def\bN {\mathbf{N}}

\def\bR {\mathbf{R}}

\def\fa {\mathfrak{a}}

\def\fb {\mathfrak{b}}

\def\fH {\mathfrak{H}}
\def\fK {\mathfrak{K}}
\def\fP {\mathfrak{P}}
\def\fp {\mathfrak{p}}
\def\fQ {\mathfrak{Q}}
\def\fq {\mathfrak{q}}

\def\fv {\mathfrak{v}}
\def\fw {\mathfrak{w}}
\def\fZ {\mathfrak{Z}}

\def\cA {\mathcal{A}}
\def\cB {\mathcal{B}}
\def\cC {\mathcal{C}}
\def\cD {\mathcal{D}}

\def\cF {\mathcal{F}}

\def\cH {\mathcal{H}}
\def\cJ {\mathcal{J}}
\def\cK {\mathcal{K}}
\def\cL {\mathcal{L}}

\def\cP {\mathcal{P}}

\def\cS {\mathcal{S}}

\def\cV {\mathcal{V}}

\def\scrD{\mathscr{D}}
\def\scrH{\mathscr{H}}

\def\a {{\alpha}}
\def\b {{\beta}}
\def\g {{\gamma}}

\def\de {{\delta}}
\def\eps {{\epsilon}}

\def\ll {{\lambda}}
\def\L {{\Lambda}}

\def\Si {{\Sigma}}

\def\om {{\omega}}

\def\d {{\partial}}
\def\grad {{\nabla}}
\def\Dlt {{\Delta}}

\def\rstr {{\big |}}

\def\la {\langle}
\def\ra {\rangle}
\def \La {\bigg\langle}
\def \Ra {\bigg\rangle}

\def\bu {{\bullet}}

\newcommand{\Dom}{\operatorname{Dom}}
\newcommand{\fDom}{\operatorname{Form-Dom}}
\newcommand{\Span}{\operatorname{span}}
\newcommand{\Supp}{\operatorname{supp}}
\newcommand{\Graph}{\operatorname{graph}}
\newcommand{\Det}{\operatorname{det}}
\newcommand{\Tr}{\operatorname{trace}}

\newcommand{\Ker}{\operatorname{Ker}}

\newcommand{\Rank}{\operatorname{rank}}
\newcommand{\Ran}{\operatorname{ran}}
\newcommand{\MKd}{\operatorname{dist_{MK,2}}}
\newcommand{\Op}{\operatorname{OP}}

\def\hb {{\hbar}}
\def\wtilde {\widetilde}

\newcommand{\ba}{\begin{aligned}}
\newcommand{\ea}{\end{aligned}}
\newcommand{\be}{\begin{equation}}
\newcommand{\ee}{\end{equation}}
\newcommand{\lb}{\label}

\newtheorem{Thm}{Theorem}[section]

\newtheorem{Prop}[Thm]{Proposition}
\newtheorem{Cor}[Thm]{Corollary}
\newtheorem{Lem}[Thm]{Lemma}
\newtheorem{Def}[Thm]{Definition}

\newcommand{\vilA}{\mathcal A}
\newcommand{\vilB}{\mathcal B}

\newcommand{\valphi}{a}
\newcommand{\valPhi}{\mathcal A}

\newcommand{\nablaq}{\nabla^{\text{Q}}}

\newcommand{\tcr}{}
\newcommand{\tcb}{}
\newcommand{\tce}{}
\newcommand{\tcee}{}
\newcommand{\fpfq}{\fP\otimes\fQ}
\newcommand{\zfr}{\cF}

\newcommand{\ket}[1]{| #1\rangle}
\newcommand{\bra}[1]{\langle #1|}

\newcommand{\ketbra}[2]{\ket{#1}\bra{#2}}

\begin{document}

\title[Quantum Optimal Transport]{Towards Optimal Transport\\ for Quantum Densities}

\author[E. Caglioti]{Emanuele Caglioti}
\address[E.C.]{Dipartimento di Matematica, Sapienza Universit\`a di Roma, P.le A. Moro 5, 00185 Roma, Italy}
\email{caglioti@mat.uniroma1.it}

\author[F. Golse]{Fran\c cois Golse}
\address[F.G.]{Centre de Math\'ematiques Laurent Schwartz, \'Ecole polytechnique, route de Saclay, 91128 Palaiseau Cedex, France}
\email{francois.golse@polytechnique.edu}

\author[T. Paul]{Thierry Paul}
\address[T.P.]{Laboratoire Jacques-Louis Lions, Sorbonne Universit\'es \& CNRS, bo\^\i te courrier 187, 75252 Paris Cedex 05, France}
\email{paul@ljll.math.upmc.fr}

\begin{abstract}
An analogue of the quadratic Wasserstein (or Monge-Kantorovich) distance between Borel probability measures on $\bR^d$ has been defined in [F. Golse, C. Mouhot, T. Paul: Commun. Math. Phys. 343 (2015), 165--205] for density operators 
on $L^2(\bR^d)$, and used to estimate the convergence rate of various asymptotic theories in the context of quantum mechanics. The present work proves a Kantorovich type duality theorem for this quantum variant of the Monge-Kantorovich 
or Wasserstein distance, and discusses the structure of optimal quantum couplings. {\tcr Specifically, we prove that, 
under some boundedness and constraint hypothesis on the Kantorovich potentials, optimal quantum couplings involve a gradient type structure similar in the quantum paradigm to the Brenier transport map.
On the contrary, when the two quantum densities have finite rank, the structure involved by the optimal coupling has, in general, no classical counterpart.\tce}
\end{abstract}




\maketitle





\section{Introduction}\lb{S-Intro}


Let $\mu,\nu\in\cP(\bR^d)$ (the set of Borel probability measures on $\bR^d$). Given a l.s.c. function $C:\,\bR^d\times\bR^d\to[0,+\infty]$, the Monge problem in optimal transport is to minimize the functional
$$
I_C[T]=\int_{\bR^d}C(x,T(x))\mu(dx)\in[0,+\infty]
$$
over the set of Borel maps $T:\,\bR^d\to\bR^d$ such that $\nu=T\#\mu$ (the push-forward measure of $\mu$ by $T$). Here $C(x,y)$ represents the cost of transporting the point $x$ to the point $y$, so that $I_C[T]$ represents the total cost of 
transporting the probability $\mu$ on $\nu$ by the map $T$. An optimal transport map $T$ may fail to exist in full generality, so that one considers instead the following relaxed variant of the Monge problem, known as the Kantorovich 
problem:
$$
\inf_{\pi\in\Pi(\mu,\nu)}\iint_{\bR^d\times\bR^d}C(x,y)\pi(dxdy)\,.
$$
Here, $\Pi(\mu,\nu)$ is the set of couplings of $\mu$ and $\nu$, i.e. the set of Borel probability measures on $\bR^d\times\bR^d$ such that
$$
\iint_{\bR^d\times\bR^d}(\phi(x)+\psi(y))\pi(dxdy)=\int_{\bR^d}\phi(x)\mu(dx)+\int_{\bR^d}\psi(x)\nu(dx)
$$
for all $\phi,\psi\in C_b(\bR^d)$ (where $C_b(\bR^d)$ designates the set of bounded and continuous real-valued functions defined on $\bR^d$). An optimal coupling $\pi_{opt}$ always exists, so that the inf is always attained in the Kantorovich 
problem (see Theorem 1.3 in \cite{VillaniAMS} or Theorem 4.1 in \cite{VillaniTOT}). Of course, if an optimal map $T$ exists for the Monge problem, the push-forward of the measure $\mu$ by the map $x\mapsto(x,T(x))$, which can be (informally) 
written as
\be\lb{TranspCoupl}
\pi(dxdy):=\mu(dx)\de_{T(x)}(dy)
\ee
is an optimal coupling for the Kantorovich problem.

In the special case where $C(x,y)=|x-y|^2$ (the square Euclidean distance between $x$ and $y$)
$$
\MKd(\mu,\nu):=\inf_{\pi\in\Pi(\mu,\nu)}\sqrt{\iint_{\bR^d\times\bR^d}|x-y|^2\pi(dxdy)}
$$
is a distance on
$$
\cP_2(\bR^d):=\left\{\mu\in\cP(\bR^d)\text{ s.t. }\int_{\bR^d}|x|^2\mu(dx)<\infty\right\}\,,
$$
referred to as the Monge-Kantorovich, or the Wasserstein distance of exponent $2$ (see chapter 7 in \cite{VillaniAMS}, or chapter 6 in \cite{VillaniTOT}, or chapter 7 in \cite{AmbrosioGS}). In that case, there is ``almost'' an optimal transport
map, in the following sense: $\pi\in\Pi(\mu,\nu)$ is an optimal coupling for the Kantorovich problem if and only if there exists a proper\footnote{I.e. not identically equal to $+\infty$.} convex l.s.c. function $\valphi:\,\bR^d\to\bR\cup\{+\infty\}$ such 
that
$$
\Supp(\pi)\subset\Graph(\d\valphi)
$$
(where $\d\valphi$ denotes the subdifferential of $\valphi$). This is the Knott-Smith optimality criterion \cite{KnottSmith} (Theorem 2.12 (i) in \cite{VillaniAMS}). If $\mu$ satisfies the condition
\be\lb{SmallSets}
B\text{ is Borel measurable and }\cH^{d-1}(B)<\infty\implies\mu(B)=0\,,
\ee
there exists a unique optimal coupling $\pi$ of the form \eqref{TranspCoupl} for the Kantorovich problem, with $T=\grad\valphi$, where $\valphi$ is a proper convex l.s.c. function\footnote{In particular $\grad\valphi$ is defined a.e. on $\bR^d$.} on 
$\bR^d$. (In condition \eqref{SmallSets}, the notation $\cH^{d-1}(B)$ designates the $d-1$-dimensional measure of $B$.) This is the Brenier optimal transport theorem \cite{Brenier} (stated as Theorem 2.12 (ii) in \cite{VillaniAMS}). It allows recasting 
\eqref{TranspCoupl} as 
\be\label{TranspCoupl2}
(y-\grad\valphi(x))\pi(dxdy)=0,
\ee
and the a.e. defined map $\grad\valphi$ is referred to as the ``Brenier optimal transport map''.

Integrating $\pi$ against a test function depending only on $x$ shows that $\grad\valphi$ transports the $x$-marginal $\mu$ of $\pi$ to its $y$-marginal $\nu$, i.e.
\be\label{brenier1}
\nu=\grad\valphi\#\mu\,.
\ee
(This equality can obviously be deduced from \eqref{TranspCoupl} as well.)
 
\bigskip
Recently, an analogue of the Monge-Kantorovich-Wasserstein distance $\MKd$ has been defined in \cite{FGMouPaul} on the set $\cD(\fH)$ of density operators on the Hilbert space $\fH:=L^2(\bR^d)$. (Recall that a density operator on $\fH$ 
is a linear operator $R$ on $\fH$ such that $R=R^*\ge 0$ and $\Tr(R)=1$.) This definition is based on the following well known correspondence between classical and quantum paradigms.

\noindent
(a) Bounded continuous functions $f\equiv f(q,p)$ on the phase space $\bR^d_q\times\bR^d_p$ should be put in correspondence with bounded operators on the Hilbert space $\fH=L^2(\bR^d_q)$ of square-integrable functions defined on the configuration space $\bR^d_q$.

\noindent
(b) The (Lebesgue) integral of (integrable) functions on $\bR^d_q\times\bR^d_q$ should be replaced by the trace of (trace-class) operators on $\fH$.

\noindent
(c) The coordinates $q_j$ (for $j=1,\ldots,d$) on the null section of the phase space $\bR^d_q\times\bR^d_p$ should be put in correspondence  with the (unbounded) self-adjoint operators $Q_j$ on $\fH$ defined by
$$
\Dom(Q_j):=\left\{\psi\in\fH\text{ s.t. }\int_{\bR^d}q_j^2|\psi(q)|^2dq<\infty\right\}\,,\quad (Q_j\psi)(q):=q_j\psi(q)
$$
for all $j=1,\ldots,d$.

\noindent
(d) The coordinates $p_j$ (for $j=1,\ldots,d$) on the fibers of the phase space $\bR^d_q\times\bR^d_p$ should be put in correspondence with the (unbounded) self-adjoint operators $P_j$ on $\fH$ defined by
$$
\Dom(P_j):=\left\{\psi\in\fH\text{ s.t. }\int_{\bR^d}|\d_{q_j}\psi(q)|^2dq<\infty\right\}\,,\quad (P_j\psi)(q):=-i\hb\d_{q_j}\psi(q)
$$
for all $j=1,\ldots,d$.

\noindent
(e) The first order differential operators $f\mapsto \{q_j,f\}$ and $f\mapsto\{p_j,f\}$, where $\{\cdot,\cdot\}$ is the Poisson bracket on $\bR^d_q\times\bR^d_p$ such that 
$$
\{p_j,p_k\}=\{q_j,q_k\}=0\,,\quad\{p_j,q_k\}=\de_{jk}\qquad\text{ for }j,k=1,\ldots,d
$$
should be replaced with the derivations on $\cL(\fH)$ (the algebra of bounded linear operators on $\fH$) defined by
$$
A\mapsto\tfrac{i}\hb[Q_j,A]\quad\text{ and }\quad A\mapsto\tfrac{i}\hb[P_j,A]
$$
for $j=1,\ldots,d$.

\smallskip
Following these principles, the quadratic transport cost from $(x,\xi)$ to $(y,\eta)$ in $\bR^d\times\bR^d$ should be replaced with the differential operator on $\bR^d_x\times\bR^d_y$
\be\lb{DefC}
C:=\sum_{j=1}^d((x_j-y_j)^2-\hb^2(\d_{x_j}-\d_{y_j})^2)\,.
\ee
Henceforth we denote by $H$ the Hamiltonian 
\be\lb{DefH}
H:=\sum_{j=1}^d(Q_j^2+P_j^2)=|x|^2-\hb^2\Dlt_x
\ee
of the quantum harmonic oscillator. Given $R,S\in\cD_2(\fH)$, the set of density operators $\rho$ on $\fH$ such that $\Tr(\rho^{1/2}H\rho^{1/2})<\infty$, the quantum analogue of the Monge-Kantorovich-Wasserstein distance $\MKd$ is defined 
by the quantum Kantorovich problem (see Definition 2.2 in \cite{FGMouPaul})
\be\lb{DefMKhb}
MK_\hb(R,S):=\inf_{F\in\cC(R,S)}\sqrt{\Tr_{\fH\otimes\fH}(F^{1/2}CF^{1/2})}\,,
\ee
where $\cC(R,S)$ is the set of quantum couplings of $R$ and $S$, i.e.
\be\label{defsetofcoup}
\cC(R,S):=\{F\in\cD(\fH\otimes\fH)\text{ s.t. }\Tr_{\fH\otimes\fH}((A\otimes I+I\otimes B)F)=\Tr_\fH(AR+BS)\}\,.
\ee
(See Definition 2.1 in \cite{FGMouPaul}.) The functional $MK_\hb$ is a particularly convenient tool to obtain a convergence rate for the mean-field limit in quantum mechanics that is uniform in the Planck constant $\hb$ (see Theorem 2.4 in 
\cite{FGMouPaul}, and Theorem 3.1 in \cite{FGPaulPulvi} for precise statements of these results).

The striking analogy between the Wasserstein distance $\MKd$ and the quantum functional $MK_\hb$ suggests the following questions concerning a possible Brenier type theorem in quantum mechanics,  motivated in a heuristic way 
by the following considerations.

As mentioned before the classical underlying paradigm for quantum mechanics is the classical phase space $\bR^{2d}=T^*\bR^d$ equipped with the standard symplectic structure leading to the Poisson bracket defined in item $(e)$ above. 
Therefore, in this setting and under assumption \eqref{SmallSets}, equation \eqref{TranspCoupl2} reads
\be\label {TranspCoupl3}
(z'-\grad\valphi(z))\pi(dz,dz')=0\,,
\ee
where $z:=(q,p)$ and $z':=(q',p')$ are the coordinates on the phase space $T^*\bR^d$ and $dz:=dqdp,\ dz'=dq'dp'$.
 
Defining the mapping $J:T^*\bR^d\to T^*\bR^d$ entering the definition of the symplectic form $\sigma$ of $T^*\bR^d$ as $\sigma(dz,dz')=dz\wedge dJz'$ --- in the $z=(q,p)$ coordinates 
$$
J=\begin{pmatrix}
0&I_{\bR^d}\\-I_{\bR^d}&0
\end{pmatrix}
$$
--- equation \eqref{TranspCoupl3} can be put in the form 
\be\label{TranspCoupl4}
(z'-\{Jz,\valphi(z)\})\pi(dzdz')=0\,.
\ee
This symplectic formulation of the Brenier theorem is more likely to have an analogue in quantum mechanics. Indeed, according to the items (c), (d) and (e) above, the factor $(z'-\{Jz,\valphi(z)\})$ 
should be put in correspondence with the (vector-valued) operator on  $\fH\otimes\fH$
\be\label{heur1}
I_\fH\otimes Z-\frac1{i\hbar}[JZ,\valPhi]\otimes  I_\fH=I_\fH\otimes Z-\nablaq\valPhi\otimes  I_\fH\,,
\ee
for some operator $\valPhi$ on $\fH$. In \eqref{heur1}, $Z$ designates the vector of operator-valued coordinates $(Q_1,\dots,Q_d,P_1,\dots,P_d)$, and we use the notation $\nablaq:=\frac1{i\hbar}[JZ,\cdot]$.

Having in mind that the optimal classical coupling $\pi$ should be put in correspondence with an optimal element $F_{op}$  of $\cC(R,S)$ defined in \eqref{defsetofcoup}, the only ambiguity which remains in giving a quantum 
version of  \eqref{TranspCoupl4} is the choice of an ordering for the product of the operators $I_\fH\otimes Z-{\nablaq\valPhi}\otimes  I_\fH$ and $F_{op}$. 

It happens that this ambiguity will be resolved by distributing the square-root of $F_{op}$ on both sides of the expression \eqref{heur1}, which leads us to the very symmetric equality (see Theorem \ref{T-QTransp} in the next 
section):
\be\label{heur2}
F_{op}^{1/2}(I_\fH\otimes Z-\nablaq\valPhi\otimes  I_\fH)F_{op}^{1/2}=0\,.
\ee
(Notice that one cannot define a square-root of the optimal coupling in the classical case, since such a coupling is a Dirac measure, as shown by \eqref{TranspCoupl}.)

Clearly \eqref{heur2} gives a hint on the structure of optimal quantum couplings in the definition \eqref{DefMKhb} of the $MK_\hb(R,S)$ and on an analogue of the notion of Brenier optimal transport map. Notice that we are 
missing a quantum analogue of the original variational problem considered by Monge, or, equivalently, of the coupling \eqref{TranspCoupl}, so that defining a notion of quantum optimal transport seems far from obvious.

Nevertheless, \eqref{heur2} says that, once projected on the orthogonal of the kernel of an optimal coupling, the operators $I_\fH\otimes Z$ and ${\nablaq\valPhi}\otimes I_\fH$ are equal, in agreement 
with Brenier's theorem put in the form: ``the support of the optimal coupling is the graph of the gradient of a convex function".

The presence of $F_{op}^{1/2}$ on both sides of the expression between parenthesis in the left hand side of \eqref{heur2} forbids getting a quantum equivalent to \eqref{brenier1}, whose formulation is not clear anyway. 
Indeed, changes of variables in quantum mechanics are ill defined, except for linear symplectic mappings through the metaplectic representation. However, denoting by $Z'$ the (operator-valued) vector 
$$
Z':={\nablaq\valPhi}\,,
$$
(with the same operator $\valPhi$ as in \eqref{heur1}-\eqref{heur2}) and writing the trace of the left hand side of \eqref{heur2} in terms of the marginals of $F_{op}$ shows that
\be\label{heur3}
\Tr{(ZR)}=\Tr{(Z'S)}\,.
\ee
Formula \eqref{heur3} can be interpreted in the framework of the so-called Ehrenfest correspondance principle (abusively called Ehrenfest's \textit{Theorem} sometimes) \cite{Ehr,Hepp}: in quantum mechanics, $\Tr{(ZR)}$ is known as 
the expected value of the variable $Z$ in the state $R$ (in the case where $R=|\varphi\rangle\langle\varphi|$, then $\Tr{(ZR)}=(\varphi\cdot Z\varphi)_\fH$). It is the only deterministic quantity that one can associate to a particle 
in a  given state, by taking the average of the (non-deterministic) result of (an --- in principle --- infinite number of) measurements. It is interpreted in the (statistical) Ehrenfest picture as the classical value of the coordinate $Z$ of the 
state $R$. Thus the Ehrenfest interpretation of \eqref{heur3} is clear: the deterministic information we have on the state $R$ is transported to the corresponding one on the state $S$ by the change of variables $Z\mapsto Z'$.

%
%
%
\smallskip
{\tcr 
In the present article, we first state a Kantorovich duality theorem (Theorem \ref{T-QDual}) for $MK_\hb$, i.e., for every density operators $R,S$ on $\fH$,
$$
MK_\hb(R,S)^2=\sup_{\genfrac{}{}{0pt}{3}{A=A^*,\,B=B^*\in\cL(\fH)}{A\otimes I+I\otimes B\le C}}\Tr_\fH(RA+SB)\,,
$$
where $C$ is defined in \eqref{DefC}. In Theorem \ref{T-ExOptiAB}, we prove that the $\sup$ in the right hand side of the equality above is attained for some possibly unbounded operators $\fa$ and $\fb$ defined on appropriate Gelfand
triples with $\Ker(R)^\perp$ and $\Ker(S)^\perp$ as pivot spaces. Theorem \ref{T-CSOpt} provides a criterion for the $\sup$ on the right hand side of the equality above to be attained on bounded operators $A,B$ satisfying  the inequality 
constraint on the form domain of $C$. It provides also a family of density operators $R$ and $S$ for which this criterion is satisfied.

 Theorem \ref{T-QTransp} 
 is devoted to  an analogue of Brenier's theorem for quantum optimal couplings.

\noindent 
When 
 the sup in the equality above is attained by two operators $A$ and $B$ bounded on $\fH$ 
 such that the constraint $A\otimes I+I\otimes B\le C$ is satisfied on the form-domain of $C$,  we show in Theorem \ref{T-QTransp} $(1)$  our quantum result ``\`a la Brenier", namely the formula \eqref{heur2}
already mentioned
$$
F_{op}^{1/2}(I_\fH\otimes Z-\nablaq\valPhi\otimes  I_\fH)F_{op}^{1/2}=0\,
$$
 with $\valPhi=\tfrac12(H-\fa)$. Here $H$ is the harmonic oscillator defined by \eqref{DefH} and $\nablaq$ is defined by \eqref{heur1}.
 
 \noindent{\tcb  On the other hand,\tcee} when $R$ and $S$ are of finite rank, formula \eqref{heur2} has to be replaced by the following one (Theorem \ref{T-QTransp} $(2)$)
 \be\label{heur2fr}
 F_{op}^{1/2}(\zfr-\nablaq\valPhi'\otimes  I_\fH)F_{op}^{1/2}=0\,
 \ee
 with $\valPhi'=\tfrac12(H'-\fa)$. Here $H'$ is the harmonic oscillator $H$ projected on $\Ker(R)^\perp$ and $\zfr$ is the following vector  operator valued  on $\Ker(R)^\perp\otimes\Ker(S)^\perp$:
  \be\label{heur2frbis}
 \zfr_j=\sum_{k=1}^d\tfrac1{i\hbar}[(JZ^R)_j,Z_k^R]\otimes Z_k^S,\ j=1,\dots,d.
 \ee 
 where  $Z^R$ (resp. $Z^S$) is the vector $Z$ projected, component by component, on $\Ker(R)^\perp$ (resp. $\Ker(S)^\perp$)
  (see Theorem \ref{T-QTransp} $(b)$ for explicit  expressions).
  
There is no chance that  the term $\tfrac1{i\hbar}[(JZ^R)_j,Z_k^R]$ in \eqref{heur2frbis} reduces to $\delta_{i.k}I$, leading to $\cF_j=I\otimes Z_k$ so that \eqref{heur2fr} would reduce to \eqref{heur2}. Indeed, it is well known that there is no representation of the canonical relations in finite dimension. But at the contrary, nothing prevents  $F_{op}^{1/2}\zfr_k F_{op}^{1/2}$ to be equal to  (a multiple of) $F_{op}^{1/2}(I\otimes Z^R_k) F_{op}^{1/2}$. We will show, Lemma \ref{lemlemlem} in Section \ref{bipart}, that this is indeed the case for the quantum bipartite matching problem for two one-dimensional particles with equal masses, studied extensively in \cite{CagliotiFGPaul}.

  Natural examples of  classical analogues to the finite rank (independent of the Planck constant) quantum situation are the cases where $\mu,\nu$ are singular. Therefore these cases are not covered by the Knott-Smith-Brenier result. This is the case for the bipartite problem we just mentioned for which 
  $\mu=\tfrac{1+\eta}2\delta_a+\tfrac{1-\eta}2\delta_{-a},\ \nu =\tfrac12\delta_b+\tfrac12\delta_{-b} ,\ -1<\eta<1$ in the classical situation and $R=
  \tfrac{1+\eta}2|a\rangle\langle a|+\tfrac{1-\eta}2|-a\rangle\langle- a|,\ S =\tfrac12|b\rangle\langle b|+\tfrac12|-b\rangle\langle -b|$ in the quantum one. 
  When $\eta=0$, $\mu$ is optimally transported to $\nu$ by any flow which send $\pm a$ to $\pm b$, and in this case \eqref{heur2fr} takes the form of \eqref{heur2}, see Proposition \ref{oufder}. But when $\eta>0$, the mass of $a$ has to be splited in two parts, an amount $\tfrac12$ to be send to $b$ and an amount $\tfrac\eta2$ which goes to $-b$, and $\mu$ is optimaly transported to $\nu$ by a multivalued map.
 
 Therefore, beside the fact that formula \eqref{heur2} represents a quantum analogue of the   Knott-Smith-Brenier result,  formulas \eqref{heur2fr}-\eqref{heur2frbis} have in general no analogue in terms of classical (monovalued) flow.
 \vskip 1cm
 

The main results, Theorems \ref{T-QDual}. \ref{T-ExOptiAB}, \ref{T-CSOpt} and \ref{T-QTransp}, are stated in Section \ref{S-results} and proved in Sections \ref{S-proof1}, \ref{S-proof2}, \ref{S-proof3} and \ref{S-proof4} respectively. Section \ref{S-examples} is devoted to some examples, including the finite rank and T\"oplitz situations, and the three Appendices contain some technical material, including a result on monotone convergence for trace-class operators in Apendix \ref{monconv}.

\smallskip
\tce}
To conclude this introduction, we mention other attempts at defining analogues of the Wasserstein, or Monge-Kantorovich distances in the quantum setting. For instance \.Zyczkowski and S\l omczy\'nski \cite{ZyckSlom} (see also section 7.7
in chapter 7 of \cite{BengtssonZyczkowski}) proposed to consider the original Monge distance (also called the Kantorovich-Rubinstein distance, or the Wasserstein distance of exponent $1$) between the Husimi transforms of the density 
operator (see \eqref{Husimi} for a definition of this transform).

Besides the quantity $MK_\hb$ appeared in \cite{FGMouPaul}, other analogues of the Wasserstein distance of exponent $2$ for quantum densities have been proposed by several other authors. For instance Carlen and Maas have 
defined a quantum analogue of the Benamou-Brenier formula (see \cite{BenaBrenier} or Theorem 8.1 in chapter 8 of \cite{VillaniTOT}) for the classical Wasserstein distance of exponent $2$, and their idea has been used to obtain 
a quantum equivalent of the so-called HWI inequality: see \cite{CarlMaas1,CarlMaas2,RouzDatt}. 

Other propositions for generalizing Wasserstein distances to the quantum setting have emerged more recently, such as \cite{Ikeda} (mainly focussed on pure states) or \cite{PalmTrev}, very close to our definition of $MK_\hb$, except 
that the set of couplings used in the minimization is different and based instead on the notion of ``quantum channels'' (see also \cite{PalmTrev2} for a definition of a quantum Wasserstein distance of order $1$).


\section{Main Results}\label{S-results}


The key argument in deriving the structure \eqref{TranspCoupl} of optimal couplings for the Kantorovich problem involves a min-max type result known as ``Kantorovich duality''. For each $\mu,\nu\in\cP_2(\bR^d)$, one has
\be\label{kanto1}
\MKd(\mu,\nu)^2=\sup_{\substack{\phi,\psi\in C_b(\bR^d)\\ \phi(x)+\psi(y)\leq|x-y|^2}}\left(\int_{\bR^d}\phi(x)\mu(dx)+\int_{\bR^d}\psi(y)\nu(dy)\right)\,.
\ee
When $\mu,\nu$ do not charge small sets, in the sense that they satisfy \eqref{SmallSets}, one can prove that the supremum in the r.h.s. of \eqref{kanto1} is actually attained and
\begin{eqnarray}
\MKd(\mu,\nu)^2&=&\min_{\pi\in\Pi(\mu,\nu)}\iint_{\bR^d\times\bR^d}|x-y|^2\pi(dxdy)\nonumber\\
&=&\iint_{\bR^d\times\bR^d}|x-y|^2\pi_{op}(dxdy)\nonumber\\
&=&\max_{\substack{\phi\in L^1(\mu),\,\psi\in L^1(\nu)\\ \phi(x)+\psi(y)\leq|x-y|^2\\ \mu\otimes\nu\text{\scriptsize-a.e.}}}\left(\int_{\bR^d}\phi(x)\mu(dx)+\int_{\bR^d}\psi(y)\nu(dy)\right)\nonumber\\
&=&
\int_{\bR^d}\phi_{op}(x)\mu(dx)+\int_{\bR^d}\psi_{op}(y)\nu(dy)\nonumber
\end{eqnarray}
for two proper convex l.s.c. functions $\phi_{op}$ and $\psi_{op}$ on 
$\bR^d$. 

Moreover, $\valphi(x):=\tfrac12(x^2-\phi_{op}(x))$ is precisely the function appearing  in \eqref{TranspCoupl2},  the gradient of which defines a.e. the Brenier optimal transport map of the previous section. (See Theorem 1.3, 
Proposition 2.1 and Theorem 2.9 in \cite{VillaniAMS}.) 

Likewise, the operator $\valPhi$ in \eqref{heur2} will be similarly related to an optimal operator appearing in a dual formulation of definition \eqref{DefMKhb}, to be presented below.

\bigskip
Before we state the quantum analogue of the Kantorovich duality, we need some technical preliminaries.

The quantum transport cost is the operator
\be\label{defC}
C:=\sum_{j=1}^d((x_j-y_j)-\hbar^2(\d_{x_j}-\d_{y_j})^2)\,,
\ee
viewed as an unbounded self-adjoint operator on $L^2(\bR^d\times\bR^d)$ with domain
\be\lb{DomC}
\Dom(C):=\{\psi\in L^2(\bR^d\times\bR^d)\text{ s.t. }|x-y|^2\psi\text{ and }|D_x-D_y|^2\psi\in L^2(\bR^d\times\bR^d)\}\,.
\ee
Henceforth we denote by $H$ the Hamiltonian of the quantum harmonic oscillator, i.e.
\be\lb{HarmOsc}
H:=|x|^2-\hbar^2\Dlt_x\,,
\ee
which is a self-adjoint operator on $L^2(\bR^d)$ with domain\footnote{If $u\in C_c^\infty(\bR^d)$, one has
$$
\ba
\int_{\bR^d}|x|^2u(x)Hu(x)dx=&\int_{\bR^d}(|x|^4u(x)^2+\hb^2|x|^2|\grad u(x)|^2)dx+\hb^2\int_{\bR^d}x\cdot\grad\left(u(x)^2\right)dx
\\
=&\int_{\bR^d}((|x|^4-d\hb^2)u(x)^2+\hb^2|x|^2|\grad u(x)|^2)dx\,,
\ea
$$
so that $\||x|^2u\|^2_{L^2(\bR^d)}\le\|Hu\|_{L^2(\bR^d)}\||x|^2u\|_{L^2(\bR^d)}+d\hb^2\|u\|_{L^2(\bR^d)}$. Thus, if $u\in L^2(\bR^d)$ and $Hu\in L^2(\bR^d)$, then one has $|x|^2u\in L^2(\bR^d)$, which implies in turn that $-\Dlt u\in L^2(\bR^d)$.
The same argument shows that $U\in L^2(\bR^d\times\bR^d)$ and $CU\in L^2(\bR^d\times\bR^d)$ imply that $|x-y|^2U\in L^2(\bR^d\times\bR^d)$. These observations imply that the domains of $H$ and $C$ are the spaces given above.}
\be\lb{DomH}
\Dom(H):=\{\phi\in H^2(\bR^d)\text{ s.t. }|x|^2\phi\in L^2(\bR^d)\}\,.
\ee
In the sequel, we shall also need the form-domains of the operators $H$ and $C$:
\be\lb{fDom}
\ba
\fDom(H):=&\{\phi\in H^1(\bR^d)\text{ s.t. }|x|\phi\in L^2(\bR^d)\}\,,
\\
\fDom(C):=&\{\psi\in\fH\!\otimes\!\fH\text{ s.t. }(x_j\!\!-\!y_j)\psi\text{ and }(\d_{x_j}\!\!\!-\!\d_{y_j})\psi\in\fH\!\otimes\!\fH\,,\,\,1\le j\le d\}\,.
\ea
\ee
The definition of the form-domain of a self-adjoint operator can be found for instance in section VIII.6, Example 2 of \cite{ReedSim1}. Observe that
\be\lb{fDomIncl}
\ba
\fDom(H\otimes I\!+\!I\otimes H)&=\{\psi\in H^1(\bR^d\!\!\times\!\bR^d)\text{ s.t. }(|x|\!+\!|y|)\psi\in L^2(\bR^d\!\!\times\!\bR^d)\}
\\
&\subset\fDom(C)\,.
\ea
\ee

\begin{Lem}\lb{L-QfDom}
Let $R,S\in\cD_2(\fH)$, and let $Q\in\cC(R,S)$. Each eigenfunction $\Phi$ of $Q$ such that $Q\Phi\not=0$ belongs to $\fDom(H\otimes I+\otimes H)$ and
$$
0\le\la\Phi|H\otimes I+I\otimes H|\Phi\ra\le\Tr_\fH(R^{1/2}HR^{1/2}+S^{1/2}HS^{1/2})<\infty\,.
$$
In particular $\Phi\in\Dom(C)$ with 
\be\lb{BdPhiDom}
\la\Phi|C|\Phi\ra\le 2\Tr_\fH(R^{1/2}HR^{1/2}+S^{1/2}HS^{1/2})<\infty\,.
\ee
\end{Lem}

\begin{proof}
Since $R,S\in\cD_2(\fH)$, one has
$$
\ba
\Tr_{\fH\otimes\fH}(Q^{1/2}(H\otimes I)Q^{1/2})=\Tr_\fH(R^{1/2}HR^{1/2})<\infty
\\
\Tr_{\fH\otimes\fH}(Q^{1/2}(I\otimes H)Q^{1/2})=\Tr_\fH(S^{1/2}HS^{1/2})<\infty
\ea
$$
for each $Q\in\cC(R,S)$ by Lemma \ref{L-UBdCyclicity}. In particular
$$
\Tr_{\fH\otimes\fH}(Q^{1/2}(H\otimes I+I\otimes H)Q^{1/2})<\infty\,.
$$
Let $(\Phi_k)_{k\ge 0}$ be a complete orthonormal system of eigenvectors of $Q$, and let $(\ll_k)_{k\ge 0}$ be the sequence of eigenvalues of $Q$ such that $Q\Phi_k=\ll_k\Phi_k$ for each $k\ge 0$. Thus
$$
\ba
\ll_k>0\implies\Phi_k\in\fDom(H\otimes I+I\otimes H)
\\
\text{ and }\sum_{k\ge 0}\ll_k\la\Phi_k|H\otimes I+I\otimes H|\Phi_k\ra=\Tr_{\fH\otimes\fH}(Q^{1/2}(H\otimes I+I\otimes H)Q^{1/2})
\\
=\Tr_\fH(R^{1/2}HR^{1/2}\!\!+\!S^{1/2}HS^{1/2})<\infty&\,,
\ea
$$
and this implies the desired inequality. Using \eqref{fDomIncl} shows that
$$
\ba
\Psi\in\fDom(H\otimes I+I\otimes H)\implies&\Psi\in\fDom(C)\text{ and }
\\
&0\le\la\Psi|C|\Psi\ra\le\la\Psi|H\otimes I+I\otimes H|\Psi\ra\,.
\ea
$$
\end{proof}

\subsection{A Quantum Analogue to the Kantorovich Duality}

The statement below is an analogue of the Kantorovich Duality Theorem (Theorem 1.3 in \cite{VillaniAMS}, or Theorem 6.1.1 in \cite{AmbrosioGS}) for the quantum transport cost operator $C$ defined by \eqref{defC}.

\begin{Thm}[Quantum duality]\lb{T-QDual}
Let $R,S\in\cD_2(\fH)$. Then
\be\lb{QDual}
\min_{F\in\cC(R,S)}\Tr_{\fH\otimes\fH}(F^{1/2}CF^{1/2})=\sup_{(A,B)\in\fK}\Tr_\fH(RA+SB)\,,
\ee
where
$$
\fK:=\{(A,B)\in\cL(\fH)\times\cL(\fH)\text{ s.t. }A=A^*\,, B=B^*\text{ and }A\otimes I+I\otimes B\le C\}\,.
$$
\end{Thm}

In the definition of $\fK$, the inequality
$$
A\otimes I+I\otimes B\le C
$$
means that
$$
\la\psi|A\otimes I+I\otimes B|\psi\ra\le\la\psi|C|\psi\ra
$$
for all $\psi\in\fDom(C)$.

Notice that the $\inf$ on the left hand side of the duality formula is attained --- in other words, there always exists an optimal coupling $F\in\cC(R,S)$. On the contrary, the $\sup$ in the right hand side of the duality formula is in general not attained
--- at least not attained in the class $\fK$ in general.

\subsection{Existence of Optimal Operators $A,B$}

In this section, we explain how the $\sup$ in the right hand side of the duality formula is attained in a class of operators $(A,B)$ larger than $\fK$.

\subsubsection{Gelfand triple associated to a nonnegative trace-class operator}

We shall use repeatedly the following construction. Given a separable Hilbert space $\scrH$ and $T\in\cL^1(\scrH)$ such that $T=T^*\ge 0$, let $(\xi_n)_{n\ge 1}$ be a complete orthonormal basis of $\scrH$ of eigenvectors of $T$. Set 
\be\lb{DefJ0}
\cJ_0[T]:=\Span\{\xi_n\text{ s.t. }\la\xi_n|T|\xi_n\ra>0\}\,,
\ee
and
\be\lb{Def()T}
(\phi|\psi)_T:=\la\phi|T^{-1}|\psi\ra\,,\quad\phi,\psi\in\cJ_0[T]\,.
\ee
Let $\cJ[T]$ designate the completion of $\cJ_0[T]$ for the inner product $(\cdot|\cdot)_T$. Obviously
\be\lb{JKJ'}
\cJ[T]\subset\overline{\cJ_0[T]}=\Ker(T)^\perp\subset\cJ[T]'
\ee
(where $\overline{\cJ_0[T]}$ is the closure of $\cJ_0[T]$ in $\fH$). The first inclusion is continuous since, for each $\phi\in\cJ_0[T]$, one has
$$
\|\phi\|^2_\fH\le\|T\|(\phi|\phi)_T=\|T\|\la\phi|T^{-1}|\phi\ra\,.
$$
The operator $T^{-1/2}$, which is a priori defined on $\cJ_0[T]$ only, has a unique continuous extension which is the unitary transformation
\be\lb{UnitTransfo}
T^{-1/2}:\,\cJ[T]\to\Ker(T)^\perp\quad\text{ with adjoint }\quad T^{-1/2}:\,\Ker(T)^\perp\to\cJ[T]'\,.
\ee
In other words, one has a Gelfand triple
\be\lb{Gelfand}
\cJ[T]\mathop{\subset}_c\Ker(T)^\perp\subset\cJ[T]'\,.
\ee
(Notice that the embedding $\cJ[T]\subset\Ker(T)^\perp$ is compact since $T^{1/2}$ is a Hilbert-Schmidt, and therefore compact, operator on $\fH$.) With the unitary transformation \eqref{UnitTransfo}, one defines the isometric isomorphism
\be\lb{Isom2}
\cL(\cJ[T],\cJ[T]')\ni\fZ\mapsto T^{1/2}\fZ T^{1/2}=Z\in\cL(\Ker(T)^\perp)\,.
\ee
Under this isomorphism, $\fZ^*$ is obviously mapped to $Z^*$. 

\subsubsection{The optimality class $\tilde\fK(R,S)$}

While the original class $\fK$ is independent of the quantum density operators $R$ and $S$, the optimality class $\fK(R,S)$ significantly depends on $R,S$. 

\begin{Def}\label{deftildek} For each $R,S\in\cD_2(\fH)$, let $\tilde\fK(R,S)$ be the set of $(\fv,\fw)$ with $\fv\in\cL(\cJ[R],\cJ[R]')$ and $\fw\in\cL(\cJ[S],\cJ[S]')$ such that

\noindent
(a) the operators $V=R^{1/2}\fv R^{1/2}$ and $W=S^{1/2}\fw S^{1/2}$ satisfy
$$
\ba
2R^{1/2}HR^{1/2}&\ge V=V^*\in\cL^1(\Ker(R)^\perp)
\\
2R^{1/2}HR^{1/2}&\ge W\!=\!W^*\!\in\cL^1(\Ker(S)^\perp)\,;
\ea
$$
(b) for each $\Phi\in\cJ_0[R]\otimes\cJ_0[S]$, one has
$$
\la\Phi|\fv\otimes I+I\otimes\fw|\Phi\ra\le\la\Phi|C|\Phi\ra\,.
$$
\end{Def}

Notice that the left hand side of the inequality in condition (b) is well defined, since $\cJ_0[R]\subset\cJ[R]$, so that $\fv\cJ_0[R]\subset\cJ[R]'$. Hence any element of $\fv\cJ_0[R]$ is a linear functional which can be evaluated on any element
of $\cJ_0[R]\subset\cJ[R]$, and likewise any element of $\fw\cJ_0[S]$ is a linear functional which can be evaluated on any element of $\cJ_0[S]\subset\cJ[S]$. 

As for the right hand side, let $(e_j)_{j\ge 1}$ and $(f_k)_{k\ge 1}$ be complete orthonormal systems of eigenvectors of $R$ and $S$ respectively in $\fH$. By the implication in \eqref{J0infDom} (see Lemma \ref{L-Energ} in the Appendix)
$$
\ba
e_j\in\Ker(R)^\perp\implies e_j\in\fDom(H)\,,\\ f_k\in\Ker(S)^\perp\implies f_k\in\fDom(H)\,.
\ea
$$
In particular
\be\lb{J0xJ0in}
\cJ_0[R]\otimes\cJ_0[S]\subset\fDom(H\otimes I+I\otimes H)\subset\fDom(C)
\ee
so that the right hand side of the inequality in (b) is finite.

\subsubsection{The $\sup$ is attained in $\tilde\fK(R,S)$}

Passing from $\fK$ to $\tilde\fK(R,S)$ is equivalent to seeking the optimal Kantorovich potential in $L^1(\bR^d\!,\mu)$ as in Theorems 1.3 or Theorem 2.9 of \cite{VillaniAMS}, instead of $C_b(\bR^d)$ --- see the last sentence in Theorem 1.3 
of \cite{VillaniAMS}, together with Remark 1.6 in that same reference.

\begin{Thm} [Existence of optimal duality potentials]\lb{T-ExOptiAB} For all $R,S\in\cD_2(\fH)$,
$$
\min_{F\in\cC(R,S)}\Tr_{\fH\otimes\fH}(F^{1/2}CF^{1/2})=\max_{(\fa,\fb)\in\tilde\fK(R,S)}\Tr_\fH(R^{1/2}\fa R^{1/2}+S^{1/2}\fb S^{1/2})\,.
$$
If $R$ and $S$ are of finite rank, $\tilde\fK(R,S)\subset\cL(\Ker(R)^\perp)\times\cL(\Ker(S)^\perp)$, so that any optimal pair $(\fa,\fb)$ for the $\max$ in the right hand side of the equality above consists of operators $\fa$ and $\fb$ defined on 
the finite-dimensional linear spaces $\Ker(R)^\perp$ and $\Ker(S)^\perp$.
\end{Thm}

\subsection{Structure of optimal couplings}

In the classical setting, pick a proper convex l.s.c. function $\phi:\,\bR^d\mapsto\bR\cup\{+\infty\}$, and let $\mu\in\cP(\bR^d)$ satisfy condition \eqref{SmallSets} and
\be\lb{CSatur}
\int_{\bR^d}(|x|^2+|\grad\phi(x)|^2+|\phi(x)|+|\phi^*(\grad\phi(x))|)\mu(dx)<\infty\,.
\ee
Then 
\be\lb{COpti}
\pi(dxdy):=\mu(dx)\de(y-\grad\phi(x))
\ee
is an optimal coupling of the measures $\mu$ and $\nu:=\grad\phi\#\mu$ for the Kantorovich problem with the cost $C(x,y)=|x-y|^2$.

\smallskip
{\tcb 
We begin with a necessary and sufficient condition on density operators $R,S\in\cD_2(\fH)$ to have the $\sup$ in \eqref{QDual} attained in 
$\fK$, along with an optimality criterion for the couplings of such density operators.
This is the quantum analogue of the sufficient condition in Theorem 6.1.4 of \cite{AmbrosioGS}.

\tcee}
\begin{Thm}[Optimality criterion]\lb{T-CSOpt}
Let $(A,B)\in\fK$ be such that 
$$
\Ker(C-A\otimes I-I\otimes B)\not=\{0\}\,.
$$
Let $(\Phi_j)$ be a complete orthonormal system in $\Ker(C-A\otimes I-I\otimes B)$, and let 
\be\lb{FormOptiCoupl}
F:=\sum_j\ll_j|\Phi_j\ra\la\Phi_j|\,,\quad\text{ with }\quad\ll_j\ge 0\quad\text{ and }\quad\sum_j\ll_j=1\,.
\ee
Call $F_1:=\Tr_2(F)$ and $F_2:=\Tr_1(F)$ the partial traces of $F$ on the second and first factor in $\fH\otimes\fH$ respectively. Then $F$ is an optimal coupling of $F_1$ and $F_2$:
$$
\ba
\Tr_{\fH\otimes\fH}(F^{1/2}CF^{1/2})=&\min_{Q\in\cC(F_1,F_2)}\Tr_{\fH\otimes\fH}(Q^{1/2}CQ^{1/2})
\\
=&\sup_{(a,b)\in\fK}\Tr_\fH(F_1a+F_2b)=\Tr_\fH(F_1A+F_2B)\,.
\ea
$$
Conversely, if $(A,B)\in\fK$ is an optimal pair for $R,S\in\cD_2(\fH)$, i.e. if
\be\lb{MK=Dual}
MK_\hb(R,S)^2=\Tr_\fH(RA+SB)\,,
\ee
then $\Ker(C-A\otimes I-I\otimes B)\not=\{0\}$ and any optimal coupling of $R$ and $S$, i.e. any $F\in\cC(R,S)$ such that $MK_\hb(R,S)^2=\Tr_{\fH\otimes\fH}(F^{1/2}CF^{1/2})$ is of the form \eqref{FormOptiCoupl}.

\end{Thm}
{\tcb 
In the previous theorem (Theorem \ref{T-CSOpt}), we have obtained a complete description of the densities $R,S\in\cD_2(\fH)$ such that the $\sup$ in \eqref{QDual} is attained in $\fK$. Next, we give necessary conditions on the structure 
of the optimal couplings $F\in\cC(R,S)$ for such density operators $R$ and $S$.

\tcee}
\vskip 0.5cm
\smallskip
In the classical setting, the structure \eqref{COpti} of optimal couplings is a straightforward consequence of \eqref{CSatur}. Indeed, the set of points where the Young inequality
$$
\phi(x)+\phi^*(y)\ge x\cdot y
$$
becomes an equality is included in $\Graph(\d\phi)$. This suggests the idea of looking for a quantum analogue of the Brenier optimal transport map in the optimality criterion in Theorem \ref{T-CSOpt}.

\smallskip

We shall need the following basic functional analytic considerations. The linear space $\Dom(C)$ endowed with the inner product
$$
(\Phi,\Psi)\mapsto(\Phi|\Psi)_{\Dom(C)}=(\Phi|\Psi)_{\fH\otimes\fH}+(C\Phi|C\Psi)_{\fH\otimes\fH}
$$
is a Hilbert space. Hence $C\in\cL(\Dom(C),\fH\otimes\fH)$ (with norm at most $1$). Since $C$ is symmetric on $\Dom(C)$, it has a unique extension as an element of $\cL(\fH,\Dom(C)')$ (where $\Dom(C)'$ designates 
the topological dual of $\Dom(C)$), which is defined by the formula
\be\lb{ExtC}
\la C\Phi,\Psi\ra_{\Dom(C)',\Dom(C)}:=(\Phi|C\Psi)_{\fH\otimes\fH}
\ee
for all $\Phi\in\fH\otimes\fH$ and $\Psi\in\Dom(C)$. 

On the other hand, the linear space $\fDom(H\otimes I+I\otimes H)$ endowed with the inner product
$$
(\Phi,\Psi)\mapsto(\Phi|\Psi)_{\fDom(H\otimes I+I\otimes H)}=(\Phi|\Psi)_{\fH\otimes\fH}+\la\Phi|H\otimes I+I\otimes H|\Psi\ra
$$
is a Hilbert space. If $T\in\cL(\fDom(H\otimes I+I\otimes H),\fH\otimes\fH)$ is a symmetric operator, it has a unique extension as an element of $\cL(\fH,\fDom(H\otimes I+I\otimes H)')$. This extension is defined by the formula
\be\lb{ExtT}
\la T\Phi,\Psi\ra_{\fDom(H\otimes I+I\otimes H)',\fDom(H\otimes I+I\otimes H)}=(\Phi|T\Psi)_{\fH\otimes\fH}
\ee
(where $\fDom(H\otimes I+I\otimes H)'$ is the topological dual of $\fDom(H\otimes I+I\otimes H)$) for all $\Phi\in\fH\otimes\fH$ and $\Psi\in\fDom(H\otimes I+I\otimes H)$.

In particular 
\be\lb{[TC]}
[T,C]:\Dom(C)\cap\fDom(H\!\otimes\!I\!+\!I\!\otimes\!H)\!\to\!\Dom(C)'\!+\!\fDom(H\!\otimes\!I\!+\!I\!\otimes\!H)'
\ee
is a continuous linear map. Since 
$$
\Dom(C)'\!+\!\fDom(H\!\otimes\!I\!+\!I\!\otimes\!H)'\subset\left(\Dom(C)\cap\fDom(H\!\otimes\!I\!+\!I\!\otimes\!H)\right)'\,,
$$
the bilinear functional
$$
(\Phi,\Psi)\mapsto\la\Phi|[T,C]|\Psi\ra:=(T\Phi|C\Psi)_{\fH\otimes\fH}-(C\Phi|T\Psi)_{\fH\otimes\fH}
$$
is continuous on $\Dom(C)\cap\fDom(H\!\otimes\!I\!+\!I\!\otimes\!H)$.

Henceforth we use the following notation:
$$
q_j\psi(x_1,\ldots,x_d):=x_j\psi(x_1,\ldots,x_N)\,,\quad p_j\psi(x_1,\ldots,x_d):=-i\hb\d_{x_j}\psi(x_1,\ldots,x_d)
$$
for all $\psi\in\fDom(H)$ and all $j=1,\ldots,d$, and
$$
\scrD_{q_j}S:=\tfrac{i}{\hb}[p_j,S]\,,\qquad\scrD_{p_j}S:=-\tfrac{i}{\hb}[q_j,S]\,.
$$

\begin{Thm}\lb{T-QTransp}
Let $R,S\in\cD_2(\fH)$, let $F\in\cC(R,S)$ be an optimal coupling, i.e.
$$
\Tr_{\fH\otimes\fH}(F^{1/2}CF^{1/2})=\min_{Q\in\cC(R,S)}\Tr_{\fH\otimes\fH}(Q^{1/2}CQ^{1/2})\,.
$$
{\tcr and $(A,B)\in\tilde\fK(R,S)$ a pair of optimal operators such that
$$
\Tr_\fH(RA+SB)=\sup_{(a,b)\in\fK}\Tr_\fH(Ra+Sb)\,,
$$

Then,\tce}
\begin{enumerate}
\item {\tcr if 
$A\in\cL(\cJ(R),\fH)$ (resp. $B\in\cL(\cJ(S),\fH)$), let us denote by the same letters  $ A, B$ two extensions of $A,B$ to $\cL(\fH)$ such that 
$$
 A\otimes I+I\otimes  B\leq C
$$
on $\fDom(C)$ (in other words, $( A, B)\in\fK$ defined in Theorem \ref{T-QDual})\footnote{Note that even when $\Ker(R)=\Ker(S)=\{0\}$, so that $\cJ_0(E)\otimes\cJ_0(S)$ is dense in $\fH\otimes\fH$, Theorem \ref{T-ExOptiAB} provides optimal operators $A,B$ satisfying the constraint inequality only on $\cJ_0(E)\otimes\cJ_0(S)$  and not automatically on $\fDom(C)$. This is why the preceding constraint inequality has to be supposed to hold true.}.

\tce}  
Then

\smallskip
\noindent
(a) any eigenvector $\Phi$ of $F$ such that $F\Phi\not=0$ satisfies 
$$
\Phi\in\Dom(C)\text{ and }C\Phi=(A\otimes I+I\otimes B)\Phi\,;
$$
(b)  Let us denote
$$
\vilA:=\tfrac12(H-A)\text{ and } \vilB:=\tfrac12(H-B)\,,
$$
where $H$ is the harmonic oscillator in \eqref{HarmOsc}. 

\noindent Then, for each $j=1,\ldots,d$, one has
$$
\ba
F^{1/2}(I\otimes q_j-\scrD_{q_j}\vilA\otimes I)F^{1/2}=F^{1/2}(I\otimes p_j-\scrD_{p_j}\vilA\otimes I)F^{1/2}=0\,,
\\
F^{1/2}(q_j\otimes I-I\otimes\scrD_{q_j}\vilB)F^{1/2}=F^{1/2}(p_j\otimes I-I\otimes\scrD_{p_j}\vilB))F^{1/2}=0\,.
\ea
$$
\item{\tcr  if $R$ and $S$ have finite rank, one knows  by Theorem \ref{T-ExOptiAB} that $A\in\cL(\Ker(R)^\perp)$ and $B\in\cL(\Ker(S)^\perp)$. Let us denote by $\fP$ and $\fQ$ be the orthogonal projections on $\Ker(R)^\perp$ and $\Ker(S)^\perp$ respectively.

 Then,  by identifying $F$ with its projection on $\Ker(R)^\perp\otimes\Ker(S)^\perp$ thanks to Lemma \ref{L-SuppCoupling}, one has (on $\Ker(R)^\perp\otimes\Ker(S)^\perp$),

\noindent
(c) 
$ 
(\fP\otimes\fQ C\fP\otimes\fQ 
-A\otimes I_{\ker(S)^\perp}+I_{\ker(R)^\perp}\otimes B)F=0$.

\noindent (d)  Let us denote
$$
\vilA':=\tfrac12(\fP H\fP-A)\text{ and } \vilB':=\tfrac12(\fQ H\fQ-B)\,,
$$
and, for each $j=1,\ldots,d$,  $Q_j^R=\fP Q_j\fP,P_j^R=\fP P_j\fP,Q_j^S=\fQ Q_j\fQ,P_j^S=~\fQ P_j\fQ$. One has
$$
\ba
F^{1/2}(\sum_{k=1}^d(\tfrac i\hb[P_j^R,Q_k^R]\otimes Q_k^S
+
\tfrac i\hb[P_j^R,P_k^R]\otimes P_k^S
)-\scrD_{Q_j^R}\vilA'\otimes I)F^{1/2}=0\\
F^{1/2}(\sum_{k=1}^d(\tfrac i\hb[Q_j^R,Q_k^R]\otimes Q_k^S
+
\tfrac i\hb[Q_j^R,P_k^R]\otimes P_k^S
)-\scrD_{P_j}\vilA'\otimes I)F^{1/2}=0\\
F^{1/2}(\sum_{k=1}^d( Q_k^R\otimes \tfrac i\hb[P_j^S,Q_k^S]
+
 P_k^R\otimes \tfrac i\hb[P_j^s,P_k^s]
)-I\otimes\scrD_{Q_j^R}\vilB')F^{1/2}=0\\
F^{1/2}(\sum_{k=1}^d( Q_k^R\otimes \tfrac i\hb[Q_j^S,Q_k^S]
+
 P_k^R\otimes \tfrac i\hb[Q_j^s,P_k^s]
)-I\otimes\scrD_{P_j^S}\vilB')F^{1/2}=0\\
\ea
$$
\tce}
\end{enumerate}
\end{Thm}

As mentioned in the introduction, {\tcr the two last identities of $(b)$\tce}  are  analogous to the condition \eqref{TranspCoupl3}
$$
(z'-\grad\valphi(z))\pi(dz,dz')=0
$$
obtained in the setting of classical optimal transport in the case where the convex function $\valphi$ is smooth, so that $\d\valphi(z)=\{\grad\valphi(z)\}$ (see the Brenier or the Knott-Smith theorems, stated as Theorem 2.12 (i)-(ii) in \cite{VillaniAMS}.
Indeed, using the (vector-valued) operator 
$$
\nablaq=(\scrD_{q_1},\dots,\scrD_{q_d},\scrD_{p_1},\dots,\scrD_{p_d})
$$
together with the vector of operators $Z$ defined right after \eqref{heur1}, statement (b) of Theorem \ref{T-QTransp} reads
\begin{eqnarray}
F^{\frac12}(Z\otimes I-I\otimes\nablaq\valPhi)F^{\frac12}=0.
\end{eqnarray}
Notice that the quantum analogue of the function $\valphi$ is the operator $\tfrac12(H-A)$ (equivalently, the classical analogue of $A$ is $(|q|^2+|p|^2)-2\phi(q,p)$): see Remark 2.13 (iii) following Theorem 2.12 in \cite{VillaniAMS}, where the relation 
between $\phi$ and the optimal pair in the Kantorovich duality theorem is described in detail.

{\tcr Concerning $(d)$, it is a straightforward computation to show that the two first equalities can be synthesized as formulas  \eqref{heur2fr}-\eqref{heur2frbis} in the introduction.\tce}


\section{Proof of Theorem \ref{T-QDual}}\label{S-proof1}


Set $E:=\cL(\fH\otimes\fH)$. Define $f,g:\,E\to\bR\cup\{+\infty\}$ by the formulas
$$
f(T):=\left\{\ba &0&&\qquad\hbox{ if }T=T^*\ge -C\,,\\&+\infty&&\qquad\hbox{ otherwise,}\ea\right.
$$
and
$$
g(T):=\left\{\ba &\Tr_\fH(RA+SB)&&\qquad\hbox{ if }T=T^*=A\otimes I+I\otimes B\,,\\&+\infty&&\qquad\hbox{ otherwise.}\ea\right.
$$
For each $T=T^*\in\cL(\fH\otimes\fH)$, the constraint $T\ge -C$ in the definition of $f$ is to be understood as follows: 
$$
\la\phi|T|\phi\ra\ge-\la\phi|C|\phi\ra\quad\hbox{ for each }\phi\in\fDom(C)\,.
$$
On the other hand, the nullspace of the linear map
$$
\Gamma:\,\cL(\fH)\times\cL(\fH)\ni(A,B)\mapsto A\otimes I+I\otimes B\in\cL(\fH\otimes\fH)
$$
is $\Ker(\Gamma)=\{(tI, -tI)\text{ s.t. }t\in\bC\}$. Since $\Tr_\fH(R)=\Tr_{\fH}(S)=1$, one has 
$$
\Tr_\fH(RA+SB)=t\Tr_\fH(R-S)=0\quad\text{ for all }(A,B)=(tI,-tI)\in\Ker(\Gamma)
$$
so that
$$
A\otimes I+I\otimes B\mapsto\Tr_\fH(RA+SB)
$$
defines a unique linear functional on $\Ran(\Gamma)$. Besides 
$$
(A\otimes I+I\otimes B)^*=A^*\otimes I+I\otimes B^*\,,
$$
so that, by cyclicity of the trace,
$$
\ba
T=A\otimes I+I\otimes B\text{ and }T=T^*\implies A=A^*\text{ and }B=B^*
\\
\implies g(T)=\Tr_\fH(RA^*+SB^*)=\overline{\Tr_\fH(AR+BS)}=\overline{g(T)}&\,.
\ea
$$
Therefore, the prescription above defines indeed a unique function $g$ on $E$ with values in $(-\infty,+\infty]$.

One easily checks that $f$ and $g$ are convex.  Indeed, $f$ is the indicator function (in the sense of the definition in \S 4 of \cite{Rockafellar} on p. 28) of the convex set 
$$
\{T=T^*\in E\text{ s.t. }T\ge -C\}\,,
$$
while $g$ is the extension by $+\infty$ of a real-valued linear functional defined on the linear subspace $\Ran(\Gamma)$ of $E$. Clearly,
$$
f(0)=g(0)=0\,.
$$
Moreover $f$ is continuous at $0$. Indeed, the Heisenberg uncertainty inequality implies that 
\be\lb{HeisenIneq}
C\ge 2d\hbar I\,,
\ee
so that 
$$
T=T^*\hbox{ and }\|T\|<d\hbar\implies T\ge -2d\hbar I\ge -C\,.
$$
Hence
$$
T=T^*\hbox{ and }\|T\|<d\hbar\implies f(T)=0\,,
$$
so that $f$ is continuous at $0$. 

By the Fenchel-Rockafellar duality theorem (Theorem 1.12 in \cite{Brezis})
$$
\inf_{T\in E}(f(T)+g(T))=\max_{\L\in E'}(-f^*(-\L)-g^*(\L))\,.
$$

Let us compute $f^*$ and $g^*$. First
$$
f^*(-\L)=\sup_{T\in E}(\la-\L,T\ra-f(T))=\sup_{\genfrac{}{}{0pt}{2}{T\in E}{T=T^*\ge -C}}\la-\L,T\ra\,.
$$
If $\L\in E'$ is not $\ge 0$, there exists $T_0=T_0^*\ge 0$ such that $\la\L,T_0\ra=-\a<0$. In particular, $nT_0=nT_0^*\ge -C$ for each $n\ge 0$, so that
$$
f^*(-\L)\ge\sup_{n\ge 1}\la-\L,nT_0\ra=\sup_{n\ge 1}n\a=+\infty\,.
$$
For $\L\in E'$ such that $\L\ge 0$, define
$$
\la\L,C\ra:=\sup_{\genfrac{}{}{0pt}{2}{T\in E}{T=T^*\le C}}\la\L,T\ra\in[0,+\infty]\,.
$$
(That $\la\L,C\ra\ge 0$ comes from observing that $T=0$ satisfies the constraints.) With this definition
$$
f^*(-\L)=\left\{\ba &\la\L,C\ra&&\qquad\hbox{ if }\L\ge 0\,,\\&+\infty&&\qquad\hbox{ otherwise.}\ea\right.
$$

Next
$$
g^*(\L)=\sup_{T\in E}(\la\L,T\ra-g(T))=\sup_{\genfrac{}{}{0pt}{2}{T=T^*\in E}{T=A\otimes I+I\otimes B}}(\la\L,T\ra-\Tr(RA+SB))\,.
$$
If there exists $A=A^*\in\cL(\fH)$ and $B=B^*\in\cL(\fH)$ such that 
$$
\la\L,A\otimes I+I\otimes B\ra>\Tr(RA+SB)\,,
$$
then
$$
g^*(\L)\ge\sup_{n\ge 1}\left(n\la\L,A\otimes I+I\otimes B\ra-n\Tr_\fH(RA+SB)\right)=+\infty\,.
$$
Likewise, if
$$
\la\L,A\otimes I+I\otimes B\ra<\Tr_\fH(RA+SB)\,,
$$
then
$$
g^*(\L)\ge\sup_{n\ge 1}\left(\la\L,-n(A\otimes I+I\otimes B)\ra-\Tr_\fH(-n(RA+SB))\right)=+\infty\,.
$$
Hence
$$
g^*(\L)=\left\{\ba &0&&\qquad\hbox{ if }\la\L,A\otimes I+I\otimes B\ra=\Tr_\fH(RA+SB)\,,\\&+\infty&&\qquad\hbox{ otherwise.}\ea\right.
$$
Notice that the prescription $\la\L,T\ra=\Tr_\fH(RA+SB)\text{ whenever }T=T^*\in\Ran(\Gamma)$ defines a unique linear functional on $\Ran(\Gamma)$ since $\Ker(\Gamma)=\{0\}$ as explained above.

By the Fenchel-Rockafellar duality theorem recalled above, 
$$
\ba
\inf_{T\in E}(f(T)+g(T))=\inf_{\genfrac{}{}{0pt}{2}{A=A^*\,,\,\, B=B^*\in\cL(\fH)}{A\otimes I+I\otimes B\ge -C}}\Tr_\fH(RA+SB)
\\
=\max_{\L\in E'}(-f^*(-\L)-g^*(\L))=\max_{\genfrac{}{}{0pt}{2}{0\le\L\in E'}{\la\L,A\otimes I+I\otimes B\ra=\Tr_\fH(RA+SB)}}-\la\L,C\ra 
\ea
$$
or equivalently, after exchanging the signs,
$$
\sup_{(A,B)\in\fK}\Tr_\fH(RA+SB)=\min_{\genfrac{}{}{0pt}{2}{0\le\L\in E'}{\la\L,A\otimes I+I\otimes B\ra=\Tr_\fH(RA+SB)}}\la\L,C\ra\,.
$$
(We recall that the constraint $A\otimes I+I\otimes B\le C$ in the definition of $\fK$ is to be understood as explained immediately after the statement of Theorem \ref{T-QDual}.)

\smallskip
One can further restrict the $\min$ on the right hand side with the following observations.

\begin{Lem}\lb{L-L>0}
Let $V=\cL(\scrH)$ where $\scrH$ is a separable Hilbert space. If $\ell\in V'$ satisfies $\ell\ge 0$, then 
$$
T=T^*\in V\implies\la\ell,T\ra\in\bR\qquad\text{ and }\quad\|\ell\|=\la\ell,I_{\scrH}\ra\,.
$$ 
\end{Lem}

\begin{proof}
Indeed, for all $T=T^*\in V$, one has $-\|T\| I_{\scrH}\le T\le\|T\| I_{\scrH}$, so that 
$$
-\|T\| I_{\scrH}\le T\le\|T\| I_{\scrH}\,,\quad\text{ so that }-\|T\|\la\ell,I_{\scrH}\ra\le\la\ell,T\ra\le\|T\|\la\ell,I_{\scrH}\ra\,.
$$
In particular, for all $T=T^*\in V$, one has $\la\ell,T\ra\in\bR$.  For all $T\in V$ (not necessarily self-adjoint), write 
$$
\Re(T)=\tfrac12(T+T^*)\text{ and }\Im(T):=\tfrac12i(T^*-T)\,.
$$
If $\la\ell,T\ra\not=0$, there exists $\a\in\bC$ s.t. $|\a|=1$ and $\la\ell,\a T\ra=|\la\ell,T\ra|$. These considerations show immediately that $\la\ell,\Im(\a T)\ra=0$ so that 
$$
\ba
|\la\ell,T\ra|=\la\ell,\Re(\a T)\ra\le&\la\ell,I_{\scrH}\ra\|\Re(\a T)\|
\\
\le&\tfrac12\la\ell,I_{\scrH}\ra(\|\a T\|+\|(\a T)^*\|)=\la\ell,I_{\scrH}\ra\|T\|\,.
\ea
$$
Hence $\|\ell\|\le\la\ell,I_{\scrH}\ra$, while it is obvious that $\la\ell,I_{\scrH}\ra\le\|\ell\|$. This concludes the proof of Lemma \ref{L-L>0}.
\end{proof}

\begin{Lem}\lb{L-ReprQ}
Let $0\le \L\in E'$. Then there exists $Q\in\cL^1(\fH\otimes\fH)$ such that
$$
Q=Q^*\ge 0\,,\quad\hbox{ and }\quad\|Q\|_{\cL^1}\le\|\L\|\,,
$$
and $L\in E'$ such that
$$
L\ge 0\,,\quad L\rstr_{\cK(\fH\otimes\fH)}=0\,,\quad\hbox{ and }\quad\|L\|\le\|\L\|\,,
$$
satisfying
$$
\L=\Tr_{\fH\otimes\fH}(Q\bu)+L\,.
$$
\end{Lem}

\begin{proof}
Since $\cL^1(\fH\otimes\fH)=\cK(\fH\otimes\fH)'$, one has
$$
\L\rstr_{\cK(\fH\otimes\fH)}=\Tr_{\fH\otimes\fH}(Q\bu)\,,
$$
for some $Q\in\cL^1(\fH\otimes\fH)$. 

First, observe that 
$$
\L\ge 0\implies Q=Q^*\ge 0\,.
$$
Indeed, since $Q\in\cL^1(\fH\otimes\fH)$, then $Q$ is compact. Writing
$$
\Re(Q)=\tfrac12(Q+Q^*)\quad\text{ and }\quad\Im(Q)=-\tfrac{i}2(Q-Q^*)\,,
$$
one has $\Re(Q)=\Re(Q)^*$ and $\Im(Q)=\Im(Q)^*$, so that 
$$
\la\L,\Im(Q)\ra=\Tr_{\fH\otimes\fH}(\Re(Q)\Im(Q))+i\Tr_{\fH\otimes\fH}(\Im(Q)^2)\in\bR\,.
$$
Since 
$$
\ba
\Tr_{\fH\otimes\fH}(\Re(Q)\Im(Q))=&\Tr_{\fH\otimes\fH}(\Im(Q)\Re(Q))
\\
=&\Tr_{\fH\otimes\fH}((\Re(Q)\Im(Q))^*)=\overline{\Tr_{\fH\otimes\fH}(\Re(Q)\Im(Q))}\in\bR\,,
\ea
$$
one concludes that 
$$
\Tr_{\fH\otimes\fH}(\Im(Q)^2)=0\,,\quad\text{ so that }\quad\Im(Q)^2)=0\,.
$$
Thus $Q=Q^*$. Next observe that, for each $\xi\in\fH\otimes\fH$, 
$$
|\xi\ra\la\xi|=(|\xi\ra\la\xi|)^*\ge 0\quad\text{ so that }\quad\la\L,|\xi\ra\la\xi|\ra=\Tr_{\fH\otimes\fH}(Q|\xi\ra\la\xi|)=\la\xi|Q|\xi\ra\ge 0\,.
$$

Let $(\phi_n)_{n\ge 0}$ be a complete orthonormal sequence of eigenvectors of $Q$ in $\fH\otimes\fH$, and let $\ll_n$ be the eigenvalue of $Q$ associated to $\phi_n$. Then
$$
\|Q\|_{\cL^1}=\Tr_{\fH\otimes\fH}(Q)=\sup_{n\ge 1}\sum_{k=1}^n\ll_n=\sup_{n\ge 1}\La\L,\sum_{k=1}^n|\phi_n\ra\la \phi_n|\Ra\le\la\L,I_{\fH\otimes\fH}\ra=\|\L\|\,.
$$

Define 
$$
L:=\L-\Tr_{\fH\otimes\fH}(Q\bu)\,,
$$
so that
$$
L\rstr_{\cK(\fH\otimes\fH)}=0
$$
by construction. Let $\Pi_n$ be the orthogonal projection on $\Span(\phi_0,\ldots,\phi_n)$. Obviously $\Pi_nQ=Q\Pi_n$. Then, for each $T=T^*\ge 0$ in $E$, 
$$
\ba
0\le\la\L,(I_{\fH\otimes\fH}-\Pi_n)T(I_{\fH\otimes\fH}-\Pi_n)\ra=\la\L,T\ra-\la\L,T\Pi_n\ra-\la\L,\Pi_nT\ra+\la\L,\Pi_nT\Pi_n\ra
\\
=\la\L,T\ra-\Tr_{\fH\otimes\fH}(Q(T\Pi_n+\Pi_nT-\Pi_nT\Pi_n))=\la\L,T\ra-\Tr_{\fH\otimes\fH}(\Pi_nQ\Pi_nT)
\\
\to\la\L,T\ra-\Tr_{\fH\otimes\fH}(QT)=\la L,T\ra
\ea
$$
as $n\to\infty$, since $Q\in\cL^1(\fH\otimes\fH)$, so that $\Pi_nQ\Pi_n\to Q$ in $\cL^1(\fH\otimes\fH)$ as $n\to\infty$. This shows that $L\ge 0$. In particular (see footnote above), one has 
$$
\|L\|=\la L,I_{\fH\otimes\fH}\ra=\la\L,I_{\fH\otimes\fH}\ra-\Tr_{\fH\otimes\fH}(Q)\le\la\L,I_{\fH\otimes\fH}\ra=\|\L\|\,.
$$
This conlcudes the proof of Lemma \ref{L-ReprQ}.\end{proof}

\smallskip
\begin{Lem}\lb{L-L=0}
Let $0\le \L\in E'$ satisfy
$$
\la\L,A\otimes I+I\otimes B\ra=\Tr_\fH(RA+SB)\,,\quad\hbox{ for all }A=A^*\hbox{ and }B=B^*\in\cL(\fH)\,.
$$
Then $\L$ is of the form
$$
\L=\Tr_{\fH\otimes\fH}(Q\bu)\,,\quad\hbox{ with }Q=Q^*\ge 0\hbox{ and }\Tr_{\fH\otimes\fH}(Q)=1\,.
$$
In particular, $Q$ is a coupling of $R$ and $S$.
\end{Lem}

\begin{proof}
Let $(e_1,e_2,\ldots)$ be a complete orthonormal system in $\fH$, and let $P_n$ be the orthogonal projection on $\Span(e_1,\ldots,e_n)$. Consider
$$
T_n:=(I_{\fH}-P_n)\otimes P_n+P_n\otimes(I_{\fH}-P_n)\ge 0\,,\qquad n\ge 1\,.
$$
Since $P_n\otimes P_n\ge 0$, one has
$$
0\le T_n\le I_{\fH}\otimes P_n+P_n\otimes I_{\fH}\le I_{\fH\otimes\fH}\,.
$$
Hence
$$
\ba
0\le\la\L,T_n\ra\le\Tr_{\fH\otimes\fH}(Q((I_{\fH}-P_n)\otimes I_{\fH}+I_{\fH}\otimes(I_{\fH}-P_n)))+\la L,T_n\ra&
\\
\le\Tr_{\fH}((Q_1+Q_2)(I_{\fH}-P_n))+\la L,I_{\fH\otimes\fH}\ra
\\
\to\la L,I_{\fH\otimes\fH}\ra=\la\L,I_{\fH\otimes\fH}\ra-\Tr_{\fH\otimes\fH}(Q)
\ea
$$
as $n\to+\infty$. In the formula above, $Q_1,Q_2$ are the partial traces of $Q$, defined as follows:
$$
\ba
Q_1\in\cL^1(\fH)\text{ and }\Tr_\fH(Q_1A)=\Tr_{\fH\otimes\fH}(Q(A\otimes I_\fH))\,,
\\
Q_2\in\cL^1(\fH)\text{ and }\Tr_\fH(Q_2A)=\Tr_{\fH\otimes\fH}(Q(I_\fH\otimes A))\,,
\ea
$$
for each $A\in\cL(\fH)$.

Thus
$$
\varlimsup_{n\to\infty}\la\L,T_n\ra\le\la\L,I_{\fH\otimes\fH}\ra-\Tr_{\fH\otimes\fH}(Q)=1-\Tr_{\fH\otimes\fH}(Q)\,.
$$
Taking $A=I$ and $B=0$ shows indeed that $\la\L,I_{\fH\otimes\fH}\ra=\Tr_\fH(R)=1$.

On the other hand, $(I-P_n)\otimes (I-P_n)\ge 0$, so that
$$
T_n=I_\fH\otimes P_n+P_n\otimes I_\fH-2P_n\otimes P_n\,,
$$
and hence
$$
\ba
\la\L,T_n\ra=&\Tr_\fH((R+S)P_n)-2\la\L,P_n\otimes P_n\ra
\\
=&\Tr_\fH((R+S)P_n)-2\Tr_{\fH\otimes\fH}(Q(P_n\otimes P_n))
\ea
$$
since $P_n\otimes P_n$ is a finite-rank operator (and therefore a compact operator). Thus
$$
\lim_{n\to\infty}\la\L,T_n\ra=\Tr_\fH(R+S)-2\Tr_{\fH\otimes\fH}(Q)=2(1-\Tr_{\fH\otimes\fH}(Q))\,.
$$

Therefore
$$
0\le 2(1-\Tr_{\fH\otimes\fH}(Q))=\lim_{n\to\infty}\la\L,T_n\ra=\varlimsup_{n\to\infty}\la\L,T_n\ra\le 1-\Tr_{\fH\otimes\fH}(Q)\,,
$$
so that
$$
1=\Tr_{\fH\otimes\fH}(Q)\quad\hbox{ and }\quad\| L\|=\la\L,I\ra-\Tr_{\fH\otimes\fH}(Q)=1-\Tr_{\fH\otimes\fH}(Q)=0\,.
$$

Summarizing, we have proved that $\L$ is represented by $Q\in\cL^1(\fH\otimes\fH)$ such that $\Tr_{\fH\otimes\fH}(Q)=1$, and the condition $\L\ge 0$ implies that $Q=Q^*\ge 0$ according to Lemma \ref{L-ReprQ}.
Finally, the definition of $\L$ implies that
$$
\ba
\la\L,A\otimes I_\fH\ra=\Tr_{\fH\otimes\fH}(Q(A\otimes I_\fH))=\Tr_\fH(RA)\,,
\\
\la\L,I_\fH\otimes B\ra=\Tr_{\fH\otimes\fH}(Q(I_\fH\otimes B))=\Tr_\fH(SB)\,,
\ea
$$
so that the partial traces of $Q$ are $Q_1=R$ and $Q_2=S$, meaning that $Q\in\cC(R,S)$. This concludes the proof of Lemma \ref{L-L=0}.
\end{proof}

\smallskip
At this point, we have proved that the minimizing linear functional $\L$ in the duality formula above is represented by $Q\in\cC(R,S)$. In other words,
$$
\sup_{(A,B)\in\fK}\Tr_\fH(RA+SB)=\min_{Q\in\cC(R,S)}\Tr_{\fH\otimes\fH}(QC)\,,
$$
with the notation 
$$
\Tr_{\fH\otimes\fH}(QC):=\sup_{\genfrac{}{}{0pt}{2}{T=T^*\in E}{T\le C}}\Tr_{\fH\otimes\fH}(QT)\,,
$$
where the constraint $T\le C$ has the meaning recalled above. Let us prove that
\be\lb{sup=TrQC}
\sup_{\genfrac{}{}{0pt}{2}{T=T^*\in E}{T\le C}}\Tr_{\fH\otimes\fH}(QT)=\Tr_{\fH\otimes\fH}(Q^{1/2}CQ^{1/2})\,.
\ee

Let $(\Phi_k)_{k\ge 0}$ be a complete orthonormal system of eigenvectors of $Q$, and let $(\ll_k)_{k\ge 0}$ be the sequence of eigenvalues of $Q$ such that $Q\Phi_k=\ll_k\Phi_k$ for each $k\ge 0$. With the notation 
in Appendix \ref{S-QTCost}, one has, by Lemma \ref{L-QfDom},
$$
\ba
2\hb\sum_{k\ge 0}\sum_{\genfrac{}{}{0pt}{2}{m_1,\ldots,m_d\ge 0}{n_1,\ldots,n_d\ge 0}}\ll_k(2(n_1+\ldots+n_d)+d)|\la\Psi_{m_1,\ldots,n_d,n_1,\ldots,n_d}|\Phi_k\ra|^2
\\
=\sum_{k\ge 0}\ll_k\la\Phi_k|C|\Phi_k\ra\le 2\sum_{k\ge 0}\ll_k\la\Phi_k|H\otimes I+I\otimes H|\Phi_k\ra&<\infty\,.
\ea
$$
By Corollary \ref{C-Energ}, one has $C_N:=(I_{\fH\otimes\fH}+\tfrac1NC)^{-1}C=C_N^*\in\cL(\fH\otimes\fH)$ for each $N\ge 1$, and
\be\lb{TrCNQ}
\Tr_{\fH\otimes\fH}(QC_N)\to\Tr_{\fH\otimes\fH}(Q^{1/2}CQ^{1/2})\quad\text{ as }N\to\infty\,.
\ee
Since $0\le C_N=C_N^*\le C$ for each $N\ge 1$
\be\lb{TrCNQ<}
\ba
\lim_{N\to\infty}\Tr_{\fH\otimes\fH}(C_NQ)
\le\sup_{\genfrac{}{}{0pt}{2}{T=T^*\in E}{T\le C}}\Tr_{\fH\otimes\fH}(QT)=\sup_{(A,B)\in\fK}\Tr_\fH(RA+SB)\,.
\ea
\ee
On the other hand, since $\ll_k>0\!\implies\!\Psi_k\!\in\!\fDom(H\otimes I\!+\!I\otimes H)\!\subset\!\fDom(C)$, for each $(A,B)\in\fK$, one has
$$
\sum_{k\ge 0}\ll_k\la\Psi_k|A\otimes I_\fH+I_\fH\otimes B|\Psi_k\ra\le\sum_{k\ge 0}\ll_k\la\Psi_k|C|\Psi_k\ra\,,
$$
or equivalently, since $Q\in\cC(R,S)$, 
\be\lb{TrRA+SB<}
\ba
\Tr_\fH(RA+SB)=&\Tr_{\fH\otimes\fH}(Q^{1/2}(A\otimes I_\fH+I_\fH\otimes B)Q^{1/2})
\\
\le&\Tr_{\fH\otimes\fH}(Q^{1/2}CQ^{1/2})\,.
\ea
\ee
The inequalities \eqref{TrCNQ}, \eqref{TrCNQ<} and \eqref{TrRA+SB<} obviously imply \eqref{sup=TrQC}, and this concludes the proof if Theorem \ref{T-QDual}.


\section{Proof of Theorem \ref{T-ExOptiAB}}\label{S-proof2}


Let $(A_k,B_k)\in\fK$ be a maximizing sequence, i.e.
$$
\Tr(RA_k+SB_k)\to\sup_{(A,B)\in\fK}\Tr(RA+SB)=:\tau\in[0,+\infty)\quad\text{ as }k\to\infty\,.
$$
That $\tau<+\infty$ comes from the fact that the $\inf$ in Theorem \ref{T-QDual} is attained by some optimal coupling $F\in\cC(R,S)$, and that $R$ and $S$ both belong to $\cD_2(\fH)$. Indeed, using Lemma \ref{L-UBdCyclicity} shows that
$$
\ba
F\in\cC(R,S)\implies 0\le\Tr(F^{1/2}CF^{1/2})\le&2\Tr(F^{1/2}(H\otimes I+I\otimes H)F^{1/2})
\\
=&2\Tr(R^{1/2}HR^{1/2}+S^{1/2}HS^{1/2})<+\infty\,.
\ea
$$

\subsection{Step 1: normalizing the maximizing sequence.}

For each $k\ge 1$, set
$$
a_k:=2H-A_k\quad\text{ and }\quad b_k=2H-B_k\,.
$$
Thus
$$
\ba
a_k\otimes I+I\otimes b_k&\ge 2(H\otimes I+I\otimes H)-C
\\
&=\sum_{j=1}^d((-i\hb\d_{x_j}-i\hb\d_{y_j})^2+(x_j+y_j)^2)=:\Si\ge 0\,.
\ea
$$
The operator $\Si$ satisfies the same uncertainty inequality as $C$:
\be\lb{DefSi}
\ba
\Si=\sum_{j=1}^d((x_j+y_j)+i(-i\hb\d_{x_j}-i\hb\d_{y_j}))((x_j+y_j)-i(-i\hb\d_{x_j}-i\hb\d_{y_j}))
\\
+\sum_{j=1}^di([-i\hb\d_{x_j},x_j]+[-i\hb\d_{y_j},y_j])\ge 2d\hbar I\otimes I&\,,
\ea
\ee
and
\be\lb{fDomSi}
\fDom(\Si)=\{\psi\in\fH\!\otimes\!\fH\text{ s.t. }(x_j\!\!+\!y_j)\psi\text{ and }(\d_{x_j}\!\!\!+\!\d_{y_j})\psi\in\fH\!\otimes\!\fH\,,\,\,1\le j\le d\}\,.
\ee

Set $\a_k:=\sup\{\a\in\bR\hbox{ s.t. }a_k\ge\a I\}$ for each $k\ge 1$. Since $H=H^*\ge 0$, one has $a_k\ge-A_k\ge-\|A_k\|I$, so that $\a_k\ge-\|A_k\|$. On the other hand, let $e_0$ be a normalized eigenvector of $R$ such that $Re_0\not=0$. 
Since $R\in\cD_2(\fH)$, one has $0\le\la e_0|H|e_0\ra<+\infty$ by \eqref{J0infDom} (see Lemma \ref{L-Energ} in the Appendix), so that
$$
a_k\ge\a I\implies\a\le\la e_0|a_k|e_0\ra\le 2\la e_0|H|e_0\ra+\|A_k\|\,.
$$
Hence $\a_k\in[-\|A_k\|,2\la e_0|H|e_0\ra+\|A_k\|]$. By definition of $\a_k$, there exists $\phi_n\in\Dom(H)$ such that 
$$
\|\phi_n\|_\fH=1\quad\hbox{ and }\la\phi_n|a_k|\phi_n\ra\to\a_k\hbox{ as }n\to\infty\quad\hbox{ for each }k\ge 1\,.
$$
Thus
$$
\la\phi_n|a_k|\phi_n\ra I+b_k\ge 2d\hbar I\quad\hbox{ for each }n\ge 1\,,
$$
so that
$$
\a_kI+b_k\ge 2d\hbar I\,.
$$
On the other hand, again by definition of $\a_k$, one has
$$
a_k-\a_kI\ge 0\,.
$$
Setting
$$
\hat a_k:=a_k-\a_kI+d\hbar I\,,\quad\hat b_k:=b_k+\a_kI-d\hbar I\,,
$$
one has
$$
\ba
\hat a_k\otimes I+I\otimes\hat b_k=a_k\otimes I+I\otimes b_k\ge\Si&\,,
\\
\hat a_k=\hat a_k^*\ge d\hbar I\,,\qquad\hat b_k=\hat b_k^*\ge d\hbar I&\,.
\ea
$$
Finally
$$
\ba
\\
0\le\Tr_\fH(R^{1/2}\hat a_kR^{1/2}+S^{1/2}\hat b_kS^{1/2})=&\Tr_\fH(R^{1/2}a_kR^{1/2}+S^{1/2}b_kS^{1/2})
\\
=&2\Tr_\fH(R^{1/2}HR^{1/2}+S^{1/2}HS^{1/2})
\\
&-\Tr_\fH(RA_k+SB_k)
\\
\to&2\Tr_\fH(R^{1/2}HR^{1/2}+S^{1/2}HS^{1/2})-\tau
\ea
$$
as $k\to\infty$.

\subsection{Step 2: defining the unbounded operators $\fa$ and $\fb$.}

With the minimizing sequence $(a_k,b_k)$ replaced with its normalized variant $(\hat a_k,\hat b_k)$ as explained in the previous section, one has
$$
\ba
0\le\Tr(R^{1/2}\hat a_kR^{1/2})\le\sup_k\Tr(R^{1/2}\hat a_kR^{1/2})<+\infty\,,
\\
0\le\Tr(S^{1/2}\hat b_kS^{1/2})\le\sup_k\Tr(S^{1/2}\hat b_kS^{1/2})<+\infty\,,
\ea
$$
since both these sequences are converging as $k\to\infty$. Therefore, the sequences of operators $R^{1/2}\hat a_kR^{1/2}$ and $S^{1/2}\hat b_kS^{1/2}$ are bounded in $\cL^1(\fH)$. Since $\cL^1(\fH)$ is the topological dual of $\cK(\fH)$
(the algebra of compact operators on $\fH$), the Banach-Alaoglu theorem implies that there exists a subsequence of $(\hat a_k,\hat b_k)$ (abusively denoted $(\hat a_k,\hat b_k)$ for simplicity) such that
$$
R^{1/2}\hat a_kR^{1/2}\to V\quad\hbox{ and }\quad S^{1/2}\hat b_kS^{1/2}\to W\quad\hbox{ in }\cL^1(\fH)\text{ weak-* as }k\to\infty\,.
$$
Since 
$$
\hat a_k=\hat a_k^*\ge d\hbar I\quad\hbox{ and }\quad\hat b_k=\hat b_k^*\ge d\hbar I\,,
$$
one has 
$$
V=V^*\ge d\hbar R\quad\hbox{ and }\quad W=W^*\ge d\hbar S\,.
$$
In particular
$$
\Ker(V)\subset\Ker(R)\quad\text{ and }\quad\Ker(W)\subset\Ker(S)\,.
$$
On the other hand
$$
\Ran(V)\subset\overline{\Ran(R^{1/2})}=\overline{\Ran(R)}\quad\text{ and }\quad\Ran(W)\subset\overline{\Ran(S^{1/2})}=\overline{\Ran(S)}\,.
$$
(To check the first inclusion, pick $\xi=Vx$, and observe that 
$$
\la y|R^{1/2}\hat a_kR^{1/2}|x\ra=\Tr(R^{1/2}\hat a_kR^{1/2}|x\ra\la y|)\to\la y|V|x\ra
$$
so that $\xi_k=R^{1/2}\hat a_kR^{1/2}x\in\Ran(R^{1/2})$ satisfies $\xi_k\to\xi$ weakly in $\fH$. Hence $\xi$ belongs to the weak closure of $\Ran(R^{1/2})$, which is equal to its strong closure $\overline{\Ran(R^{1/2})}$ since $\Ran(R^{1/2})$ 
is a convex subset of $\fH$: see Theorem 3.7 in \cite{Brezis}.) Since
$$
\ba
\Ker(V)^\perp=\overline{\Ran(V)}\subset\overline{\Ran(R)}=\Ker(R)^\perp\,,
\\
\Ker(W)^\perp=\overline{\Ran(W)}\subset\overline{\Ran(S)}=\Ker(S)^\perp\,,
\ea
$$
(see Corollary 2.18 (iv) in \cite{Brezis}) one has
$$
\ba
\Ker(V)=\Ker(R)\quad\text{ and }\quad\overline{\Ran(V)}=\Ker(R)^\perp\,,
\\
\Ker(W)=\Ker(S)\quad\text{ and }\quad\overline{\Ran(W)}=\Ker(S)^\perp\,.
\ea
$$
In particular
$$
V\in\cL^1(\Ker(R)^\perp)\quad\text{ and }\quad W\in\cL^1(\Ker(S)^\perp)\,.
$$
Let $\fv\in\cL(\cJ[R],\cJ[R]')$ and $\fw\in\cL(\cJ[S],\cJ[S]')$ be the operators associated to $V$ and $W$ by \eqref{Isom2}; since $V=V^*$ and $W=W^*$, one has
$$
\fv^*=\fv\quad\text{ and }\quad\fw^*=\fw\,.
$$
Next 
$$
\cJ_0[R]\otimes\cJ_0[S]\subset\fDom(H\otimes I+I\otimes H)\subset\fDom(\Si)
$$
where the first inclusion comes from \eqref{J0xJ0in}, and the second from \eqref{fDom} and \eqref{fDomSi}. By construction, the sequences $(\hat a_k)_{k\ge 1}$ and $(\hat b_k)_{k\ge 1}$ satisfy
$$
\la\Phi|\hat a_k\otimes I+I\otimes\hat b_k-\Si|\Phi\ra\ge 0\quad\hbox{ for all }\Phi\in\cJ_0[R]\otimes\cJ_0[S]\,.
$$
For each $\phi\in\cJ_0[R]$, there exists a unique $\wtilde\phi\in\cJ_0[R]$ such that $R^{1/2}\wtilde\phi=\phi$, so that
$$
\la\phi|\hat a_k|\phi\ra=\la\wtilde\phi|R^{1/2}\hat a_kR^{1/2}|\wtilde\phi\ra\to\la\wtilde\phi|V|\wtilde\phi\ra=\la\phi|\fv|\phi\ra\qquad\text{ as }k\to\infty\,.
$$
Likewise $\la\psi|\hat b_k|\psi\ra\to\la\psi|\fw|\psi\ra$ as $k\to\infty$ for each $\psi\in\cJ_0[R]$. Passing to the limit in the last inequality implies that
$$
\la\Phi|\fv\otimes I+I\otimes\fw-\Si|\Phi\ra\ge 0\quad\hbox{ for all }\Phi\in\cJ_0[R]\otimes\cJ_0[S]\,.
$$
Let $\fa=\fa^*\in\cL(\cJ[R],\cJ[R]')$ and $\fb=\fb^*\in\cL(\cJ[S],\cJ[S]')$ be the operators associated to $2R^{1/2}HR^{1/2}-V\in\cL^1((\Ker(R)^\perp))$ and to $2S^{1/2}HS^{1/2}-W\in\cL^1((\Ker(S)^\perp))$ respectively. The last inequality on 
$\fv$ and $\fw$ implies that $(\fa,\fb)\in\tilde\fK(R,S)$.

\subsection{Step 3: relaxing the constraint.}

In this step we prove the following: for each $(\bar a,\bar b)\in\tilde\fK(R,S)$ and each $F\in\cC(R,S)$, one has
\be\lb{RelaxIneq}
\Tr_{\fH\otimes\fH}(F^{1/2}CF^{1/2})\ge\Tr_\fH(R^{1/2}\bar aR^{1/2}+S^{1/2}\bar bS^{1/2})\,.
\ee

Let $(e'_j)$ and $(f'_l)$ be orthonormal sequences of eigenvectors of $R$ and $S$ belonging to $\cJ_0[R]$ and $\cJ_0[S]$, and assumed to be complete in $\Ker(R)^\perp$ and $\Ker(S)^\perp$ respectively. Call $\fp_m$ and $\fq_n$
the orthogonal projections on the $m$ first elements of $(e'_j)$ and on the $n$ first elements of $(f'_l)$ respectively, so that
$$
\ba
0\le\fp_1\le\ldots\le\fp_m\le\sup_{m}\fp_m=\fP=\text{ orthogonal projection on }\ker(R)^\perp\,,
\\
0\le\fq_1\le\ldots\le\fq_n\le\sup_{n}\fq_n=\fQ=\text{ orthogonal projection on }\ker(S)^\perp\,.
\ea
$$

We shall argue instead in terms of the operators 
$$
\bar v=\bar v^*\in\cL(\cJ[R],\cJ[R]')\quad\text{ and }\quad\bar w=\bar w^*\in\cL(\cJ[S],\cJ[S]')
$$
associated by \eqref{Isom2} to the operators 
$$
\ba
R^{1/2}HR^{1/2}\!-\!R^{1/2}\bar aR^{1/2}\in\cL^1(\Ker(R)^\perp)\,,
\\
S^{1/2}HS^{1/2}-\,\,S^{1/2}\bar bS^{1/2}\in\cL^1(\Ker(S)^\perp)\,.
\ea
$$

Since $(\bar a,\bar b)\in\fK(R,S)$, one has $\la\Phi|C-\bar a\otimes I-I\otimes\bar b|\Phi\ra\ge 0$ for each $\Phi\in\cJ_0[R]\otimes\cJ_0[S]$, so that, for each $m,n$, one has
$$
\ba
\fp_m\otimes\fq_n\Si_N\fp_m\otimes\fq_n\le\fp_m\otimes\fq_n\Si\fp_m\otimes\fq_n\le&(\fp_m\bar v\fp_m)\otimes\fq_n+\fp_m\otimes(\fq_n\bar w\fq_n)
\\
\le&(\fp_m\bar v\fp_m)\otimes\fQ+\fP\otimes(\fq_n\bar w\fq_n)\,,
\ea
$$
where $\Si_N:=(I_{\fH\otimes\fH}+\tfrac1N\Si)^{-1}\Si=\Si_N^*\in\cL(\fH\otimes\fH)$ and $0\le\Si_N\le\Si$ for each $N\ge 1$.

That $(\fp_m\bar v\fp_m)\otimes\fq_n\le(\fp_m\bar v\fp_m)\otimes\fQ$ is seen easily, for instance by the following argument. Let $\Phi$ be any element of $\Ker(R)^\perp\otimes\Ker(S)^\perp$, which we decompose on the complete orthonormal system $(e'_j\otimes f'_l)$:
$$
\Phi=\sum_{j,l}\Phi_{jl}e'_j\otimes f'_l\,,\quad\sum_{j,l}|\Phi_{jl}|^2=\|\Phi\|_{\fH\otimes\fH}^2<\infty\,.
$$
Then
$$
\ba
\la\Phi|(\fp_m\bar v\fp_m\otimes\fq_n|\Phi\ra&=\sum_{\genfrac{}{}{0pt}{2}{1\le j,k\le m}{1\le l\le n}}\overline{\Phi_{jl}}\Phi_{kl}\la\bar ve'_j|e'_k\ra_{\cV',\cV}
\\
&\le\sum_l\sum_{1\le j,k\le m}\overline{\Phi_{jl}}\Phi_{kl}\la\bar ve'_j|e'_k\ra_{\cV',\cV}=\la\Phi|\fp_n\bar v\fp_n\otimes\fQ|\Phi\ra\,,
\ea
$$
since the matrix $(\la\bar ve'_j|e'_k\ra_{\cV',\cV})_{1\le j,k\le m}$ is Hermitian nonnegative. The analogous inequality for $\bar w$ is proved similarly. 

Thus, for each $F\in\cC(R,S)$, one has
$$
\Tr_{\fH\otimes\fH}(F(\fp_m\otimes\fq_n)\Si_N(\fp_m\otimes\fq_n))\le\Tr_{\fH\otimes\fH}(F((\fp_m\bar v\fp_m)\otimes\fQ+\fP\otimes(\fq_n\bar w\fq_n)))\,.
$$

\begin{Lem}\lb{L-SuppCoupling}
Let $R,S\in\cD(\fH)$ and let $\fP$ and $\fQ$ be the orthogonal projections on $\Ker(R)^\perp$ and $\Ker(S)^\perp$ repectively. For each $F\in\cC(R,S)$, one has 
$$
F=(\fP\otimes\fQ)F(\fP\otimes\fQ)\,.
$$
\end{Lem}

Taking this lemma for granted, we conclude the proof of \eqref{RelaxIneq}. First
$$
\ba
\Tr_{\fH\otimes\fH}(F((\fp_m\bar v\fp_m)\otimes\fQ))=&\Tr_{\fH\otimes\fH}(F(\fP\otimes\fQ)((\fp_m\bar v\fp_m)\otimes I)(\fP\otimes\fQ))
\\
=&\Tr_{\fH\otimes\fH}((\fP\otimes\fQ)F(\fP\otimes\fQ)((\fp_m\bar v\fp_m)\otimes I))
\\
=&\Tr_{\fH\otimes\fH}(F((\fp_m\bar v\fp_m)\otimes I))
\\
=&\Tr_\fH(R(\fp_m\bar v\fp_m))=\Tr_\fH(R^{1/2}(\fp_m\bar v\fp_m)R^{1/2})
\\
=&\Tr_\fH(\fp_mR^{1/2}\bar vR^{1/2}\fp_m)\le\Tr_\fH(R^{1/2}\bar vR^{1/2})
\ea
$$
where the first equality comes from the fact that $\fP\fp_m=\fp_m=\fp_m\fP$, the second and the fifth equality follow by cyclicity of the trace, the third equality from the lemma above, the fourth equality from the fact that $F\in\cC(R,S)$,
and the sixth equality from the fact that $\fp_m$ is a spectral projection of $R$ and therefore commutes with $R$. The last inequality is obtained by computing the trace of $R^{1/2}\bar vR^{1/2}\in\cL^1(\fH)$ on a complete orthonormal
system in $\fH$ whose $m$ first vectors span $\Ran(\fp_m)$. By the same token,
$$
\Tr_{\fH\otimes\fH}(F(\fP\otimes(\fq_n\bar w\fq_n)))\le\Tr_\fH(S^{1/2}\bar wS^{1/2})\,.
$$

On the other hand
$$
\ba
\Tr_{\fH\otimes\fH}(F(\fp_m\otimes\fq_n)\Si_N(\fp_m\otimes\fq_n))\to\Tr_{\fH\otimes\fH}(F(\fP\otimes\fQ)\Si_N(\fP\otimes\fQ))
\\
=\Tr_{\fH\otimes\fH}(F\Si_N)
\ea
$$
passing to the limit in $m,n$ for each $N\ge 1$. Indeed
$$
(\fp_m\otimes\fq_n)\Si_N(\fp_m\otimes\fq_n)\to(\fP\otimes\fQ)\Si_N(\fP\otimes\fQ)\quad\text{ strongly in }\cL(\fH\otimes\fH)
$$
for each $N\ge 1$, since for each $\Psi\in\fH\otimes\fH$
$$
\ba
\|(\fP\otimes\fQ)\Si_N(\fP\otimes\fQ)\Psi-(\fp_m\otimes\fq_n)\Si_N(\fp_m\otimes\fq_n)\Psi\|
\\
\le\|(\fP\otimes\fQ-\fp_m\otimes\fq_n)\Si_N(\fP\otimes\fQ)\Psi\|
\\
+\|(\fp_m\otimes\fq_n)\Si_N(\fP\otimes\fQ-\fp_m\otimes\fq_n)\Psi\|
\\
\le\|(\fP\otimes\fQ-\fp_m\otimes\fq_n)\Si_N(\fP\otimes\fQ)\Psi\|
\\
+\|\Si_N\|\|(\fP\otimes\fQ-\fp_m\otimes\fq_n)\Psi\|\to 0
\ea
$$
in $m,n$ for each $N\ge 1$. Then, one concludes as in Example 3 of chapter 2 in \cite{Simon}.

Thus, we have proved that
$$
\Tr_{\fH\otimes\fH}(F\Si_N)\le\Tr_\fH(R^{1/2}\bar vR^{1/2}+S^{1/2}\bar wS^{1/2})\,,\quad\text{ for each }N\ge 1\,.
$$
By Corollary \ref{C-Energ}, $\Tr_{\fH\otimes\fH}(F\Si_N)\to\Tr_{\fH\otimes\fH}(F^{1/2}\Si F^{1/2})$ as $N\to\infty$, so that
$$
\Tr_{\fH\otimes\fH}(F^{1/2}\Si F^{1/2})\le\Tr_\fH(R^{1/2}\bar vR^{1/2}+S^{1/2}\bar wS^{1/2})\,,
$$
which is equivalent to the sought inequality \eqref{RelaxIneq}.

\subsection{Step 4: the squeezing argument}

Pick an optimal coupling $F_{opt}\in\cC(R,S)$. (We recall that the existence of such a coupling is one of the conclusions of Theorem \ref{T-QDual}, and follows from the Fenchel-Rockafellar duality theorem.) One has the following chain of 
inequalities:
\be\lb{SqueezIneqChain}
\ba
\sup_{(A,B)\in\fK}\Tr_\fH(RA+SB)\le\sup_{(\bar a,\bar b)\in\tilde\fK(R,S)}\Tr_\fH(R^{1/2}\bar aR^{1/2}+S^{1/2}\bar bS^{1/2})
\\
\le\Tr_{\fH^{\otimes 2}}(F_{opt}^{1/2}CF_{opt}^{1/2})&\,.
\ea
\ee
The second inequality has been proved in Step 3.

As for the first inequality, observe first that
\be\lb{SqueezIneqChain1}
\sup_{(A,B)\in\fK}\Tr_\fH(RA+SB)=\sup_{(A,B)\in\hat\fK}\Tr_\fH(RA+SB)\,,
\ee
with the notation
$$
\hat\fK:=\{(A,B)\in\fK\text{ s.t. }A\le 2H-d\hbar I\text{ and }B\le 2H-d\hbar I\}\,.
$$
This is proved by the normalization argument in Step 1: pick
$$
\rho:=\sup\{\a\in\bR\text{ s.t. }2H-A\ge\a I\}\,.
$$
Then $\rho\in[-\|A\|,2\la e_0|H|e_0\ra+\|A\|]$, where $e_0$ is a normalized eigenvector of $R$ such that $Re_0\not=0$, and one has
$$
A+(\rho-d\hb)I\le 2H-d\hb I\quad\text{ and }\quad B-(\rho-d\hb)I\le 2H-d\hb I
$$
by the same argument as in Step 1. (Indeed, by definition of $\rho$, there exists a sequence $\phi_n\in\Dom(H)$ such that $\|\phi_n\|_\fH=1$ and $\la\phi_n|2H-A|\phi_n\ra\to\rho$ as $n\to\infty$. With the inequality $A\otimes I+I\otimes B\le C$,
this implies that, for each $\psi\in\Dom(H)$, one has
$$
\rho\|\psi\|^2_\fH+\la\psi|2H-B|\psi\ra\ge\la\phi_n\otimes\psi|2(H\otimes I+I\otimes H)-C|\phi_n\otimes\psi\ra\ge 2d\hb\|\psi\|^2_\fH\,,
$$
since $2(H\otimes I+I\otimes H)-C\ge 2d\hb I\otimes I$.) Observing that
$$
(A,B)\in\fK\implies(A+(\rho-d\hb)I,B-(\rho-d\hb)I)\in\hat\fK\,,
$$
and that
$$
\Tr_\fH(RA+SB)=\Tr_\fH(R(A+(\rho-d\hb)I)+S(B-(\rho-d\hb)I))
$$
leads to \eqref{SqueezIneqChain1}.

Let $\fP$ and $\fQ$ be the $\fH$-orthogonal projections on $\Ker(R)^\perp$ and $\Ker(S)^\perp$ respectively, as in the previous section. We claim that
$$
(A,B)\in\hat\fK\implies(\fP A\fP,\fQ B\fQ)\in\tilde\fK(R,S)\,.
$$
Indeed 
$$
\ba
\fP A\fP=(\fP A\fP)^*\in\cL(\Ker(T)^\perp)\subset\cL(\cJ[R],\cJ[R]')
\\
\fQ B\fQ=(\fQ B\fQ)^*\in\cL(\Ker(S)^\perp)\subset\cL(\cJ[S],\cJ[S]')&\,.
\ea
$$
because of the double continuous embedding \eqref{JKJ'}. Then
$$
\ba
2H\ge A\implies 2R^{1/2}HR^{1/2}\ge R^{1/2}AR^{1/2}=R^{1/2}\fP A\fP R^{1/2}
\\
2H\ge B\implies 2S^{1/2}HS^{1/2}\ge S^{1/2}BS^{1/2}=S^{1/2}\fQ B\fQ S^{1/2}
\ea
$$
since $\fP R^{1/2}=R^{1/2}\fP=R^{1/2}$ and $\fQ S^{1/2}=S^{1/2}\fQ=S^{1/2}$, and
$$
\ba
\Tr_\fH(2R^{1/2}HR^{1/2}-R^{1/2}\fP A\fP R^{1/2})&
\\
\le\Tr_\fH(2R^{1/2}HR^{1/2})+\|A\|&<+\infty\,,
\\
\Tr_\fH(2S^{1/2}HS^{1/2}-S^{1/2}\fQ B\fQ S^{1/2})&
\\
\le\Tr_\fH(2S^{1/2}HS^{1/2})+\|B\|&<+\infty\,,
\ea
$$
so that
$$
\ba
2R^{1/2}HR^{1/2}-R^{1/2} \fP A\fP R^{1/2}\in\cL^1(\Ker(R)^\perp)\,,
\\
2S^{1/2}HS^{1/2}-S^{1/2}\fQ B\fQ S^{1/2}\in\cL^1(\Ker(S)^\perp)\,.
\ea
$$
Finally, the inequality 
$$
\la\Phi|A\otimes I+I\otimes B|\Phi\ra\le\la\Phi|C|\Phi\ra
$$
holds for all $\Phi\in\cJ_0[R]\otimes\cJ_0[S]\subset\fDom(C)$. Observe that
$$
\Phi\in\cJ_0[R]\otimes\cJ_0[S]\implies((I-\fP)\otimes I)\Phi=(I\otimes(I-\fQ))\Phi=0\,,
$$
and therefore
$$
\la\Phi|(\fP A\fP)\otimes I+I\otimes(\fQ B\fQ)|\Phi\ra\le\la\Phi|C|\Phi\ra
$$
for all $\Phi\in\cJ_0[R]\otimes\cJ_0[S]$, so that $(\fP A\fP,\fQ B\fQ)\in\wtilde\fK(R,S)$. 

Since $\fP R^{1/2}=R^{1/2}\fP=R^{1/2}$ and $\fQ S^{1/2}=S^{1/2}\fQ=S^{1/2}$, one has
$$
\Tr_\fH(RA+SB)=\Tr_\fH(R^{1/2}\fP A\fP R^{1/2}+S^{1/2}\fQ B\fQ S^{1/2})\,,
$$
we conclude that
\be\lb{SqueezIneqChain2}
\sup_{(A,B)\in\hat\fK}\Tr(RA+SB)\le\sup_{(\bar a,\bar b)\in\hat\fK(R,S)}\Tr(R^{1/2}\bar aR^{1/2}+S^{1/2}BS^{1/2})\,.
\ee

Then \eqref{SqueezIneqChain1} and \eqref{SqueezIneqChain2} imply the chain of inequalities \eqref{SqueezIneqChain}. By the quantum duality theorem (Theorem \ref{T-QDual}), all the inequalities in \eqref{SqueezIneqChain}
are equalities:
\be\lb{SqueezEqChain}
\ba
\sup_{(A,B)\in\fK}\Tr_\fH(RA+SB)=&\sup_{(\bar a,\bar b)\in\tilde\fK(R,S)}\Tr_\fH(R^{1/2}\bar a+S^{1/2}\bar bS^{1/2})
\\
=&\Tr_{\fH\otimes\fH}(F_{opt}^{1/2}CF_{opt}^{1/2})\,.
\ea
\ee

\subsection{Step 5: the pair $(\fa,\fb)\in\wtilde\fK(R,S)$ is optimal.}

For each finite rank orthogonal projection $P=P^*=P^2\in\cL(\fH)$, one has
$$
\ba
\Tr_\fH(PR^{1/2}\fv R^{1/2}P)=\Tr_\fH(PR^{1/2}\fv R^{1/2})
\\
=\lim_{k\to\infty}\Tr_\fH(PR^{1/2}\hat a_kR^{1/2})=\lim_{k\to\infty}\Tr_\fH(PR^{1/2}\hat a_kR^{1/2}P)&\,,
\\
\Tr_\fH(PS^{1/2}\fw S^{1/2}P)=\Tr_\fH(PS^{1/2}\fw S^{1/2})
\\
=\lim_{k\to\infty}\Tr_\fH(PR^{1/2}\hat a_kR^{1/2})=\lim_{k\to\infty}\Tr_\fH(PR^{1/2}\hat a_kR^{1/2}P)&\,,
\ea
$$
since
$$
R^{1/2}\hat a_kR^{1/2}\to V=R^{1/2}\fv R^{1/2}\text{ and }S^{1/2}\hat b_kS^{1/2}\to W=S^{1/2}\fw S^{1/2}
$$
in $\cL^1(\fH)$ weak$-*$ by construction. 

Since $\hat a_k=\hat a_k^*\ge 0$ and $\hat b_k=\hat b_k^*\ge 0$ for each $k\ge 0$ (by construction), one has
$$
\ba
\Tr_\fH(PR^{1/2}\hat a_kR^{1/2}P)=\|PR^{1/2}\hat a_kR^{1/2}P\|_1
\\
\le\|R^{1/2}\hat a_kR^{1/2}\|_1=\Tr_\fH(R^{1/2}\hat a_kR^{1/2})&\,,
\\
\Tr_\fH(PS^{1/2}\hat b_kS^{1/2}P)=\|PS^{1/2}\hat b_kS^{1/2}P\|_1
\\
\le\|S^{1/2}\hat b_kS^{1/2}\|_1=\Tr_\fH(S^{1/2}\hat b_kS^{1/2})&\,.
\ea
$$
Thus, for each finite rank $P=P^*=P^2\in\cL(\fH)$, one has
$$
\ba
\Tr_\fH(P(R^{1/2}\fv R^{1/2}+S^{1/2}\fw S^{1/2})P)\le\lim_{k\to\infty}\Tr_\fH(R^{1/2}\hat a_kR^{1/2}+S^{1/2}\hat b_kS^{1/2})
\\
=\Tr_\fH(2R^{1/2}HR^{1/2}+2S^{1/2}HS^{1/2})-\tau&\,.
\ea
$$
Indeed
$$
\hat a_k=2H-A_k-\a_kI+d\hb I\text{ and }\hat b_k=2H-B_k+\a_kI-d\hb I\,,
$$
so that
$$
\ba
\Tr_\fH(R^{1/2}\hat a_kR^{1/2}+S^{1/2}\hat b_kS^{1/2})
\\
=2\Tr_\fH(R^{1/2}HR^{1/2}+S^{1/2}HS^{1/2})-\Tr_\fH(R^{1/2}A_kR^{1/2}+S^{1/2}B_kS^{1/2})
\\
\to 2\Tr_\fH(R^{1/2}HR^{1/2}+S^{1/2}HS^{1/2})-\tau
\ea
$$
by definition of the sequence $(A_k,B_k)$ (which is a maximizing sequence for the right hand side of \eqref{QDual}).

Since $R^{1/2}\fv R^{1/2}\in\cL^1(\fH)$ and $S^{1/2}\fw S^{1/2}\in\cL^1(\fH)$, one has
$$
\ba
\Tr_\fH(R^{1/2}\fv R^{1/2}+S^{1/2}\fw S^{1/2})=\sup_{\genfrac{}{}{0pt}{2}{P^2=P=P^*}{\Rank(P)<\infty}}\Tr_\fH(P(R^{1/2}\fv R^{1/2}+S^{1/2}\fw S^{1/2})P)
\\
\le 2\Tr_\fH(R^{1/2}HR^{1/2}+S^{1/2}HS^{1/2})-\tau&\,.
\ea
$$
Equivalently, in terms of $\fa$ and $\fb$, one has
$$
\Tr_\fH(R^{1/2}\fa R^{1/2}+S^{1/2}\fb S^{1/2})\ge\tau
$$
and we deduce from the first equality in \eqref{SqueezEqChain} that
$$
\Tr_\fH(R^{1/2}\fa R^{1/2}+S^{1/2}\fb S^{1/2})\ge\sup_{(\bar a,\bar b)\in\tilde\fK(R,S)}\Tr_\fH(R^{1/2}\bar aR^{1/2}+S^{1/2}\bar bS^{1/2}).
$$
Since $(\fa,\fb)\in\tilde\fK(R,S)$ as proved at the end of Step 2, the inequality above is an equality and the pair $(\fa,\fb)$ is optimal. 

\smallskip
Finally, if $R$ and $S$ have finite ranks $\cJ_0[R]=\cJ[R]=\Ker(R)^\perp$ and $\cJ_0[S]=\cJ[S]=\Ker(S)^\perp$. Since these spaces are finite-dimensional, their dual spaces are finite dimensional with the same dimension. Thus the inclusions 
$\Ker(R)^\perp\subset\cJ[R]'$ and $\Ker(S)^\perp\subset\cJ[S]'$ in \eqref{Gelfand} are equalities. Any optimal pair $(\fa_0,\fb_0)\in\tilde\fK(R,S)$ such that
$$
\Tr_\fH(R^{1/2}\fa_0R^{1/2}+S^{1/2}\fb_0S^{1/2})=MK_\hb(R,S)^2
$$
satisfies $\fa_0\in\cL(\Ker(R)^\perp)$ and $\fb_0\in\cL(\Ker(S)^\perp)$. 

This concludes the proof of Theorem \ref{T-ExOptiAB}.

\bigskip
It remains to prove Lemma \ref{L-SuppCoupling}

\begin{proof}[Proof of Lemma \ref{L-SuppCoupling}]
One has
$$
\ba
\Tr_{\fH\otimes\fH}(((I-\fP)\otimes I)F((I-\fP)\otimes I))=&\Tr_{\fH\otimes\fH}(((I-\fP)\otimes I)F)
\\
=&\Tr_{\fH}((I-\fP)R)=0
\ea
$$
since $I-\fP$ is the orthogonal projection on $\Ker(R)$, so that 
$$
((I-\fP)\otimes I)F((I-\fP)\otimes I)=0
$$
since $((I-\fP)\otimes I)F((I-\fP)\otimes I)=(((I-\fP)\otimes I)F((I-\fP)\otimes I))^*\ge 0$. Next observe that
$$
\ba
|\la\phi\otimes\psi|(\fP\otimes I)F((I-\fP)\otimes I)|\phi'\otimes\psi'\ra|^2
\\
\le\la\phi\otimes\psi|(\fP\otimes I)F(\fP\otimes I)|\phi\otimes\psi\ra\la\phi'\otimes\psi'|((I-\fP)\otimes I)F((I-\fP)\otimes I)|\phi'\otimes\psi'\ra
\ea
$$
for each $\phi,\phi',\psi,\psi'\in\fH$ by the Cauchy-Schwarz inequality since $F=F^*\ge 0$, so that
$$
\la\phi\otimes\psi|(\fP\otimes I)F((I-\fP)\otimes I)|\phi'\otimes\psi'\ra=0\,.
$$
Hence
$$
(\fP\otimes I)F((I-\fP)\otimes I)=((\fP\otimes I)F((I-\fP)\otimes I))^*=((I-\fP)\otimes I)F(\fP\otimes I)=0\,,
$$
so that
$$
F=(\fP\otimes I)F(\fP\otimes I)\,.
$$
The same argument shows that
$$
F=(I\otimes\fQ)F(I\otimes\fQ)\,,
$$
so that
$$
F=(I\otimes\fQ)F(I\otimes\fQ)=(I\otimes\fQ)(\fP\otimes I)F(\fP\otimes I)(I\otimes\fQ)=(\fP\otimes\fQ)F(\fP\otimes\fQ)\,,
$$
which is precisely the desired equality.
\end{proof}


\section{Proof of Theorems \ref{T-CSOpt}} \label{S-proof3}



\begin{proof}[Proof of Theorem \ref{T-CSOpt}]
Since $\Phi_j\in\Ker(C-A\otimes I-I\otimes B)$, one has in particular $\Phi_j\in\Dom(C)$ with
$$
\|C\Phi_j\|\le(\|A\|+\|B\|)\|\Phi_j\|=\|A\|+\|B\|\quad\text{ for all }j\,.
$$
Therefore 
$$
\sum_m\ll_m\la\Phi_m|C|\Phi_m\ra\le(\|A\|+\|B\|)\sum_m\ll_m=\|A\|+\|B\|\,,
$$
so that
$$
F^{1/2}CF^{1/2}:=\sum_{j,k}\sqrt{\ll_j\ll_k}\la\Phi_j|C|\Phi_k\ra|\Phi_j\ra\la\Phi_k|\in\cL^1(\fH\otimes\fH)
$$
by Lemma \ref{L-Energ}. Since $\Phi_k\in\Ker(C-A\otimes I-I\otimes B)$ for all $k$, one has
$$
\ba
F^{1/2}CF^{1/2}:=&\sum_{j,k}\sqrt{\ll_j\ll_k}\la\Phi_j|A\otimes I+I\otimes B|\Phi_k\ra|\Phi_j\ra\la\Phi_k|
\\
=&F^{1/2}(A\otimes I+I\otimes B)F^{1/2}\,,
\ea
$$
and thus
$$
\ba
\Tr_{\fH\otimes\fH}(F^{1/2}CF^{1/2})=&\Tr_{\fH\otimes\fH}(F^{1/2}(A\otimes I+I\otimes B)F^{1/2})
\\
=&\Tr_{\fH\otimes\fH}(F(A\otimes I+I\otimes B))=\Tr_\fH(F_1A+F_2B)\,,
\ea
$$
where the second equality follows from cyclicity of the trace, while the third equality comes from the definition of $F_1$ and $F_2$ as the partial traces of $F$.
Therefore
\be\lb{ReverseIneq}
\ba
\inf_{G\in\cC(F_1,F_2)}\Tr_{\fH\otimes\fH}(G^{1/2}CG^{1/2})\le&\Tr_{\fH\otimes\fH}(F^{1/2}CF^{1/2})
\\
=&\Tr_\fH(F_1A+F_2B)\le\sup_{(a,b)\in\fK}\Tr_\fH(F_1a+F_2b)\,.
\ea
\ee
For each $G\in\cC(R,S)$, let $(\Psi_k)_{k\ge 1}$ be a complete orthonormal system of eigenvectors of $G$, and let $(\g_k)_{k\ge 1}$ be the sequence of eigenvalues of $G$, so that $G\Psi_k=\g_k\Psi_k$ for each $k\ge 1$.
Then
$$
\Tr_{\fH\otimes\fH}(G^{1/2}CG^{1/2})<\infty\iff\sum_{k\ge 1}\g_k\la\Psi_k|C|\Psi_k\ra<\infty\,.
$$
Thus, if $\Tr_{\fH\otimes\fH}(G^{1/2}CG^{1/2})<\infty$, one has, as explained in Lemma \ref{L-QfDom}
$$
\Psi_k\in\fDom(C)\text{ for each }k\ge 0\text{ s.t. }\g_k>0\,.
$$
For all $(a,b)\in\fK$, one has therefore
$$
\g_k>0\implies\la\Psi_k|C-a\otimes I-I\otimes b|\Psi_k\ra\ge 0\,,
$$
so that
$$
\ba
0\le&\sum_{k\ge 0}\g_k\la\Psi_k|C-a\otimes I-I\otimes b|\Psi_k\ra
\\
=&\Tr_{\fH\otimes\fH}(G^{1/2}CG^{1/2})-\Tr_{\fH\otimes\fH}(G^{1/2}(a\otimes I+I\otimes b)G^{1/2})
\\
=&\Tr_{\fH\otimes\fH}(G^{1/2}CG^{1/2})-\Tr_{\fH\otimes\fH}(G(a\otimes I+I\otimes b))
\\
=&\Tr_{\fH\otimes\fH}(G^{1/2}CG^{1/2})-\Tr_\fH(F_1a+F_2b)
\ea
$$
Therefore
\be\lb{DirectIneq}
\sup_{(a,b)\in\fK}\Tr_\fH(F_1a+F_2b)\le\inf_{G\in\cC(R,S)}\Tr_{\fH\otimes\fH}(G^{1/2}CG^{1/2})\,.
\ee
Putting together \eqref{ReverseIneq} and \eqref{DirectIneq} leads to the announced result.

Conversely, if $R,S\in\cD_2(\fH)$ satisfies \eqref{MK=Dual}, let $F$ be any optimal coupling of $R$ and $S$. Then
$$
MK_\hb(R,S)^2=\Tr_{\fH\otimes\fH}(F^{1/2}CF^{1/2})=\Tr_\fH(RA+SB)\,,
$$
Since $R,S\in\cD_2(\fH)$, the quantity $MK_\hb(R,S)^2=\Tr_{\fH\otimes\fH}(F^{1/2}CF^{1/2})$ is finite, so that all the eigenvectors of $F$ corresponding to positive eigenvalues belong to $\fDom(C)$.
The second equality above can be equivalently recast as
$$
\Tr_{\fH\otimes\fH}(F^{1/2}(C-A\otimes I-I\otimes B)F^{1/2})=0\,.
$$
Since $\la\Phi|C-A\otimes I-I\otimes B|\Phi\ra\ge 0$ for all $\Phi\in\fDom(C)$, this implies that all the eigenvectors of $F$ corresponding to positive eigenvalues belong to $\Ker(C-A\otimes I-I\otimes B)$. In particular, this nullspace is not equal to $\{0\}$
and $F$ is of the form \eqref{FormOptiCoupl}.
\end{proof}
\section{Proof of Theorem \ref{T-QTransp}}\label{S-proof4}
\begin{proof}[{\tcr Proof of $(1)$ when $A,B\in\fK$\tce}]\ 

\noindent Since $F\in\cC(R,S)$ with $R,S\in\cD_2(\fH)$, any eigenvector $\Phi$ of $F$ such that $F\Phi\not=0$ satisfies 
$$
\Phi\in\fDom(H\otimes I+I\otimes H)\subset\fDom(C)
$$
by Lemma \ref{L-QfDom}. By cyclicity of the trace
$$
\ba
\Tr_{\fH\otimes\fH}(F^{1/2}(A\otimes I+I\otimes B)F^{1/2})=&\Tr_{\fH\otimes\fH}(F(A\otimes I+I\otimes B))
\\
=&\Tr_\fH(RA+SB)\,.
\ea
$$
Since $F^{1/2}(H\otimes I+I\otimes H)F^{1/2}\in\cL^1(\fH\otimes\fH)$ and since $C\le 2(H\otimes I+I\otimes H)$ on $\fDom(H\otimes I+I\otimes H)$, one has
$$
\Tr_{\fH\otimes\fH}(F^{1/2}(C-A\otimes I-I\otimes B)F^{1/2})=0\,.
$$
Let $(\Phi_j)_{j\ge 1}$ be a complete orthonormal sequence of eigenvectors of $F$, and define $\ll_j\ge 0$ by the condition $F\Phi_j=\ll_j\Phi_j$, for each $j\ge 1$. Then
$$
\ba
0=&\Tr_{\fH\otimes\fH}(F^{1/2}(C-A\otimes I-I\otimes B)F^{1/2})
\\
=&\sum_{j\ge 1}\la\Phi_j|F^{1/2}(C-A\otimes I-I\otimes B)F^{1/2}|\Phi_j\ra
\\
=&\sum_{j\ge 1}\ll_j\la\Phi_j|C-A\otimes I-I\otimes B|\Phi_j\ra\,,
\ea
$$
so that
$$
\ll_j>0\implies\la\Phi_j|C-A\otimes I-I\otimes B|\Phi_j\ra=0\,,\quad\text{ for all }j\ge 1\,.
$$
Indeed, since $\Phi_j\in\fDom(C)$ and $(A,B)\in\fK$, one has 
$$
\la\Phi_j|C-A\otimes I-I\otimes B|\Phi_j\ra\ge 0\,,\qquad\text{ for all }j\ge 1\,.
$$
Since $\la\Phi|C-A\otimes I-I\otimes B|\Phi\ra\ge 0$ for each $\Phi\in\fDom(C)$, we conclude from the Cauchy-Schwarz inequality that 
$$
\ll_j>0\implies\la\Phi|C-A\otimes I+I\otimes B|\Phi_j\ra=0\,,\quad\text{ for all }j\ge 1\text{ and }\Phi\in\fDom(C)\,.
$$
In particular, choosing $\Phi=\Psi_{m_1,\ldots,m_d,n_1,\ldots,n_d}$ (in the notation of section \ref{S-QTCost}) shows that
$$
2\hb(2(n_1,\ldots,n_d)+d)\la\Psi_{m_1,\ldots,m_d,n_1,\ldots,n_d}|\Phi_j\ra=\la\Psi_{m_1,\ldots,m_d,n_1,\ldots,n_d}|A\otimes I+I\otimes B|\Phi_j\ra\,,
$$
so that
$$
4\hb^2\sum_{\genfrac{}{}{0pt}{2}{m_1,\ldots,m_d\ge 0}{n_1,\ldots,n_d\ge 0}}(2(n_1,\ldots,n_d)+d)^2|\la\Psi_{m_1,\ldots,m_d,n_1,\ldots,n_d}|\Phi_j\ra|^2\le(\|A\|+\|B\|)^2\,.
$$
This implies that $\Phi_j\in\Dom(C)$ with $\|C\Phi_j\|\le\|A\|+\|B\|$ for each $j\ge 1$. Hence
$$
\ba
\ll_j>0\implies&(C-A\otimes I-I\otimes B)\Phi_j\in\fH\times\fH
\\
\text{ and }&(C-A\otimes I-I\otimes B)\Phi_j\perp\fDom(C)\,.
\ea
$$
Since $\fDom(C)$ is dense in $\fH\otimes\fH$, we conclude that $(C-A\otimes I-I\otimes B)\Phi_j=0$ for all $j\ge 1$ such that $\ll_j>0$. This proves (a).

\smallskip
For each $j=1,\ldots,d$, one has
$$
\ba
(\scrD_{q_j}\otimes I)(C-A\otimes I-I\otimes B)=\scrD_{q_j}(H-A)\otimes I-2I\otimes q_j
=
2(\scrD_{q_j}\vilA\otimes I-I\otimes q_j)
\\
\in\cL(\fDom(H\otimes I+I\otimes H),\fDom(H\otimes I+I\otimes H)')&\,,
\\ 
(\scrD_{p_j}\otimes I)(C-A\otimes I-I\otimes B)=\scrD_{p_j}(H-A)\otimes I-2I\otimes p_j
=
2(\scrD_{p_j}\vilA\otimes I-I\otimes p_j)
\\
\in\cL(\fDom(H\otimes I+I\otimes H),\fDom(H\otimes I+I\otimes H)')&\,.
\ea
$$
Applying \eqref{ExtT} with $T=q_j\otimes I,\,p_j\otimes I,\,I\otimes q_j,\,I\otimes p_j$ shows that all these operators, which are bounded from $\fDom(H\otimes I+I\otimes H)$ into $\fH\otimes\fH$, can be extended as bounded operators 
from $\fH\otimes\fH$ to $\fDom(H\otimes I+I\otimes H)'$. Since $A\in\cL(\fH)$, this shows that the right hand sides of these identities belong to 
$$
\cL(\fDom(H\otimes I+I\otimes H),\fDom(H\otimes I+I\otimes H)'\,.
$$
Next, \eqref{[TC]} with $T=q_j\otimes I$ or $T=p_j\otimes I$ show that these identities hold in the space
$$
\cL(\fDom(H\otimes I+I\otimes H)\cap\Dom(C),\fDom(H\otimes I+I\otimes H)'+\Dom(C)')\,.
$$
Likewise, for $j=1,\ldots,d$, one has
$$
\ba
(I\otimes\scrD_{q_j})(C-A\otimes I-I\otimes B)=2(I\otimes\scrD_{q_j}\vilB-q_j\otimes I)
\\
\in\cL(\fDom(H\otimes I+I\otimes H),\fDom(H\otimes I+I\otimes H)')&\,,
\\ 
(I\otimes\scrD_{p_j})(C-A\otimes I-I\otimes B)=2(I\otimes\scrD_{p_j}\vilB-p_j\otimes I)
\\
\in\cL(\fDom(H\otimes I+I\otimes H),\fDom(H\otimes I+I\otimes H)')&\,.
\ea
$$

Let $\Phi,\Psi\in\Ker(F)^\perp$ be eigenvectors of $F$. According to (a), one has 
$$
\Phi,\Psi\in\Dom(C)\cap\fDom(H\otimes I+I\otimes H)\,,
$$
and
\be\lb{PhiDqPsi}
\ba
2\la(\scrD_{q_j}\vilA\otimes I-I\otimes q_j)\Phi,\Psi\ra_{\fDom(H\otimes I+I\otimes H)',\fDom(H\otimes I+I\otimes H)}
\\
=\tfrac{i}\hb((p_j\otimes I)\Phi|(C-A\otimes I-I\otimes B)\Psi)_{\fH\otimes\fH}
\\
-\tfrac{i}\hb((C-A\otimes I-I\otimes B)\Phi|(p_j\otimes I)\Psi)_{\fH\otimes\fH}&=0
\ea
\ee
for all $j=1,\ldots,d$, since $\Phi,\Psi\in\Ker(C-A\otimes I-I\otimes B)$ by (a). Likewise
\be\lb{PhiDpPsi}
\ba
2\la(\scrD_{p_j}\vilA\otimes I-I\otimes p_j)\Phi,\Psi\ra_{\fDom(H\otimes I+I\otimes H)',\fDom(H\otimes I+I\otimes H)}
\\ 
=-\tfrac{i}\hb((p_j\otimes I)\Phi|(C-A\otimes I-I\otimes B)\Psi)_{\fH\otimes\fH}
\\
+\tfrac{i}\hb((C-A\otimes I-I\otimes B)\Phi|(p_j\otimes I)\Psi)_{\fH\otimes\fH}&=0
\ea
\ee
for all $j=1,\ldots,d$. 

Let $(\Phi_k)_{k\ge 1}$ be a complete orthonormal system of eigenvectors of $F$ in $\fH\otimes\fH$, and let $\ll_k\ge 0$ be defined by $F\Phi_k=\ll_k\Phi_k$. Then
$$
F^{1/2}=\sum_{k\ge 1}\ll_k|\Phi_k\ra\la\Phi_k|\,,
$$
and, for each $T\in\cL(\fDom(H\otimes I+I\otimes H),\fDom(H\otimes I+I\otimes H)')$ and each $\phi,\psi\in\fH\otimes\fH$, one has
$$
\ba
\sum_{\genfrac{}{}{0pt}{2}{k,l\ge 1}{\ll_k,\ll_l>0}}\!\!\!\!\!\sqrt{\ll_k\ll_l}\,(\Phi_l|\psi)_{\fH^{\otimes 2}}\overline{(\Phi_k|\phi)}_{\fH^{\otimes 2}}\la T\Phi_k,\!\Phi_l\ra_{\fDom(H\otimes I\!+\!I\otimes H)',\fDom(H\otimes I\!+\!I\otimes H)}
\\
=\la\phi|F^{1/2}TF^{1/2}|\psi\ra&.
\ea
$$
Observe that this last series is absolutely convergent since
$$
\ba
\sum_{\genfrac{}{}{0pt}{2}{k,l\ge 1}{\ll_k,\ll_l>0}}\ll_k^{1/2}\ll_l^{1/2}|\la T\Phi_k,\Phi_l\ra_{\fDom(H\otimes I+I\otimes H)',\fDom(H\otimes I+I\otimes H)}
\\
\times\overline{(\Phi_k|\phi)}_{\fH\otimes\fH}(\Phi_l|\psi)_{\fH\otimes\fH}|
\\
\le\|T\|_{\cL(\fDom(H\otimes I+I\otimes H),\fDom(H\otimes I+I\otimes H)')}
\\
\times\sum_{\genfrac{}{}{0pt}{2}{k\ge 1}{\ll_k>0}}\ll_k^{1/2}|(\Phi_k|\phi)_{\fH\otimes\fH}|\sum_{\genfrac{}{}{0pt}{2}{l\ge 1}{\ll_l>0}}\ll_l^{1/2}|(\Phi_l|\psi)_{\fH\otimes\fH}|
\\
\le\|T\|_{\cL(\fDom(H\otimes I+I\otimes H),\fDom(H\otimes I+I\otimes H)')}\sum_{n\ge 1}\ll_n
\\
\times\left(\sum_{k\ge 1}|(\Phi_k|\phi)_{\fH\otimes\fH}|^2\right)^{1/2}\left(\sum_{l\ge 1}|(\Phi_k|\psi)_{\fH\otimes\fH}|^2\right)^{1/2}
\\
\le\|T\|_{\cL(\fDom(H\otimes I+I\otimes H),\fDom(H\otimes I+I\otimes H)')}\|\phi\|_{\fH\otimes\fH}\|\psi\|_{\fH\otimes\fH}&\,.
\ea
$$
Hence
\be\lb{SeriesF12TF12}
\ba
\sum_{\genfrac{}{}{0pt}{2}{k,l\ge 1}{\ll_k,\ll_l>0}}\ll_k^{1/2}\ll_l^{1/2}\la T\Phi_k,\Phi_l\ra_{\fDom(H\otimes I+I\otimes H)',\fDom(H\otimes I+I\otimes H)}|\Phi_k\ra\la\Phi_l|
\\
=F^{1/2}TF^{1/2}\in\cL(\fH\otimes\fH)
\ea
\ee
by the Riesz representation theorem. Setting successively 
$$
\ba
T=&\scrD_{q_j}\vilA\otimes I-I\otimes q_j\quad\text{ and }\quad T=&\scrD_{p_j}\vilA\otimes I-I\otimes p_j\,,
\\
T=&I\otimes\scrD_{q_j}\vilB-q_j\otimes I\quad\text{ and }\quad T=&I\otimes\scrD_{p_j}\vilB-p_j\otimes I\,,
\ea
$$
for all $j=1,\ldots,d$ in \eqref{SeriesF12TF12} and using \eqref{PhiDqPsi} and \eqref{PhiDpPsi} implies statement (b).
\end{proof}
{\tcr 
%

\begin{proof}[Proof of $(2)$]\ 

The proof of the statement $(2)$ follows closely the line of the proof of the case $(1)$, simplified by the finite dimensionality.

The densities $R,S$ being of finite rank, $\cJ_0(R)=\cJ(R)=\Ker(R)^\perp$ and the same for $S$. Therefore, by Theorem \ref{T-ExOptiAB} and Definition \ref{deftildek}, we have that
$$
A\otimes I_{\ker(S)^\perp}+I_{\ker(R)^\perp}\otimes B\leq \fpfq C\fpfq:=C'.
$$ 
Moreover by the optimality condition and Lemma \ref{L-SuppCoupling},
\begin{eqnarray}
&&\Tr_{\Ker(R)^\perp\otimes \Ker(S)^\perp}{(F^{1/2}(C'-A\otimes I_{\ker(S)^\perp}+I_{\ker(R)^\perp}\otimes B)F^{1/2})}\nonumber\\
&=&\Tr_{\fH\otimes\fH}{(F^{1/2}(C-A\otimes I-I\otimes B)F^{1/2})}\nonumber\\
&=&0\nonumber
\end{eqnarray}
and, by the Cauchy-Schwarz inequality again, this time on ${\Ker(R)^\perp\otimes \Ker(S)^\perp}$, $(c)$ is proved.

Let us remark that, with the notation  defined right after \eqref{heur1},
\begin{eqnarray}
C'&=&\fP H\fP\otimes I+I\otimes\fQ H\fQ-2\fP Z\fP\otimes\fQ Z\fQ\nonumber\\
&=&
\fP H\fP\otimes I+I\otimes\fQ H\fQ-2\sum_{k=1}^d(Q_k^R\otimes Q_k^S+P_k^R\otimes P_k^S).\nonumber
\end{eqnarray}
Hence, for example, for any $j=1,\dots,d$,
$$
\scrD_{q_j}\otimes I (C'-A\otimes I-I\otimes B)
=
\sum_{k=1}^d(\scrD_{q_j}Q_k^R\otimes Q_k^R
+
\scrD_{q_j}P_k^R\otimes P_k^R
)-\scrD_{q_j}\vilA'\otimes I
$$
so that, by the same argument as before,
$$
F^{1/2}(
\sum_{k=1}^d(\scrD_{q_j}Q_k^R\otimes Q_k^R
+
\scrD_{q_j}P_k^R\otimes P_k^R
)-\scrD_{q_j}\vilA'\otimes I
)
F^{1/2}=0.
$$
By using the fact that $F^{1/2}$ commutes with $\fpfq$ thanks to Lemma \ref{L-SuppCoupling}  and doing the same argument for $\scrD_{p_j}$ instead of  $\scrD_{q_j}$ we get immediately $(d)$.

Note that $F^{1/2}(\scrD_{q_j}\vilA'\otimes I )
F^{1/2}
=
F^{1/2}(\tfrac i\hb[p_j^R,\vilA']\otimes I) 
F^{1/2}$  so that one can replace 
$\scrD_{q_j}\vilA' $ by $
\tfrac i\hb[p_j^R,\vilA']$ in  statement $(d)$.
\end{proof}\tce}


\section{Examples}\label{S-examples}


In this section, we shall study the optimal operators $\fa$ and $\fb$ from the Kantorovich duality theorem, together with the structure of optimal couplings, on a few elementary examples. {\tcr We will also give a necessary and sufficient condition for the optimal coupling of two quantum densities  of semiclassical (T\"oplitz) type to present the same feature.\tce}

\subsection{The case where $R$ is a rank-one projection}



Let $R=|\phi\ra\la\phi|$ with $\|\phi\|_\fH=1$ be a rank $1$ projection, and let $S$ be a finite-rank density operator on the Hilbert space $\fH$. By Theorem \ref{T-ExOptiAB} in the finite rank case, the optimal operators $\fa$ and $\fb$ should be 
sought in the form
\be\lb{FinRkAB}
\fa:=\a|\phi\ra\la\phi|\,,\qquad\fb:=\sum_{k=1}^n\b_k|e_k\ra\la e_k|\,,
\ee
where $(e_k)_{1\le k\le n}$ is an orthonormal basis of $\Ker(S)^\perp$, to be determined along with the real numbers $\a,\b_1,\ldots,\b_n$. We shall see that

\smallskip
\noindent
(a) the basis $(e_j)_{1\le j\le n}$ is orthonormal in $\Ker(S)^\perp$ and orthogonal for the Hermitian form $(\psi,\psi')\mapsto\la\phi\otimes\psi|C|\phi\otimes\psi'\ra$ on $\Ker(S)^\perp$ --- in other words, the lines $\bC e_j$ for $j=1,\ldots,n$ 
are mutually orthogonal principal axes of this Hermitian form in $\Ker(S)^\perp$, while

\noindent
(b) the real numbers $\a+\b_j$ for $j=1,\ldots,n$ are the eigenvalues of the Hermitian (diagonal) matrix with entries $\la\phi\otimes e_j|C|\phi\otimes e_k\ra$ for $j,k=1,\ldots,n$. 

\smallskip
These conditions do not completely determine the orthonormal basis $(e_j)_{1\le j\le n}$ and the real numbers $\a,\b_1,\ldots,\b_n$. For instance, if $(\fa,\fb)$ of the form \eqref{FinRkAB} is optimal, then $(\fa+t|\phi\ra\la\phi|,\fb-tI_{\Ker(S)^\perp})$ 
is also optimal --- this corresponds to changing $\a$ in $\a+t$ and $\b_j$ in $\b_j-t$ for $j=1,\ldots,n$. Likewise, if $\la\phi\otimes e_j|C|\phi\otimes e_j\ra=\la\phi\otimes e_k|C|\phi\otimes e_k\ra$ for some $j\not=k$, the frame $(e_j,e_k)$ can be 
replaced with its image under any rotation in the plane $\Span\{e_j,e_k\}$.

\smallskip
To prove (a)-(b), we begin with an important observation on the set of couplings of $R$ and $S$, which is a straightforward consequence of Lemma 4.1.

\begin{Lem}\lb{L-Coupl}
Assume that $R\in\cD(\fH)$ is a projection. Then $\Rank(R)=1$ and for each $S\in\cD(\fH)$, one has
$$
\cC(R,S)=\{R\otimes S\}\,.
$$
\end{Lem}

This is the quantum analogue of the case where one considers two Borel probability measures $\mu$ and $\nu$, one of which, say $\mu$, is a Dirac measures. In that case, it is obvious that $\Pi(\mu,\nu)=\{\mu\otimes\nu\}$ (all the mass from 
$\nu$ is transported to the support of the Dirac measure). Indeed, pure states, corresponding to density operators of the form $R=|\phi\ra\la\phi|$ where $\phi$ is a normalized element of $\fH$, are the quantum analogues of phase space points
in classical mechanics.

\smallskip
Taking this lemma for granted, $R\otimes S$ is the optimal coupling --- in fact the only coupling --- of $R$ and $S$. Therefore the optimal operators $\fa$ and $\fb$ satisfy
$$
\left\{
\ba
\Tr_{\fH\otimes\fH}((R^{\frac12}\otimes S^{\frac12})(C-\fa\otimes I-I\otimes\fb)(R^{\frac12}\otimes S^{\frac12}))=0\,,
\\
\la\Psi|C-\fa\otimes I-I\otimes\fb|\Psi\ra\ge 0\quad\text{ for }\Psi\!\in\!\Ker(R)^\perp\!\otimes\!\Ker(S)^\perp\,.
\ea
\right.
$$
Hence
$$
\la\phi\otimes\psi|C-\fa\otimes I-I\otimes\fb|\phi\otimes\psi'\ra=0\quad\text{ for all }\psi,\psi'\in\Ker(S)^\perp\,.
$$
This condition can be checked on any basis of $\Ker(S)^\perp$. For instance, using the orthonormal basis $(e_j)_{1\le j\le n}$ of eigenvectors of $\fb$ leads to the identity
$$
\la\phi\otimes e_j|C|\phi\otimes e_k\ra=(\a+\b_j)\de_{jk}\,,\quad\text{ for all }j,k=1,\ldots,n\,.
$$
This obviously implies the conclusions (a) and (b) on the real numbers $\a,\b_1,\ldots,\b_n$ and the orthonormal basis $(e_j)_{1\le j\le n}$ of $\Ker(S)^\perp$.

\smallskip
\begin{proof}[Proof of Lemma \ref{L-Coupl}]
We recall that, if $R$ is an orthogonal projection, one has $\Rank(R)=\Tr(R)$. On the other hand $\Tr(R)=1$ since $R\in\cD(\fH)$. Denoting by $\fQ$ the orthogonal projection on $\Ker(S)^\perp$, Lemma 4.1 implies that
$$
\ba
(R\otimes I)F(R\otimes I)=&(R\otimes I)(R\otimes\fQ)F(R\otimes\fQ)(R\otimes I)
\\
=&(R^2\otimes\fQ)F(R^2\otimes\fQ)=(R\otimes\fQ)F(R\otimes\fQ)=F\,.
\ea
$$
Thus, for each $\phi_1,\phi_2,\psi_1,\psi_2\in\fH$, one has
$$
\ba
\la\phi_1\otimes\psi_1|F|\phi_2\otimes\psi_2\ra=\la\phi_1\otimes\psi_1|(R\otimes I)F(R\otimes I)|\phi_2\otimes\psi_2\ra
\\
=\la\phi_1|e\ra\la e|\phi_2\ra\la e\otimes\psi_1|F|e\otimes\psi_2\ra
\\
=\la\phi_1|R|\phi_2\ra\la\psi_1|G|\psi_2\ra&\,,
\ea
$$
where $\|e\|_\fH=1$ and $\bC e=\text{ran}(R)$, while $G$ is the self-adjoint operator on $\fH$ such that
$$
\la\psi_1|G|\psi_2\ra=\la e\otimes\psi_1|F|e\otimes\psi_2\ra\,,\qquad\psi_1,\psi_2\in\fH\,.
$$
(The existence and uniqueness of $G$ follows from the Riesz representation theorem.)

Hence $F=R\otimes G$, and since $F\in\cC(R,S)$,
$$
\Tr((R\otimes G)(I\otimes B))=\Tr(GB)=\Tr(SB)
$$
for each finite rank $B\in\cL(\fH)$, which implies that $G=S$.
\end{proof}


\subsection{The quantum bipartite matching problem}\label{bipart}
{\tcr A classical bipartite matching problem consists in computing the optimal transport between two probability measures $\mu$ and $\nu$ given by 
\begin{eqnarray}
\mu=\tfrac{1-\eta}2\delta_{-a}+\tfrac{1+\eta}2\delta_{a},\ \  
\nu=\tfrac{1}2\delta_{-b}+\tfrac{1}2\delta_{b},\ \ a,b>0,
\nonumber
\end{eqnarray}
associated to
$$
\MKd(\mu,\nu).
$$
A quantum analogue consists in considering
$$
MK_\hb(R,S)
$$ 
where
\begin{eqnarray}
R=\tfrac{1-\eta}2|-a\rangle\langle a|+\tfrac{1-\eta}2|-a\rangle\langle- a|,\ \mbox{ and } \ 
S=\tfrac12|b,0\rangle\langle b|+\tfrac12|-b\rangle\langle -b|
\nonumber.
\end{eqnarray}

Here $|q\rangle=|q,0\rangle$ where $|q,p\rangle,\ q,p\in\bR,$ is a coherent state defined  by \eqref{CohSta}. Since $a,b>0$, $R,S$ are operators of rank $2$ so that  Theorem \ref{T-QTransp} (2) applies. Since we are in dimension $d=1$, the two first equalities of the result read
\begin{eqnarray}
F^{1/2}(\tfrac i\hb[P^R,Q^R]\otimes Q^S-\tfrac i\hb[P^R,\cA']\otimes I)F^{1/2}&=0&\nonumber\\
F^{1/2}(\tfrac i\hb[Q^R,P^R]\otimes Q^S-\tfrac i\hb[Q^R,\cA']\otimes I)F^{1/2}&=0.&\nonumber
\end{eqnarray}

Note that when $\eta=0$ (equal masses), a classical transport is just any one transporting $-a$ to $-b$ and $a$ to $b$ (see \cite[Section 1]{CagliotiFGPaul}). We will consider the quantum problem in this case $\eta=0$, that is we will study $MK_\hb(R,S)$ where
$$
R:=\frac12(\ket{a}\bra{a}+\ket{-a}\bra{-a})\ \mbox{ and }\  S:=\frac12(\ket{b}\bra{b}+\ket{-b}\bra{-b}).
$$

Define
$$
\lambda:=\bra{a}{-a}\rangle=e^{-a^2/\hbar},\qquad
\mu:=\bra{b}{-b}\rangle=e^{-b^2/\hbar},
$$
and consider the two pairs of orthogonal vectors
\be\label{phipsi}
\phi_\pm:=\frac{\ket{a}\pm\ket{-a}}{\sqrt{2(1\pm\lambda)}},\qquad
\psi_\pm:=\frac{\ket{b}\pm\ket{-b}}{\sqrt{2(1\pm\mu)}}.
\ee
Hence
$$
R=\alpha_+\ketbra{\phi_+}{\phi_+}
+
\alpha_-\ketbra{\phi_-}{\phi_-},\ 
S=\beta_+\ketbra{\psi_+}{\psi_+}
+
\beta_-\ketbra{\psi_-}{\psi_-},\qquad \alpha_\pm:=\frac12(1\pm\lambda),\ 
\beta_\pm:=\frac12(1\pm\mu).
$$

In \cite[Section 4]{CagliotiFGPaul}, we computed an optimal coupling $F$ between $R$ and $S$ of the following form: in the basis $\{\phi_+\otimes\psi_+,\phi_+\otimes\psi_-,\phi_-\otimes\psi_+,\phi_-\otimes\psi_-\}
$ $F$ is expressed by the  matrix
$$
\small{
\frac14
\begin{pmatrix}
1+\lambda\mu+\lambda+\mu&0&0&\sqrt{(1+\lambda\mu)^2-(\lambda+\mu)^2}\\
0&1-\lambda\mu+\lambda-\mu&\sqrt{(1-\lambda\mu)^2-(\lambda-\mu)^2}&0\\
0&\sqrt{(1-\lambda\mu)^2-(\lambda-\mu)^2}&1-\lambda\mu-\lambda+\mu&0\\
\sqrt{(1+\lambda\mu)^2-(\lambda+\mu)^2}&0&0&1+\lambda\mu-\lambda-\mu
\end{pmatrix}.}
$$
Therefore one sees easily that $\Ker(F)$ is generated by the two vectors 
$$
\{|++\rangle-\sqrt{\tfrac{1+\lambda}{1-\lambda}\tfrac{1+\mu}{1-\mu}}|--\rangle,|+-\rangle-\sqrt{\tfrac{1+\lambda}{1-\lambda}\tfrac{1-\mu}{1+\mu}}|-+\rangle\},
$$
 with $|\pm,\pm\rangle=\phi_\pm\otimes\psi_\pm$, so that
$\Ker{(F)}^\perp$ is the  two-dimensional subspace of $\fH$ generated by
$$
\{\varphi_1:=|++\rangle-\sqrt{\tfrac{1-\lambda}{1+\lambda}\tfrac{1-\mu}{1+\mu}}|--\rangle,\varphi_2:=|+-\rangle-\sqrt{\tfrac{1-\lambda}{1+\lambda}\tfrac{1+\mu}{1-\mu}}|-+\rangle\}.
$$

Moreover, straightforward computations show that
$$
Q^R=\tfrac a{\sqrt{1-\lambda^2}}\begin{pmatrix}
0&1\\1&0\end{pmatrix},\ P^R=\tfrac{-ia\lambda}{\sqrt{1-\lambda^2}}
\begin{pmatrix}
0&1\\-1&0
\end{pmatrix},
$$
so that 
$$
\tfrac i\hbar[P^R,Q^R]
=
\tfrac{2a^2\lambda}{1-\lambda^2}
\begin{pmatrix}
-1&0\\0&1
\end{pmatrix}
$$
and
$$
\tfrac i\hbar[P^R,Q^R]\otimes Q^S
=
\tfrac{2a^2b\lambda}{(1-\lambda^2)\sqrt{1-\mu^2}}
\begin{pmatrix}
-1&0\\0&1
\end{pmatrix}
\otimes
\begin{pmatrix}
0&1\\1&0\end{pmatrix}.
$$
 Hence one easily get that, for $\alpha,\beta\in\bC$,
$$
\begin{pmatrix}
-1&0\\0&1
\end{pmatrix}
\otimes
\begin{pmatrix}
0&1\\1&0\end{pmatrix}
\left(\alpha(|++\rangle-\sqrt{\tfrac{1-\lambda}{1+\lambda}\tfrac{1-\mu}{1+\mu}}|--\rangle)+\beta(|+-\rangle-\sqrt{\tfrac{1-\lambda}{1+\lambda}\tfrac{1+\mu}{1-\mu}}|-+\rangle)\right)
$$
$$
=
-\alpha(|+-\rangle+\sqrt{\tfrac{1-\lambda}{1+\lambda}\tfrac{1-\mu}{1+\mu}}|-+\rangle)-\beta(|++\rangle+\sqrt{\tfrac{1-\lambda}{1+\lambda}\tfrac{1+\mu}{1-\mu}}|--\rangle)
$$
and, defining $\|\alpha,\beta\rangle:=\alpha\varphi_1+\beta\varphi_2$,
$$
\langle\alpha',\beta'\|
\begin{pmatrix}
-1&0\\0&1
\end{pmatrix}
\otimes
\begin{pmatrix}
0&1\\1&0\end{pmatrix}
\|\alpha,\beta\rangle
=
(\bar\alpha'\beta+
\bar\beta'\alpha)(-1+\tfrac{1-\lambda}{1+\lambda}).
$$
Moreover
$$
\begin{pmatrix}
1&0\\0&1
\end{pmatrix}
\otimes
\begin{pmatrix}
0&1\\1&0\end{pmatrix}
\left(\alpha(|++\rangle-\sqrt{\tfrac{1-\lambda}{1+\lambda}\tfrac{1-\mu}{1+\mu}}|--\rangle)+\beta(|+-\rangle-\sqrt{\tfrac{1-\lambda}{1+\lambda}\tfrac{1+\mu}{1-\mu}}|-+\rangle)\right)
$$
$$
=
\alpha(|+-\rangle-\sqrt{\tfrac{1-\lambda}{1+\lambda}\tfrac{1-\mu}{1+\mu}}|-+\rangle)+\beta(|++\rangle-\sqrt{\tfrac{1-\lambda}{1+\lambda}\tfrac{1+\mu}{1-\mu}}|--\rangle)
$$
and
$$
\langle\alpha',\beta'\|
\begin{pmatrix}
1&0\\0&1
\end{pmatrix}
\otimes
\begin{pmatrix}
0&1\\1&0\end{pmatrix}
\|\alpha,\beta\rangle
=
(\bar\alpha'\beta+
\bar\beta'\alpha)(1+\tfrac{1-\lambda}{1+\lambda})
$$
so that 
$$
\langle\alpha',\beta'\|
\begin{pmatrix}
-1&0\\0&1
\end{pmatrix}
\otimes
\begin{pmatrix}
0&1\\1&0\end{pmatrix}
\|\alpha,\beta\rangle
=
-\lambda
\langle\alpha',\beta'\|
\begin{pmatrix}
1&0\\0&1
\end{pmatrix}
\otimes
\begin{pmatrix}
0&1\\1&0\end{pmatrix}
\|\alpha,\beta\rangle.
$$
We just proved the following lemma.
\begin{Lem}\label{lemlemlem}
$$
F^{1/2}\big(\tfrac i\hbar[P^R,Q^R]\otimes Q^S\big)F^{1/2}
=
F^{1/2}\big(-\tfrac{2a^2
\lambda^2
}{\hbar(1-\lambda^2)}I\otimes Q^S\big)F^{1/2}
$$
\end{Lem}
We also computed in \cite[Section 2]{CagliotiFGPaul} the matrix $C'$ of $\fP\otimes\fQ C\fP\otimes\fQ$, the cost projected on the range of $R\otimes S$\footnote{The cost used in \cite{CagliotiFGPaul} is shifted by $-2\hbar$ with respect to the one in the present paper.},
\be\label{cost}
C'=\begin{pmatrix}
\cA+2\hbar&0&0&\gamma\\
0&\cB+2\hbar&\delta&0\\
0&\delta&\cC+2\hbar&0\\
\gamma&0&0&\cD+2\hbar
\end{pmatrix}.
\ee
where 
$$
\begin{aligned}
\cA=a^2\frac{1-\lambda}{1+\lambda}+b^2\frac{1-\mu}{1+\mu}&,\qquad\cB=a^2\frac{1-\lambda}{1+\lambda}+b^2\frac{1+\mu}{1-\mu}, \qquad
\gamma=-\frac{2ab(1-\lambda\mu)}{\sqrt{(1-\lambda^2)(1-\mu^2)}},
\\
\cC=a^2\frac{1+\lambda}{1-\lambda}+b^2\frac{1-\mu}{1+\mu}&,\qquad\cD=a^2\frac{1+\lambda}{1-\lambda}+b^2\frac{1+\mu}{1-\mu}, \qquad\delta=-\frac{2ab(1+\lambda\mu)}{\sqrt{(1-\lambda^2)(1-\mu^2)}}.
\end{aligned}
$$
By the same computation, we get that the matrices $H^R$ and $H^S$ of the harmonic oscillator projected on the range of $R$ and the one of $S$ are
$$
H^R=\begin{pmatrix}
a\tfrac{1-\lambda}{1+\lambda}+\hbar&0\\
0&a\tfrac{1+\lambda}{1-\lambda}+\hbar
\end{pmatrix},\ 
H^S=\begin{pmatrix}
b\tfrac{1-\mu}{1+\mu}+\hbar&0\\
0&b\tfrac{1+\mu}{1-\mu}+\hbar
\end{pmatrix}.
$$
Finally, we proved in \cite[Section 2]{CagliotiFGPaul} that two optimal operators $A,B$ can be chosen in the form
$$
A
=
\begin{pmatrix}
\alpha_1&0\\0&\alpha_2
\end{pmatrix}
\qquad
B=
\begin{pmatrix}
\beta_1&0\\0&\beta_2
\end{pmatrix}
$$
where $\alpha_1,\alpha_2,\beta_1,\beta_2$ satisfy
$$
\bar a=\alpha_1+\beta_1-\cA,\quad\bar b=\alpha_1+\beta_2-\cB,\quad\bar c=\alpha_2+\beta_1-\cC,\quad\bar d=\alpha_2+\beta_2-\cD,
$$
with
$$
\bar a+\bar d=\bar b+\bar c=x,\ \ \bar a-\bar d=\sqrt{x^2-4\gamma^2},\ \  \bar b-\bar c=\sqrt{x^2-4\delta^2},\ \ x=-\tfrac{4ab(1-\lambda^2\mu^2)}{(1-\lambda^2)(1-\mu^2)}.
$$
We get, after some algebraic  computations,
$$
\alpha_1-\alpha_2=\bar a-\bar c+\cA-\cC=\frac{4\lambda a}{{1-\lambda^2}}(b-a).
$$
Let us remark now that, if 
$D=\begin{pmatrix}
U&0\\0&V
\end{pmatrix},\  U,V\in\bC$, then
$$
\tfrac i\hbar[P^R,D]=
\tfrac{a\lambda}{\hbar\sqrt{1-\lambda^2}}\begin{pmatrix}
0&V-U\\V-U&0
\end{pmatrix}=\tfrac{\lambda(V-U)}\hbar Q^R.
$$ 
Therefore, defining 
 $\cA'=\tfrac12(H^R- A)$, one find
%
\be\label{eqeqeq}
\tfrac i\hbar[P^R,\cA']=
\tfrac\lambda{2\hbar} 
(a(\tfrac{1-\lambda}{1+\lambda}-\tfrac{1+\lambda}{1-\lambda})+\alpha_1-\alpha_2)Q^R
=\tfrac{2a^2
\lambda^2
}{\hbar(1-\lambda^2)}\tfrac baQ^R.
\ee

By Lemma \ref{lemlemlem} and \eqref{eqeqeq}, and the same type of computations changing $Q^R$ in $P^R$ and $Q^S$ in $P^S$,  we get finally the following result.
\begin{Prop}\label{oufder}In the equal mass situation, we have
\begin{eqnarray}
F^{1/2}(I\otimes Q^S-\tfrac baQ^R\otimes I)F^{1/2}&=&0\nonumber\\
F^{1/2}(I\otimes P^S-\tfrac {be^{-\frac{b^2}\hb}}{ae^{-\frac{a^2}\hb}}P^R\otimes I)F^{1/2}&=&0\nonumber
\end{eqnarray}
which correspond to a transport $(-a,a)\to(-b,b)$. The renormalization of $Q^R$  and $P^R$ by $\tfrac b a$ and $\tfrac {be^{-\frac{b^2}\hb}}{ae^{-\frac{a^2}\hb}}$ respectively corresponds to sending $(Q^R,P^R)$ to $(Q^S,P^S)$ by a transport not in the (usual) form of a conjugation by a unitary transform sending $\{\phi_+,\phi_-\}$ to $\{\psi_+,\psi_-\}$.
\end{Prop}
\tce}
\subsection{The case where $R=S$ is a T\"oplitz operator}


We first recall that 
$$
MK_\hb(R,R)\ge 2d\hb>0
$$
by \eqref{HeisenIneq}, at variance with the classical setting, where $\MKd(\mu,\mu)=0$ for all $\mu\in\cP_2(\bR^d\times\bR^d)$. Therefore, computing an optimal coupling and optimal operators $\fa$ and $\fb$ is nontrivial problem even in this case.
However, this problem can be solved when $R$ is a T\"oplitz operator, {\tcr  defined as follows.\tce}

 Let $\mu$ be a Radon measure on $\bR^d\times\bR^d$; the (possibly unbounded) T\"oplitz operator with symbol $\mu$ is defined by duality by the formula
\be\lb{Toeplitz}
\la u|\Op^T_\hb[\mu]|v\ra:=\frac1{(2\pi\hb)^d}\int_{\bR^d\times\bR^d}\la u|q,p\ra\la q,p|v\ra\mu(dqdp)\,,
\ee
for all $u,v\in L^2(\bR^d)$ such that the functions $(q,p)\mapsto\la u|q,p\ra$ and $(q,p)\mapsto\la v|q,p\ra$ both belong to $L^2(\bR^d\times\bR^d,\mu)$, where
\be\lb{CohSta}
|q,p\ra(x):=(\pi\hb)^{-d/4}e^{-|x-q|^2/2\hb}e^{ip\cdot x/\hb}\,.
\ee

We recall that
$$
\mu\in\cP_2(\bR^d\times\bR^d)\implies\Op^T_\hb[\mu]\in\cD_2(\fH)\,,
$$
(see \cite{FGPaulCRAS}, Theorem 2.2 (iii)), and that
$$
MK_\hb(\Op^T_\hb[\mu],\Op^T_\hb[\nu])^2\le \MKd(\mu,\nu)^2+2d\hb
$$
for all $\mu,\nu\in\cP_2(\bR^d\times\bR^d)$ (see \cite{FGPaulCRAS}, Theorem 2.2 (iii), or Theorem 2.3 (1) in \cite{FGMouPaul}), while 
$$
2d\hb\le MK_\hb(R,S)^2\quad\text{ for all }R,S\in\cD_2(\fH)\,,
$$
according to fla. (14) in \cite{FGMouPaul}.

Hence
$$
MK_\hb(\Op^T_\hb[\mu],\Op^T_\hb[\mu])^2=2d\hb\,.
$$
This is Corollary 2.4 in \cite{FGPaulCRAS}.

An optimal element of $\cC(\Op^T_\hb[\mu],\Op^T_\hb[\mu])$ is
$$
F:=\int_{\bR^d\times\bR^d}|q,p\ra\la q,p|^{\otimes 2}\mu(dqdp)\,.
$$
That $F\in\cC(\Op^T_\hb[\mu],\Op^T_\hb[\mu])$ follows from Lemma 4.1 in \cite{FGMouPaul}. This is the analogue of the diagonal coupling $\text{diag}\#\mu$ of one Borel probability measure $\mu$ with itself, 
where $\text{diag}$ is the diagonal embedding $\text{diag}:\,x\mapsto(x,x)$. (Informally, the diagonal coupling is $\mu(dx)\de(y-x)$.)

Moreover
$$
\ba
\Tr(F^{1/2}CF^{1/2})=\sup_{\eps>0}\Tr(F(I+\eps C)^{-1}C)
\\
=\int_{\bR^d\times\bR^d}\la q,q,p,p|C|q,q,p,p\ra\mu(dqdp)=2d\hb\,.
\ea
$$
as explained in the proof of Lemma 2.1 of \cite{FGPaulCRAS}. In the formula above, we have denoted
$$
|q_1,q_2,p_1,p_2\ra=|q_1,p_1\ra\otimes|q_2,p_2\ra\,,\quad\la q_1,q_2,p_1,p_2|=\la q_1,p_1|\otimes\la q_2,p_2|
$$

We claim that one can choose in this case
\be\lb{ab=R=SToeplitz}
\fa=\fb=d\hb I_\fH\,.
\ee
Indeed, according to the Heisenberg uncertainty principle
$$
C\ge 2d\hb I\otimes I=\fa\otimes I+I\otimes\fa\,.
$$
On the other hand
$$
\ba
\Tr(\fa\Op^T_\hb[\mu])=d\hb\Tr(\Op^T_\hb[\mu])
\\
=d\hb\int_{\bR^d\times\bR^d}\mu(dxd\xi)=d\hb&\,,
\ea
$$
so that, with the choice of $\fa$ and $\fb$ above, one has
$$
\Tr(\fa\Op^T_\hb[\mu])+\Tr(\fb\Op^T_\hb[\mu])=2d\hb=\Tr(F^{1/2}CF^{1/2})\,.
$$
In the classical setting, the optimal functions $\phi$ and $\psi$ in \eqref{kanto1} are $\phi_{op}=\psi_{op}=0$. This is in complete agreement with \eqref{ab=R=SToeplitz} in the limit as $\hb\to 0$.

\subsection{When is the optimal coupling a T\"oplitz operator?}


Let $R$ and $S$ be T\" oplitz density operators, of the form $R=\Op^T_\hb[(2\pi\hb)^d\mu]$ and $S=\Op^T_\hb[(2\pi\hb)^d\nu]$ with $\mu,\nu\in\cP_2(\bR^d\times\bR^d)$. When is an optimal coupling of  $R$ and $S$ a T\" oplitz operator, 
of the form $F=\Op^T_\hb[(2\pi\hb)^{2d}\ll]$ for some $\ll\in\cP(\bR^d\times\bR^d\times\bR^d\times\bR^d)$? We shall see that this question is answered in the affirmative only under rather stringent conditions.

\smallskip
We already know the answer in two different cases discussed above:

\noindent
(a) $R=S$ (see previous section);

\noindent
(b) $\mu=\de_{q,p}$ and $\nu=\de_{q',p'}$, in which case $R$ and $S$ are rank-one operators, in which case the only (and therefore the optimal) coupling is 
$$
R\otimes S=\Op^T_\hb[(2\pi\hb)^{2d}\mu\otimes\nu]\,.
$${\tcr 
Moreover, as recalled at the end of Section \ref{S-Intro} and in Section \ref{bipart}, we studied in 
\cite{CagliotiFGPaul} the case where 
\begin{eqnarray}
R&=&\tfrac{1+\eta}2|a,0\rangle\langle a,0|+\tfrac{1-\eta}2|-a,0\rangle\langle- a,0|=\Op^T_\hb[(2\pi\hb)^{2}\mu],\ \mu=\tfrac{1+\eta}2\delta_{(a,0)}+\tfrac{1-\eta}2\delta_{(-a,0)}\nonumber\\
S&=&\tfrac12|b,0\rangle\langle b,0|+\tfrac12|-b,0\rangle\langle -b,0|=\Op^T_\hb[(2\pi\hb)^{2}\nu],\   \nu =\tfrac12\delta_{(b,0)}+\tfrac12\delta_{(-b,0)},\  a,b\in\bR^+\nonumber.
\end{eqnarray}
we proved in \cite[Section 4]{CagliotiFGPaul} that
\begin{itemize}
\item when $\eta=0$ (equal mass case), an optimal quantum coupling $F$ is the T\"oplitz quantization of a classical one $f$: 
\begin{eqnarray}
F&=&\frac12\big(|a\rangle\otimes|b\rangle \langle a|\otimes\langle b|+
|-a\rangle\otimes|-b\rangle \langle -a|\otimes\langle -b|
\big)\nonumber\\
&=&
\Op^T_\hb[(2\pi\hb)^{2}\tfrac12\big(\delta_{(a,0)}\otimes\delta_{(b,0)}+
\delta_{(-a,0)}\otimes\delta_{(-b,0)}\big)]
:=
\Op^T_\hb[(2\pi\hb)^{2}f]\nonumber
\end{eqnarray}
\item when $\eta\neq 0$ (non equal mass case), one easily shows that (we take $a=b$) a classical optimal coupling is 
$$
f=
\tfrac12\delta_{(a,0}\otimes\delta_{(a,0)}+
\tfrac{1-\eta}2\delta_{(-a,0}\otimes\delta_{(-a,0)}+\tfrac\eta2
\delta_{(a,0}\otimes\delta_{(-a,0)},
$$
and we proved that, non only $\Op^T_\hb[(2\pi\hb)^{2}f]$ is not an optimal coupling of $R,S$, but no optimal coupling of $R,S$ can be a T\"oplitz operator in this case.
\end{itemize}
\tce}

\smallskip
In the analysis below, we consider this problem when $\mu$ and $\nu$ are of the form
$$
\mu(dqdp)=m(q,p)dqdp\,,\quad\nu(dqdp)=n(q,p)dqdp\,,\qquad m,n>0\text{ a.e..} 
$$
(The case (b) mentioned above obviously fails to satisfy this assumption.) 

This assumption clearly implies that $\Ker(R)=\Ker(S)=\{0\}$. To see this, we first recall one definition of the Husimi transform of an operator $A$ on $L^2(\bR^d)$:
\be\lb{Husimi}
\tilde W_\hb[A](q,p):=\tfrac1{(2\pi\hb)^d}\la q,p|A|q,p\ra\,,\qquad q,p\in\bR^d\,.
\ee
(There is another, equivalent definition in terms of the Wigner transform: see (49), and the formula following (53) in \cite{FGMouPaul}; the equivalence between these two definitions is the formula before (54) in \cite{FGMouPaul}.)
Now, if $\phi\in\Ker(R)$, one has, by formula (54) of \cite{FGMouPaul}, 
$$
\la\phi|R|\phi\ra=\Tr(|\phi\ra\la\phi|R)=\int_{\bR^{2d}}\tilde W_\hb[|\phi\ra\la\phi|](q,p)m(q,p)dqdp=0\,.
$$
This implies that $\tilde W_\hb[|\phi\ra\la\phi|]=0$, which implies in turn that $\phi=0$. For this implication, see for instance Remark 2.3 in \cite{FGPaul}. Equivalently, using (46) and (54) in \cite{FGMouPaul} with $R=|\phi\ra\la\phi|$ shows that
$$
\|\phi\|^2_{\fH}=\Tr(|\phi\ra\la\phi|)=\int_{\bR^{2d}}\tilde W_\hb[|\phi\ra\la\phi|](q,p)dqdp=0\,.
$$

\smallskip
Let $(\fa,\fb)\in\tilde\fK(R,S)$ such that
$$
\Tr(R^{1/2}\fa R^{1/2}+S^{1/2}\fb S^{1/2})=MK_\hb(R,S)^2\,.
$$
Assume that $|q,p\ra\in\cJ[R]\cap\cJ[S]$ for each $(q,p)\in\bR^d\times\bR^d$, and that the functions
$$
(q,p)\mapsto\la q,p|\fa|q,p\ra\quad\text{ and }\quad(q,p)\mapsto\la q,p|\fb|q,p\ra
$$
are of class $C^2$ on $\bR^d\times\bR^d$. Define
$$
a(q,p):=\tilde W_\hb[\fa](q,p)\,,\quad b(q,p):=\tilde W_\hb[\fb](q,p)\,,
$$
and
$$
\tilde a(q,p):=\tfrac12(|q|^2+|p|^2-a(q,p)+d\hb)\,,\quad\tilde b(q,p):=\tfrac12(|q|^2+|p|^2-b(q,p)+d\hb)\,.
$$
Notice that $a\in L^1(\bR^{2d},\mu)$ and $b\in L^1(\bR^{2d},\nu)$ since $R^{1/2}\fa R^{1/2}$ and $S^{1/2}\fb S^{1/2}$ belong to $\cL^1(\fH)$ by definition of $\tilde\fK(R,S)$, because $\Ker(R)=\Ker(S)=\{0\}$.

\begin{Thm}
The following conditions are equivalent:

\noindent
(a) there exists $\ll\in\cP(\bR^d\times\bR^d\times\bR^d\times\bR^d)$ such that $F=\Op^T_\hb[(2\pi\hb)^{2d}\ll]$ is an optimal coupling of $R$ and $S$, i.e.
$$
F\in\cC(R,S)\quad\text{ and }\quad\Tr(F^{1/2}CF^{1/2})=MK_\hb(R,S)^2\,;
$$
(b) one has
$$
MK_\hb(R,S)^2=\MKd(\mu,\nu)^2+2d\hb\,;
$$
(c) the functions $\tilde a$ and $\tilde b$ are the Legendre transforms of each other, i.e. 
$$
\tilde a^*=\tilde b\quad\text{ and }\quad\tilde b^*=\tilde a\,,
$$
and satisfy the Monge-Amp\`ere equation
$$
\Det(\grad^2\tilde a)=\frac{m}{n\circ\grad\tilde a}\,.
$$

If these conditions are satisfied, $\ll$ is of the form
$$
\ll(dq_1dp_1dq_2dp_2)=m(q_1,p_1)\de((q_2,p_2)-\grad\tilde a(q_1,p_1))dq_1dp_1\,.
$$
\end{Thm}

\smallskip
Let us recall that, for two T\"oplitz densities $R$ and $S$ in $\cD_2(L^2(\bR^d))$ with symbols $(2\pi\hb)^d\mu$ and $(2\pi\hb)^d\nu$ resp., it has been proved in \cite{FGMouPaul} (Theorem 2.3 (1)) that 
$$
MK_\hb(R,S)^2\le\MKd(\mu,\nu)^2+2d\hb\,.
$$
We also recall the example constructed in section 3 of \cite{CagliotiFGPaul}, where $\mu$ and $\nu$ are convex combinations of two Dirac measures with identical supports, for which the inequality above is strict. This example explains the title 
of \cite{CagliotiFGPaul}: quantum optimal transport is ``cheaper'' that classical optimal transport, due to additional degrees of freedom in quantum couplings which have no classical interpretation: see the penultimate paragraph on pp. 161--162
in \cite{CagliotiFGPaul}.

At variance with the example in section 3 of \cite{CagliotiFGPaul}, the situation where the optimal coupling between two T\"oplitz densities $R,S$ is a T\"oplitz operator is the closest to the classical setting. The classical optimal transport map
between the symbols of $R$ and $S$ is transformed into an optimal quantum coupling by T\"oplitz quantization. There are no strictly quantum effects in this coupling, unlike in the case discussed in section 3 of \cite{CagliotiFGPaul}, so that, 
the inequality in Theorem 2.3 (1) of \cite{FGMouPaul} becomes an equality in this case. In other words, apart from the additional term $2d\hb$, the quantum distance between such T\"oplitz densities is indeed the classical Monge-Kantorovich 
distance between their symbols. Examples of T\"oplitz densities satisfying properties (a) and (b) of the theorem above can be found in section 2 of \cite{CagliotiFGPaul}. However, the example constructed in section 2 of \cite{CagliotiFGPaul}
does not fall exactly in the class of densities considered in the theorem above, since the symbols of the densities considered in section 2 of \cite{CagliotiFGPaul} are convex combinations of two Dirac measures (with different supports and
equal coefficients).

\begin{proof}
Assume that (a) holds. One has
\be\lb{TrFCF=}
\Tr(F^{1/2}CF^{1/2})=\int_{\bR^{4d}}\tilde W_\hb[C](q_1,p_1,q_2,p_2)\ll(dq_1dp_1dq_2dp_2)
\ee
and
$$
\tilde W_\hb[C](q_1,p_1,q_2,p_2)=c(q_1,p_1,q_2,p_2)+2d\hb
$$
with
$$
c(q_1,p_1,q_2,p_2)=|q_1-q_2|^2+|p_1-p_2|^2\,.
$$
Let us take \eqref{TrFCF=} for granted --- we shall give a quick proof of this formula at the end of the present section.

Since $F\in\cC(R,S)$, the symbol $(2\pi\hb)^{2d}\ll$ of $F$ satisfies $\ll\in\Pi(\mu,\nu)$, the set of couplings of $\mu$ and $\nu$, according to Lemma 4.1 in \cite{FGMouPaul}. Hence
$$
MK_\hb(R,S)^2=\int_{\bR^{4d}}c(q_1,p_1,q_2,p_2)\ll(dq_1dp_1dq_2dp_2)+2d\hb\ge \MKd(\mu,\nu)^2+2d\hb
$$
By Theorem 2.3 (1) of \cite{FGMouPaul}, one has
$$
MK_\hb(R,S)^2\le \MKd(\mu,\nu)^2+2d\hb\,,
$$
which proves (b).

Conversely, pick an optimal coupling $\ll\in\Pi(\mu,\nu)$ and set $F=\Op^T_\hb[(2\pi\hb)^{2d}\ll]$. Then $F\in\cC(R,S)$ by Lemma 4.1 in \cite{FGMouPaul}, and
$$
\MKd(\mu,\nu)^2\!+\!2d\hb\!=\!\int_{\bR^{4d}}c(q_1,p_1,q_2,p_2)\ll(dq_1dp_1dq_2dp_2)\!+\!2d\hb\!=\!\Tr(F^\frac12CF^\frac12)\,.
$$
Hence, (b) implies that
$$
MK_\hb(R,S)^2=\Tr(F^{1/2}CF^{1/2})\,,
$$
so that (a) holds.

If (a) holds, then
$$
\ba
\int_{\bR^{4d}}c(q_1,p_1,q_2,p_2)\ll(dq_1dp_1dq_2dp_2)+2d\hb
\\
=\Tr(F^{1/2}CF^{1/2})=\Tr(R^{1/2}\fa R^{1/2}+S^{1/2}\fb S^{1/2})
\\
=\int_{\bR^{2d}}a(q_1,p_1)\mu(dq_1dp_1)+\int_{\bR^{2d}}b(q_2,p_2)\nu(dq_2dp_2)&\,.
\ea
$$
On the other hand, since $(\fa,\fb)\in\tilde\fK(R,S)$, and since $|q,p\ra\in\cJ[R]\cap\cJ[S]$, one has
$$
\la q_1,q_2,p_1,p_2|C|q_1,q_2,p_1,p_2\ra\ge\la q_1,p_1|\fa|q_1,p_1\ra+\la q_2,p_2|\fb|q_2,p_2\ra
$$
i.e.
$$
c(q_1,p_1,q_2,p_2)\ge a(q_1,p_1)-d\hb+b(q_2,p_2)-d\hb\,.
$$
Since $a\in L^1(\bR^{2d},\mu)$ and $b\in L^1(\bR^{2d},\nu)$ and $\ll\in\Pi(\mu,\nu)$ with
$$
\ba
\int_{\bR^{4d}}c(q_1,p_1,q_2,p_2)\ll(dq_1dp_1dq_2dp_2)=&\int_{\bR^{2d}}(a(q_1,p_1)-d\hb)\mu(dq_1dp_1)
\\
&+\int_{\bR^{2d}}(b(q_2,p_2)-d\hb)\nu(dq_2dp_2)\,,
\ea
$$
we conclude from Theorem 1.3 in \cite{VillaniAMS} (the Kantorovich duality theorem) that $\ll$ is an optimal element of $\Pi(\mu,\nu)$, and that the optimal functions $\tilde a$ and $\tilde b$ are Legendre duals of each other (see Lemma 2.10
in \cite{VillaniAMS}).
 
By the Brenier theorem (Theorem 2.12 (ii) in \cite{VillaniAMS}, the measure $\ll$ is of the form
$$
\ll(dq_1dp_1dq_2dp_2)=m(q_1,p_1)\de((q_2,p_2)-\grad\Phi(q_1,p_1))dq_1dp_1\,,
$$
with $\Phi$ convex. Hence
$$
\tilde a(q_1,p_1)+\tilde b(\grad\Phi(q_1,p_1))=q_1\cdot\grad_q\Phi(q_1,p_1)+p_1\cdot\grad_p\Phi(q_1,p_1)\quad\text{ for a.e. }(q_1,p_1)\,.
$$
On the other hand, we know that
$$
\tilde a(z,\zeta)+\tilde b(\grad\Phi(q_1,p_1))\ge z\cdot\grad_q\Phi(q_1,p_1)+\zeta\cdot\grad_p\Phi(q_1,p_1)\quad\text{ for a.e. }(q_1,p_1,z,\zeta)\,.
$$
Therefore
$$
\tilde a(z,\zeta)-a(q_1,p_1)\ge(z-q_1)\cdot\grad_q\Phi(q_1,p_1)+(\zeta-p_1)\cdot\grad_p\Phi(q_1,p_1)\quad\text{ for a.e. }(q_1,p_1,z,\zeta)\,.
$$
Hence $\grad\Phi=\grad\tilde a$, and since $\grad\Phi\#\mu=\nu$, the change of variables formula implies that
$$
\Det(\grad^2\Phi)n\circ\grad\Phi=m\,.
$$
This proves (c).

Conversely, assume that (c) holds and set 
$$
\ll(dq_1dp_1dq_2dp_2)=m(q_1,p_1)\de((q_2,p_2)-\grad\tilde a(q_1,p_1))dq_1dp_1\,.
$$
Obviously, $\ll\in\Pi(\mu,\nu)$ because of the Monge-Amp\`ere equation satisfied by $\tilde a$, and Brenier's theorem implies that
$$
\int_{\bR^{4d}}c(q_1,p_1,q_2,p_2)\ll(dq_1dp_1dq_2dp_2)=\MKd(\mu,\nu)^2\,.
$$
Set $F=\Op^T_\hb[(2\pi\hb)^{2d}\ll]$; by Lemma 4.1 in \cite{FGMouPaul}, one has $F\in\cC(R,S)$. On the other hand, since $\tilde a^*=\tilde b$ and $\tilde b^*=\tilde a$, one has
$$
\tilde a(q_1,p_1)+\tilde b(q_2,p_2)=q_1\cdot q_2+p_1\cdot p_2\qquad\ll-\text{ a.e. in }(q_1,p_1,q_2,p_2)
$$
or equivalently
$$
\ba
\int_{\bR^{4d}}c(q_1,p_1,q_2,p_2)\ll(dq_1dp_1dq_2dp_2)+2d\hb=&\int_{\bR^{2d}}a(q_1,p_1)\mu(dq_1dp_1)
\\
&+\int_{\bR^{2d}}b(q_2,p_2)\nu(dq_2dp_2)\,.
\ea
$$
Since $a=\tilde W_{\hbar} [\mathfrak{a}]$ and $b=\tilde W_{\hbar}[\fb ]$, this equality can be recast as
$$
\Tr(F^{1/2}CF^{1/2})=\Tr(R^{1/2}\fa R^{1/2}+S^{1/2}\fb S^{1/2})=MK_\hb(R,S)^2\,,
$$
so that $F$ is an optimal element of $\cC(R,S)$, and (a) holds.
\end{proof}

\begin{proof}[Proof of \eqref{TrFCF=}]
Let $(e_j)_{j\ge 1}$ be a complete orthonormal system of eigenvectors of $F\in\cD_2(\fH\otimes\fH)$, so that
$$
F=\sum_{j\ge 1}\ell_j|e_j\ra\la e_j|\,,\qquad\text{ with }\sum_{j\ge 1}\ell_j=1\text{ and }\ell_j\ge 0\text{ for all }j\ge 1\,.
$$
On the other hand, by formula (48) in \cite{FGMouPaul}
$$
\Op^T_\hb[c]=C+\tfrac12\hb(\Dlt_{q_1,p_1,q_2,p_2}c)I_{\fH\otimes\fH}=C+4d\hb I\,,
$$
where we recall that 
$$
c(q_1,p_1,q_2,p_2)=|q_1-q_2|^2+|p_1-p_2|^2\,.
$$

By Tonelli's theorem, denoting $z_1:=(q_1,p_1)$ and $z_2:=(q_2,p_2)$, one has
$$
\ba
\Tr(F^{1/2}CF^{1/2})+4d\hb
\\
=\sum_{j,k\ge 1}\tfrac{\ell_j^{1/2}\ell_k^{1/2}}{(2\pi\hbar)^{2d}}\int_{\bR^{4d}}\la e_j|z_1,z_2\ra\la z_1,z_2|e_k\ra\la e_k|e_j\ra c(z_1,z_2)dq_1dp_1dq_2dp_2
\\
=\sum_{j\ge 1}\tfrac{\ell_j}{(2\pi\hbar)^{2d}}\int_{\bR^{4d}}\la e_j|z_1,z_2\ra\la z_1,z_2|e_j\ra c(z_1,z_2)dq_1dp_1dq_2dp_2
\\
=\int_{\bR^{4d}}\sum_{j\ge 1}\tfrac{\ell_j}{(2\pi\hbar)^{2d}}\la e_j|z_1,z_2\ra\la z_1,z_2|e_j\ra c(z_1,z_2)dq_1dp_1dq_2dp_2
\\
=\int_{\bR^{4d}}\tilde W[F](z_1,z_2)c(z_1,z_2)dq_1dp_1dq_2dp_2&\,.
\ea
$$
Since $F=\Op^T_\hb[(2\pi\hb)^{2d}\ll]$, one has $\tilde W_\hb[F]=e^{\frac12\hb\Dlt_{q_1,p_1,q_2,p_2}}\ll$ (see (51) and the formula following (53) in \cite{FGMouPaul}), so that
$$
\ba
\Tr(F^{1/2}CF^{1/2})+4d\hb=\int_{\bR^{4d}}\tilde W[F](z_1,z_2)c(z_1,z_2)dq_1dp_1dq_2dp_2
\\
=\int_{\bR^{4d}}e^{\frac12\hb\Dlt_{q_1,p_1,q_2,p_2}}c(q_1,p_1,q_2,p_2)\ll(dq_1dp_1dq_2dp_2)
\ea
$$
since $e^{\frac12\hb\Dlt_{q_1,p_1,q_2,p_2}}$ is self-adjoint. On the other hand
$$
e^{\frac12\hb\Dlt_{q_1,p_1,q_2,p_2}}c(q_1,p_1,q_2,p_2)=c(q_1,p_1,q_2,p_2)+4d\hb\,,
$$
so that
$$
\ba
\Tr(F^{1/2}CF^{1/2})+4d\hb=\int_{\bR^{4d}}(c(q_1,p_1,q_2,p_2)+4d\hb)\ll(dq_1dp_1dq_2dp_2)\,,
\ea
$$
which is equivalent to \eqref{TrFCF=}.

\end{proof}


\begin{appendix}


\section{The Quantum Transport Cost}\lb{S-QTCost}


The quantum cost is the differential operator
$$
C:=\sum_{j=1}^d((x_j-y_j)-\hbar^2(\d_{x_j}-\d_{y_j})^2)\,.
$$
For each $f\equiv f(x_1,y_1,\ldots,x_d,y_d)\in\cS(\bR^{2d})$, one has
$$
\ba
Cf(x_1,y_1,\ldots,x_d,y_d)
\\
=\sum_{j=1}^d(Y_j^2-4\hbar^2\d_{Y_j}^2)f(X_1+\tfrac12Y_1,X_1-\tfrac12Y_1,\ldots,X_d+\tfrac12Y_d,X_d-\tfrac12Y_d)&\,.
\ea
$$
The $d$ operators $Y_j^2-4\hbar^2\d_{Y_j}^2$ obviously commute pairwise. Since each one of these operators is the quantum Hamiltonian of a harmonic oscillator, we know that a complete orthonormal system of eigenfunctions for
$Y_j^2-4\hbar^2\d_{Y_j}^2$ on $L^2(\bR)$ is 
$$
(2\hbar)^{-1/4}h_n(Y_j/\sqrt{2\hbar})\,,\quad n\ge 0\,,
$$
where $h_n$ is the $n$-th Hermite function
$$
h_n(z):=\pi^{-1/4}(2^nn!)^{-1/2}e^{-z^2/2}H_n(z)\,,\quad H_n(z):=(-1)^ne^{z^2}(e^{-z^2})^{(n)}\,,
$$
$$
h_n(Y_j):=(\pi 2\hbar)^{-1/4}(2^nn!)^{-1/2}e^{-Y_j^2/4\hbar}H_n(Y_j/\sqrt{2\hbar})\,\quad n\ge 0\,,
$$
with
$$
(Y_j^2-4\hbar^2\d_{Y_j}^2)h_n(Y_j/\sqrt{2\hbar})=2\hbar(2n+1)h_n(Y_j/\sqrt{2\hbar})\,,\quad n\ge 0\,.
$$
Since the linear transformation of $\bR^{2d}$
$$
(X_1,Y_1,\ldots,X_d,Y_d)\mapsto(X_1+\tfrac12Y_1,X_1-\tfrac12Y_1,\ldots,X_d+\tfrac12Y_d,X_d-\tfrac12Y_d)
$$
has Jacobian determinant $(-1)^d$, it leaves the Lebesgue measure of $\bR^{2d}$ invariant, so that 
$$
\Psi_{m_1,\ldots,m_d,n_1,\ldots,n_d}(x_1,y_1,\ldots,x_d,y_d):=(2\hbar)^{-d/2}\prod_{j=1}^dh_{m_j}(\tfrac{x_j+y_j}{2\sqrt{2\hbar}})h_{n_j}(\tfrac{x_j-y_j}{\sqrt{2\hbar}})
$$
defines a complete orthonormal system of eigenfunctions of $C$, i.e.
$$
\ba
\int_{\bR^{2d}}\overline{\Psi_{m'_1,\ldots,m'_d,n'_1,\ldots,n'_d}}\Psi_{m_1,\ldots,m_d,n_1,\ldots,n_d}(x_1,y_1,\ldots,x_d,y_d)dx_1\ldots dy_d
\\
=\prod_{j=1}^d\de_{m'_j,m_j}\de_{n'_j,n_j}&\,,
\ea
$$
with
$$
\ba
C\Psi_{m_1,\ldots,m_d,n_1,\ldots,n_d}(x_1,y_1,\ldots,x_d,y_d)
\\
=2\hbar(2(n_1+\ldots+n_d)+d)\Psi_{m_1,\ldots,m_d,n_1,\ldots,n_d}(x_1,y_1,\ldots,x_d,y_d)&\,.
\ea
$$
Thus
$$
C=2\hb\sum_{\genfrac{}{}{0pt}{2}{m_1,\ldots,m_d\ge 0}{n_1,\ldots,n_d\ge 0}}(2(n_1+\ldots+n_d)+d)|\Psi_{m_1,\ldots,m_d,n_1,\ldots,n_d}\ra\la \Psi_{m_1,\ldots,m_d,n_1,\ldots,n_d}|\,;
$$
in other words, $C$ has the spectral decomposition
$$
E(d\ll)\!\!:=\!\!\sum_{\genfrac{}{}{0pt}{2}{m_1,\ldots,m_d\ge 0}{n_1,\ldots,n_d\ge 0}}\de(\ll-2\hb(2(n_1\!+\!\ldots\!+\!n_d)+d))|\Psi_{m_1,\ldots,m_d,n_1,\ldots,n_d}\ra\la \Psi_{m_1,\ldots,m_d,n_1,\ldots,n_d}|\,.
$$


\section{Monotone Convergence for Trace-Class Operators}\label{monconv}


{\tcb Here is an \tcee}analogue of the Beppo Levi monotone convergence theorem for operators in the form convenient for our purpose.

Let $\scrH$ be a separable Hilbert space and $0\le T=T^*\in\cL(\scrH)$. For each complete orthonormal system $(e_j)_{j\ge 1}$ of $\scrH$, set
$$
\Tr_\cH(T)=\|T\|_1:=\sum_{j\ge 1}\la e_j|T|e_j\ra\in[0,+\infty]\,.
$$
See Theorem 2.14 in \cite{Simon}; in particular the expression on the last right hand side of these equalities is independent of the complete orthonormal system $(e_j)_{j\ge 1}$. Then
$$
T\in\cL^1(\scrH)\iff\|T\|_1<\infty\,.
$$

\begin{Lem}[Monotone convergence]\lb{L-Monotone}
Consider a sequence $T_n=T_n^*\in\cL^1(\scrH)$ such that
$$
0\le T_1\le T_2\le\ldots\le T_n\le\ldots\,,\qquad\text{ and }\sup_{n\ge 1}\la x|T_n|x\ra<+\infty\text{ for all }x\in\scrH\,,\leqno{\text{(i)}}
$$
or
$$
0\le T_1\le T_2\le\ldots\le T_n\le\ldots\,,\qquad\text{ and }\sup_{n\ge 1}\Tr_\cH(T_n)<+\infty\,.\leqno{\text{(ii)}}
$$
Then

\noindent
(a) there exists $T=T^*\in\cL(\scrH)$ such that $T\ge 0$ and $T_n\to T$ weakly as $n\to\infty$, and

\noindent
(b) $\Tr_\scrH(T_n)\to\Tr_\scrH(T)$ as $n\to\infty$.
\end{Lem}

\begin{proof}
First we prove statements (a) and (b) under assumption (i). Since the sequence $\la x|T_n|x\ra\in[0,+\infty)$ is nondecreasing for each $x\in\scrH$, 
$$
\la x|T_n|x\ra\to\sup_{n\ge 1}\la x|T_n|x\ra=:q(x)\in[0,+\infty)\quad\text{ for all }x\in\scrH
$$
as $n\to\infty$. Hence
$$
\la x|T_n|y\ra=\la y|T_n|x\ra\to\tfrac14(q(x+y)-q(x-y)+iq(x-iy)-iq(x+iy))=:b(x,y)\in\bC
$$
as $n\to+\infty$. By construction, $b$ is a nonnegative sesquilinear form on $\scrH$. 

Consider, for each $k\ge 0$,
$$
F_k:=\{x\in\scrH\text{ s.t. }\la x|T_n|x\ra\le k\text{ for each }n\ge 1\}\,.
$$
The set $F_k$ is closed for each $k\ge 0$, being the intersection of the closed sets defined by the inequality $\la x|T_n|x\ra\le k$ as $n\ge 1$. Since the sequence $\la x|T_n|x\ra$ is bounded for each $x\in\scrH$, 
$$
\bigcup_{k\ge 0}F_k=\scrH\,.
$$
Applying Baire's theorem shows that there exists $N\ge 0$ such that $\mathring{F}_N\not=\varnothing$. In other words, there exists $r>0$ and $x_0\in\scrH$ such that
$$
|x-x_0|\le r\implies|\la x|T_n|x\ra|\le N\text{ for all }n\ge 1\,.
$$
By linearity and positivity of $T_n$, this implies
$$
|\la z|T_n|z\ra|\le \tfrac2r(M+N)\|z\|^2\text{ for all }n\ge 1\,,\quad\text{ with }M:=\sup_{n\ge 1}\la x_0|T_n|x_0\ra\,.
$$
In particular
$$
\sup_{|z|\le 1}q(z)\le\tfrac2r(M+N)\,,\quad\text{ so that }|b(x,y)|\le\frac2r(M+N)|\|x\|_\scrH\|y\|_\scrH
$$
for each $x,y\in\scrH$ by the Cauchy-Schwarz inequality. By the Riesz representation theorem, there exists $T\in\cL(\scrH)$ such that
$$
T=T^*\ge 0\,,\quad\text{ and }\quad b(x,y)=\la x|T|y\ra\,.
$$
This proves (a). Observe that $T\ge T_n$ for each $n\ge 1$, so that 
$$
\sup_{n\ge 1}\Tr_\scrH(T_n)\le\Tr_\scrH(T)\,.
$$
In particular
$$
\sup_{n\ge 1}\Tr_\scrH(T_n)=+\infty\implies\Tr_\scrH(T)=+\infty\,.
$$
Since the sequence $\Tr_\scrH(T_n)$ is nondecreasing,
$$
\Tr_\scrH(T_n)\to\sup_{n\ge 1}\Tr_\scrH(T_n)\quad\text{ as }n\to\infty\,.
$$
By the noncommutative variant of Fatou's lemma (Theorem 2.7 (d) in \cite{Simon}), 
$$
\sup_{n\ge 1}\Tr_\scrH(T_n)<\infty\implies T\in\cL^1(\scrH)\text{ and }\Tr_\scrH(T)\le\sup_{n\ge 1}\Tr_\scrH(T_n)\,.
$$
Since the opposite inequality is already known to hold, this proves (b).

\smallskip
Next we prove (a) and (b) under assumption (ii). Since any $x\in\scrH\setminus\{0\}$ can be normalized and completed into a complete orthonormal system of $\scrH$, one has
$$
\sup_{n\ge 1}\la x|T|x\ra\le\|x\|_\scrH^2\sup_{n\ge 1}\Tr_{\scrH}(T_n)<\infty\,.
$$
Thus, assumption (ii) implies (i), which implies in turn (a) and (b).
\end{proof}


\section{The Finite Energy Condition}


Let $A=A^*\ge 0$ be an unbounded self-adjoint operator on $\scrH$ with domain $\Dom(A)$, and let $E$ be its spectral decomposition.

Throughout this section, we assume that $T\in\cL^1(\scrH)$ satisfies $T=T^*\ge 0$, and let $(e_j)_{j\ge 1}$ be a complete orthonormal system of eigenvectors of $T$ with $Te_j=\tau_je_j$ and $\tau_j\in[0,+\infty)$ for each $j\ge 1$, 
such that
\be\lb{TrTA}
\sum_{j\ge 1}\tau_j\int_0^\infty \ll\la e_j|E(d\ll)|e_j\ra<\infty\,.
\ee

\begin{Lem}\lb{L-Energ}
Under the assumptions above,
\be\lb{T1/2AT1/2=}
T^{1/2}AT^{1/2}:=\sum_{j,k\ge 1}\tau_j^{1/2}\tau_k^{1/2}\left(\int_0^\infty\ll\la e_j|E(d\ll)|e_k\ra\right)|e_j\ra\la e_k|
\ee
satisfies $0\le T^{1/2}AT^{1/2}=(T^{1/2}AT^{1/2})^*\in\cL^1(\scrH)$ and 
\be\lb{TrT1/2AT1/2=}
\Tr_\scrH(T^{1/2}AT^{1/2})=\sum_{j\ge 1}\tau_j\int_0^\infty \ll\la e_j|E(d\ll)|e_j\ra\,.
\ee
In particular
\be\lb{J0infDom}
e_j\in\Ker(T)^\perp\implies e_j\in\fDom(A)\,.
\ee
\end{Lem}

\begin{proof}
For each Borel $\om\subset\bR$ and each $x,y\in\scrH$, one has
$$
|\la x|E(\om)|y\ra|=|\la E(\om)x|E(\om)y\ra|\le\|E(\om)x\|\|E(\om)y\|=\la x|E(\om)|x\ra^{1/2}\la y|E(\om)|y\ra^{1/2}
$$
since $E(\om)$ is a self-adjoint projection. In particular, for each $\a>0$, one has
$$
2|\la x|E(\om)|y\ra|\le\a\la x|E(\om)|x\ra+\tfrac1\a\la y|E(\om)|y\ra\,.
$$
Hence, for all $j,k\ge 1$
$$
a_{jk}:=\int_0^\infty\ll\la e_j|E(d\ll)|e_k\ra\in\bC
$$
satisfies 
$$
2|a_{jk}|^2\le\a a_{jj}+\tfrac1\a a_{kk}\text{ for all }\a>0\,,\text{ so that }|a_{jk}|^2\le a_{jj}a_{kk}\,.
$$
Since $(\tau_ja_{jj})_{j\ge 1}\in\ell^1(\bN^*)$ by \eqref{TrTA} and since 
$$
\la e_j|T^{1/2}AT^{1/2}|e_k\ra=\tau_j^{1/2}\tau_k^{1/2}a_{jk}=\overline{\la e_k|T^{1/2}AT^{1/2}|e_j\ra}\,,
$$
one concludes that $T^{1/2}AT^{1/2}=(T^{1/2}AT^{1/2})^*\in\cL^2(\scrH)$. Moreover, for each $x\in\scrH$
$$
\ba
\la x|T^{1/2}AT^{1/2}|x\ra=&\sum_{j,k\ge 1}\tau_j^{1/2}\tau_k^{1/2}\overline{\la e_j|x\ra}\la e_k|x\ra\int_0^\infty\ll\la e_j|E(d\ll)|e_k\ra
\\
\ge&\int_0^\infty\ll\La\sum_{j\ge 1}\tau_j^{1/2}\la e_j|x\ra e_j|E(d\ll)| \sum_{j\ge 1}\tau_j^{1/2}\la e_j|x\ra e_j\Ra
\\
=&\int_0^\infty\ll\la T^{1/2}x|E(d\ll)| T^{1/2}x \ra\ge 0\,,
\ea
$$
so that $T^{1/2}AT^{1/2}\ge 0$. Finally
$$
\ba
\sum_{l\ge 1}\la e_l|T^{1/2}AT^{1/2}|e_l\ra=\sum_{l\ge 1}\sum_{j,k\ge 1}\tau_j^{1/2}\tau_k^{1/2}\left(\int_0^\infty\ll\la e_j|E(d\ll)|e_k\ra\right)\la e_l|e_j\ra\la e_k|e_l\ra
\\
=\sum_{l\ge 1}\sum_{j,k\ge 1}\tau_j^{1/2}\tau_k^{1/2}\left(\int_0^\infty\ll\la e_j|E(d\ll)|e_k\ra\right)\de_{lj}\de_{lk}=\sum_{l\ge 1}\tau_l\int_0^\infty\ll\la e_l|E(d\ll)|e_l\ra&<\infty
\ea
$$
so that $T^{1/2}AT^{1/2}\in\cL^1(\scrH)$, with
$$
\|T^{1/2}AT^{1/2}\|_1=\Tr_{\scrH}(T^{1/2}AT^{1/2})=\sum_{l\ge 1}\tau_l\int_0^\infty\ll\la e_l|E(d\ll)|e_l\ra<\infty\,.
$$
In particular for each $j\ge 1$, one has
$$
\tau_j>0\implies\la e_j|A|e_j\ra:=\int_0^\infty \ll\la e_j|E(d\ll)|e_j\ra\le\tfrac1{\tau_j}\Tr_\scrH(T^{1/2}AT^{1/2})<\infty\,,
$$
and this proves \eqref{J0infDom}.
\end{proof}

\begin{Cor}\lb{C-Energ}
Let $\Phi_n\in C(\bR_+)$ satisfy
$$
0\le\Phi_1(r)\le\Phi_2(r)\le\ldots\le\Phi_n(r)\to r\quad\text{ as }n\to\infty\,.
$$
For each $n\ge 1$, set
$$
\Phi_n(A):=\int_0^\infty\Phi_n(\ll)E(d\ll)\in\cL(\scrH)\,,\text{ so that }\Phi_n(A)=\Phi_n(A)^*\ge 0\,.
$$
In the limit as $n\to\infty$, one has
$$
T^{\frac12}\Phi_n(A)T^{\frac12}\to T^{\frac12}AT^{\frac12}\text{ weakly, and }\Tr_\scrH(T\Phi_n(A))\to\Tr_\scrH(T^{\frac12}AT^{\frac12})\,.
$$
\end{Cor}

\begin{proof}
Since $E$ is a resolution of the identity on $[0,+\infty)$, and since $\Phi_n$ is continuous, bounded and with values in $[0,+\infty)$, the operators $\Phi_n(A)$ satisfy
$$
0\le\Phi_1(A)\le\Phi_2(A)\le\ldots\le\Phi_n(A)\le\Phi_n(A)^*\le\left(\sup_{z\ge 0}\Phi_n(z)\right)I_\scrH\,.
$$
Set $R_n:=T^{1/2}\Phi_n(A)T^{1/2}$; by definition, one has $R_n=R_n^*\in\cL(\fH)$ and 
$$
0\le R_1\le R_2\le\ldots\le R_n\le\ldots
$$
On account of \eqref{TrTA}, one has
$$
\Tr_\scrH(R_n)=\sum_{j\ge 1}\tau_j\int_0^\infty\Phi_n(\ll)\la e_j|E(d\ll)|e_j\ra\le\sum_{j\ge 1}\tau_j\int_0^\infty\ll\la e_j|E(d\ll)|e_j\ra<\infty\,.
$$
By Lemma \ref{L-Monotone}, one has $R_n\to R$ weakly, with $R\in\cL^1(\scrH)$ and $R=R^*\ge 0$. Finally
$$
T^{1/2}AT^{1/2}-R_n=\sum_{j,k\ge 1}\tau_j^{1/2}\tau_k^{1/2}\left(\int_0^\infty(\ll-\Phi_n(\ll))\la e_j|E(d\ll)|e_k\ra\right)|e_j\ra\la e_k|
$$
using the definition \eqref{T1/2AT1/2=} of $T^{1/2}AT^{1/2}$ given in Lemma \ref{L-Energ}, so that
$$
\ba
\la x|T^{1/2}AT^{1/2}-R_n|x\ra=&\int_0^\infty(\ll-\Phi_n(\ll))\La\sum_{j\ge 1}\tau_j^{1/2}\la e_j|x\ra e_j|E(d\ll)|\sum_{k\ge 1}\tau_k^{1/2}\la e_k|x\ra e_k\Ra
\\
=&\int_0^\infty(\ll-\Phi_n(\ll))\la T^{1/2}x|E(d\ll)|T^{1/2}x\ra\ge 0\,.
\ea
$$
Hence $T^{1/2}AT^{1/2}-R_n=(T^{1/2}AT^{1/2}-R_n)^*\in\cL^1(\scrH)$ with $T^{1/2}AT^{1/2}-R_n\ge 0$, and
$$
\ba
\|T^{1/2}AT^{1/2}-R_n\|_1=&\Tr_\scrH(T^{1/2}AT^{1/2}-R_n)
\\
=&\sum_{j\ge 1}\tau_j\int_0^\infty(\ll-\Phi_n(\ll))\la e_j|E(d\ll)|e_j\ra\to 0
\ea
$$
as $n\to\infty$ by monotone convergence. Therefore $R_n\to T^{1/2}AT^{1/2}$ in $\cL^1(\scrH)$, and
$$
\Tr_{\scrH}(T\Phi_n(A))=\Tr_{\scrH}(T^{1/2}\Phi_n(A)T^{1/2})\to\Tr_{\scrH}(T^{1/2}AT^{1/2})\quad\text{ as }n\to\infty\,.
$$
\end{proof}

\begin{Lem}\lb{L-UBdCyclicity}
Let $S\in\cL^1(\cH\otimes\wtilde\cH)$ satisfy the partial trace condition $\Tr_{\wtilde\scrH}(S)=T$. Then $0\le S^{1/2}(A\otimes I_{\wtilde\scrH})S^{1/2}=(S^{1/2}(A\otimes I_{\wtilde\scrH})S^{1/2})^*\in\cL^1(\scrH\otimes\wtilde\scrH)$ and one has
$$
\Tr_{\scrH\otimes\wtilde\scrH}(S^{1/2}(A\otimes I_{\wtilde\scrH})S^{1/2})=\Tr_{\scrH}(T^{1/2}AT^{1/2})\,.
$$
\end{Lem}

\begin{proof}
For all $n\ge 1$, set $A_n=\Phi_n(A)\in\cL(\fH)$, with $\Phi_n(r):=\frac{r}{1+r/n}$ for all $r\ge 0$. Thus
$$
A_n=A_n^*\quad\text{ and }\quad 0\le A_1\le A_2\le\ldots\le A_n\le\ldots
$$
Hence $S^{1/2}(A_n\otimes I_{\wtilde\scrH})S^{1/2}=(S^{1/2}(A_n\otimes I_{\wtilde\scrH})S^{1/2})^*$ for all $n\ge 1$, with
$$
0\le S^{1/2}(A_1\otimes I_{\wtilde\scrH})S^{1/2}\le S^{1/2}(A_2\otimes I_{\wtilde\scrH})S^{1/2}\le\ldots\le S^{1/2}(A_n\otimes I_{\wtilde\scrH})S^{1/2}\le\ldots
$$
The partial trace condition implies that
$$
\ba
\Tr_{\scrH\otimes\wtilde\scrH}(S^{1/2}(A_n\otimes I_{\wtilde\scrH})S^{1/2})=\Tr_{\scrH\otimes\wtilde\scrH}(S(A_n\otimes I_{\wtilde\scrH}))
\\
=\Tr_{\scrH}(TA_n)=\Tr_{\scrH}(T^{1/2}A_nT^{1/2})
\ea
$$
while
$$
\ba
\Tr_{\scrH\otimes\wtilde\scrH}(S^{1/2}(A_n\otimes I_{\wtilde\scrH})S^{1/2})&\to\Tr_{\scrH\otimes\wtilde\scrH}(S^{1/2}(A\otimes I_{\wtilde\scrH})S^{1/2})
\\
\Tr_{\scrH}(T^{1/2}A_nT^{1/2})&\to\Tr_{\scrH}(T^{1/2}AT^{1/2})
\ea
$$
as $n\to\infty$ by Lemma \ref{L-Monotone}. This implies the announced equality by uniqueness of the limit.
\end{proof}

\end{appendix}


\end{document}